\documentclass[11pt]{article}

\usepackage[latin1]{inputenc}

\usepackage{amstext,amsmath,amssymb,amsfonts,bbm,mathtools}
\usepackage{extarrows}
\usepackage{xcolor} 
\usepackage{hyperref}
\hypersetup{
	colorlinks=false,
	menubordercolor=red,
	linkbordercolor=blue
}
\usepackage{authblk}
\usepackage{caption}

\usepackage{epsfig} 

\usepackage{color} 
\usepackage{graphicx} 

\usepackage{braket} 
\usepackage{slashed} 
\usepackage{dsfont} 
\usepackage{amsthm}

\usepackage{cite}

\allowdisplaybreaks[4]

\usepackage{psfrag}
\usepackage{tikz} 
\usetikzlibrary{arrows}
\usetikzlibrary{plotmarks}
\usetikzlibrary{external}
\usetikzlibrary{patterns}
\usepackage{pgfplots}
\pgfplotsset{compat=1.9}
\usetikzlibrary{patterns}
\usetikzlibrary{calc}
\usetikzlibrary{automata,positioning}

\usepackage{float}  
\usepackage{subfig}
\usepackage[lmargin=50pt,rmargin=60pt,tmargin=60pt,bmargin=65pt]{geometry}

\newcommand{\be}{\begin{equation}}
	\newcommand{\ee}{\end{equation}} 

\newcommand{\f}{\frac}
\newcommand{\p}{\partial}

     \let\d=\delta  \let\ve=\varepsilon
     \let\th=\theta   \let\l=\lambda
              \let\r=\rho 
\let\s=\sigma \let\t=\tau     
    
\let\G=\Gamma     \let\X=F

\newcommand{\cC}{\mathcal{C}}

\newcommand{\cE}{\mathcal{E}}
\newcommand{\cF}{\mathcal{F}}

\newcommand{\cJ}{\mathcal{J}}
\newcommand{\cK}{\mathcal{K}}

\newcommand{\cT}{\mathcal{T}}

\definecolor{britishracinggreen}{rgb}{0.0, 0.26, 0.15}

\newcommand{\sgn}{{\rm sgn}}

\allowdisplaybreaks[4]

\numberwithin{equation}{section}	

\newtheorem{remark}{Remark}
\newtheorem{proposition}{Proposition}

\newtheorem{lemma}{Lemma}
\newtheorem{theorem}{Theorem}

\theoremstyle{remark}

\begin{document}

	\title{\bf\boldmath
   The small-$N$ series in the zero-dimensional  $O(N)$ model: constructive expansions and transseries
   }
	
	\author[1]{Dario Benedetti}
	\author[1,2]{Razvan Gurau}
	\author[2]{Hannes Keppler}
	\author[2]{Davide Lettera}

	\affil[1]{\normalsize \it 
		CPHT, CNRS, Ecole Polytechnique, Institut Polytechnique de Paris, Route de Saclay, \authorcr 91128 PALAISEAU, 
		France
		\authorcr \hfill}
	
	\affil[2]{\normalsize\it 
		Heidelberg University, Institut f\"ur Theoretische Physik, Philosophenweg 19, 69120 Heidelberg, Germany
		\authorcr \hfill
		\authorcr
		emails: dario.benedetti@polytechnique.edu, gurau@thphys.uni-heidelberg.de, keppler@thphys.uni-heidelberg.de, lettera@thphys.uni-heidelberg.de 
		\authorcr \hfill }
	
	\date{}
	
	\maketitle
	
	\hrule\bigskip
	
	\begin{abstract}
	We consider the 0-dimensional quartic $O(N)$ vector model and present a complete study of the partition function $Z(g,N)$ and its logarithm, the free energy $W(g,N)$, seen as functions of the coupling $g$ on a Riemann surface. Using constructive field theory techniques we prove that both $Z(g,N)$ and $W(g,N)$ are Borel summable functions along all the rays in the cut complex plane $\mathbb{C}_{\pi} =\mathbb{C}\setminus \mathbb{R}_-$. We recover the transseries expansion of $Z(g,N)$ using the intermediate field representation.
	
    We furthermore study the small-$N$ expansions of $Z(g,N)$ and $ W(g,N)$. For any $g=|g| e^{\imath \varphi}$ on the sector of the Riemann surface with $|\varphi|<3\pi/2$, the small-$N$ expansion of $Z(g,N)$ has infinite radius of convergence in $N$ while the expansion of $W(g,N)$ has a finite radius of convergence in $N$ for $g$ in a subdomain of the same sector. 
    
    The Taylor coefficients of these expansions, $Z_n(g)$ and $W_n(g)$, exhibit analytic properties similar to $Z(g,N)$ and $W(g,N)$ and have transseries expansions. The transseries expansion of $Z_n(g)$ is readily accessible:  much like $Z(g,N)$, for any $n$, $Z_n(g)$ has a zero- and a one-instanton contribution. The transseries of $W_n(g)$ is obtained using M\"oebius inversion and summing these transseries yields the transseries expansion of $W(g,N)$. The transseries of $W_n(g)$ and $W(g,N)$ are markedly different: while $W(g,N)$ displays contributions from arbitrarily many multi-instantons, $W_n(g)$ exhibits contributions of only up to $n$-instanton sectors. 

	\end{abstract}
	
	\hrule\bigskip
	\tableofcontents
	
\section{Introduction}
	
The most famous problem of the perturbative expansion in quantum field theory is the existence of ultraviolet divergences in the amplitudes of Feynman diagrams. This is successfully dealt with using the theory of perturbative renormalization. 
However, even in one and zero dimensions (quantum mechanics and combinatorial models, respectively), where renormalization is not needed, perturbation theory poses another notorious challenge: in most cases the perturbative series is only an asymptotic series, with zero radius of convergence. Borel resummation is the standard strategy to address this problem, but this comes with its own subtleties. From a practical standpoint, we are often only able to compute just the first few terms in the perturbative expansion. At a more fundamental level, singularities are present in the Borel plane, associated to instantons (and renormalons in higher dimensions). The instanton singularities are not accidental: they stem from the factorial growth of the number of Feynamn diagrams  with the perturbation order, which is also the origin of the divergence of the perturbation series.\footnote{The renormalon singularities specific to higher dimensions, are different. They are located on the positive real axis and stem from the factorial growth of the renormalized amplitude of a family of diagrams consisting in essentially one diagram per perturbation order.}

From the resummation point of view, the most inconvenient feature of perturbation theory is that it does not naively capture contributions from non-analytical terms. For example, it is well known that instanton contributions of the type $e^{-1/g}$ ($g>0$ being the coupling constant) can be present in the evaluation of some quantity of interest, but they are missed in the perturbative series as their Taylor expansion at $g=0$ vanishes identically.

Such exponentially suppressed terms are the archetypal example of nonperturbative effects, and their evaluation poses an interesting challenge.
Aiming to include them, but still relying for practical reasons on perturbative methods, one ends up with a more general form of asymptotic expansion, known as \emph{transseries}, which is roughly speaking a sum of perturbative and nonperturbative sectors, for example:
\begin{equation}
    F(g) \simeq \sum_{n\geq 0} a_n\, g^n 
    + \sum_i e^{\f{c_i}{g}} g^{\gamma_i} \sum_{n\geq 0} b_{i,n}\, g^{n} \;.
\end{equation}

Over time it became increasingly clear that, in many examples of interest, using the theory of Borel summation for the perturbative sector it is possible to reconstruct some information about the nonperturbative ones. This relation between the perturbative and nonperturbative sectors is known as \emph{resurgence}, and it was originally developed by \'{E}calle in the context of ordinary differential equations \cite{Ecalle} (see \cite{sauzin2007resurgent} for a modern review) and later extensively used in quantum field theory: for recent reviews with a quantum field theory scope, see \cite{Marino:2012zq,Dorigoni:2014hea,Dunne:2015eaa}, and in particular \cite{Aniceto:2018bis}, which contains also a comprehensive list of references to applications and other reviews.

Zero-dimensional quantum field theoretical models, which are purely combinatorial models,\footnote{These are for example of interest in the context of random geometry, see for example \cite{DiVecchia:1990ce,DiFrancesco:1993cyw,RTM}.} are useful toy models for the study of transseries expansions. Most conveniently, they allow one to set aside all the complications arising from the evaluation and renormalization of Feynman diagrams. Moreover, their partition functions and correlations typically satisfy ordinary differential equations, thus fitting naturally in the framework of \'{E}calle's theory of resurgence.
The zero dimensional $\phi^4$, or more generally $\phi^{2k}$ with $k\geq 2$, models in zero dimensions have been exhaustively studied \cite{Aniceto:2018bis,Fauvet:2019vcr}. At the opposite end, the rigorous study of the Borel summability in fully fledged quantum field theory is the object of constructive field theory \cite{jaffe2000constructive,rivasseau2014perturbative,Rivasseau:2009-0d}. It should come as no surprise that the generalization of results on resurgence in zero dimensions to fully fledged quantum field theory is very much an open topic. Not only one has to deal with renormalization, but also in higher dimensions the ordinary differential equations obeyed by the correlations become partial differential (Schwinger-Dyson) equations on tempered distributions, and one can not simply invoke \'{E}calle's theory.

From this perspective, revisiting the resurgence in zero dimensional models using techniques inspired by constructive field theory can be of great use. 
Here, following such route, we consider the zero-dimensional $O(N)$ model with quartic potential.\footnote{The same model has been considered in a similar context in \cite{Tanizaki:2014tua}, where the problem of constructing Lefschetz thimbles in the $N$-dimensional space have been studied. By using the intermediate field formalism we will bypass such problem here.}
Denoting $\phi = (\phi_a)_{a=1,\dots N}\in \mathbb{R}^N$ a vector in $\mathbb{R}^N$ and $\phi^2=\sum_{a=1}^N \phi_a\phi_a$ the $O(N)$ invariant, the partition function of the model is:\footnote{
Note that we do not use the usual normalization $g/4$ of the interaction in the $O(N)$ model, but stick to $g/4!$ in order to facilitate the comparison with the literature on transseries which deals mostly with the $N=1$ case for which the normalization $g/4!$ is standard. Also, we do not use the 't Hooft coupling $\l = g N$, which is needed for a well defined $1/N$ expansion: in this paper we keep $N$ small.
}
\begin{equation} \label{eq:Z_Phi4vec}
Z(g,N)=\int_{-\infty}^{+\infty} \left( \prod_{a=1}^N\f{d\phi_a}{\sqrt{2\pi}} \right) \; e^{-S[\phi] } \;, \qquad S[\phi]=\frac{1}{2}\phi^2+\frac{g}{4!} (\phi^2)^2 \; .
\end{equation}

The $N=1$ case has been extensively studied in \cite{Aniceto:2018bis}. One can analytically continue $Z(g,1)$, regarded as a function of the coupling constant $g$, to a maximal domain 
in the complex plane. Subsequently, one discovers that $Z(g,1)$ displays a branch cut at the real negative axis and that the nonperturbative contributions to $Z(g,1)$ are captured by its discontinuity at the branch cut. A resurgent transseries is obtained when one considers $g$ as a point on a Rienamm surface with a branch point at $g=0$. From now on we parameterize this Riemann surface as $g = |g| e^{\imath \varphi}$ and we choose as principal sheet $\varphi \in (-\pi,\pi)$.

An approach to the study of the partition function in Eq.~\eqref{eq:Z_Phi4vec} in the case $N=1$ is to use the steepest-descent method \cite{Bender:1999,Witten:2010cx}. We concisely review this in Appendix~\ref{app:expansions}. One notes that on the principal sheet 
only one Lefschetz thimble contributes. As $g$ sweeps through the princial sheet the thimble is smoothly deformed, but not in the neighborhood of the saddle point: the asymptotic evaluation of the integral is unchanged. 
When $g$ reaches the negative real axis there is a discontinuous jump in the relevant thimbles and a pair of thimbles (passing through a pair of conjugated non trivial saddle points of the action) starts contributing, giving rise to a one-instanton sector in the transseries of $Z(g,1)$. 

Another approach to the transseries expansion of $Z(g,1)$ is to use the theory of ordinary differential equations \cite{Aniceto:2018bis,Fauvet:2019vcr}. It turns out that $Z(g,1)$ obeys a second-order homogenous linear ordinary differential equation for which $g=0$ is an irregular singular point (e.g.\ \cite{Bender:1999}), giving another perspective on why the expansion one obtains is only asymptotic. 

More interestingly, one can wonder what can be said about the nonperturbative contributions to the free energy, that is, the logarithm of the partition function $W(g,1)=\ln Z(g,1)$, or to the connected correlation functions. If we aim to  study the free energy, the steepest-descent method does not generalize straightforwardly as we lack a simple integral representation for $W(g,1)$. 
One can formally write $Z(g,1)$ as a transseries and then expand the logarithm in powers of the transseries monomial $e^{\frac{c}{g}}$, thus obtaining a multi-instanton transseries. However this is very formal, as the transseries is only an asymptotic expansion, and we would like to have a direct way to obtain the asymptotic expansion of $W(g,1)$.
The closest one can get to an integral formula for the free energy is to use the Loop Vertex Expansion (LVE) \cite{Rivasseau:2007fr}, a constructive field theory expansion which we present in detail below. However, this is by no means simple, and deriving directly the transseries expansion of $W(g,1)$ using the steepest-descent method on this formula proved so far impractical.

The best method available for the study of the transseries expansion of $W(g,1)$, before our work, is to use again the theory of ordinary differential equations. One can show that $W(g,1)$ obeys a non-linear ordinary differential equation whose transseries solution can be studied \cite{Aniceto:2018bis}.

\bigskip

In this paper, we consider a general $N$ and we revisit both the partition function $Z(g,N)$ and the free energy $W(g,N)$ from a different angle. 
The paper is organized as follows.

In Section~\ref{sec:BS}, we review the Borel summability of asymptotic series as well as the notion of Borel summable functions, deriving in the process a slight extension of the Nevanlinna-Sokal theorem. 

In Section~\ref{sec:Z}, we study $Z(g,N)$ in the intermediate field representation.
This allows us to quickly prove its Borel summability along all the rays in the cut complex plane $\mathbb{C}_{\pi} = \mathbb{C}\setminus \mathbb{R}_-$. More importantly, the intermediate field representation provides a new perspective on the origin of the instanton contributions: in this representation, the steepest-descent contour never changes, but when $g$ reaches the negative real axis a singularity traverses it and detaches a Hankel contour around a cut. We insist that this Hankel contour is \emph{not} a steepest-descent contour, but it \emph{does} contribute to the asymptotic evaluation of the integral, because the cut is an obstruction when deforming the contour of integration towards the steepest-descent path.
It is precisely the Hankel contour that yields the one-instanton contribution. We then build the analytic continuation of $Z(g,N)$ to the whole Riemann surface, identify a second Stokes line, compute the Stokes data encoding the jumps in the analytic continuation at the Stokes lines and discuss the monodromy of $Z(g,N)$.
Next we observe that, because in the intermediate field representation $N$ appears only as a parameter in the action, we can perform a small-$N$ expansion:
\begin{equation}
     Z(g,N) = \sum_{n\ge 0} \frac{1}{n!} \left( -\frac{N}{2} \right)^n Z_n(g)  \; .
\end{equation}
We thus study $Z_n(g)$ for all integer $n$, proving its Borel summability in $\mathbb{C}_\pi$ and computing its transseries expansion in an extended sector of the Riemann surface, with $\arg(g)\in(-3\pi/2,3\pi/2)$, which we denote $\mathbb{C}_{3\pi/2}$.

In Section~\ref{sec:W}, we proceed to study $W(g,N)=\ln(Z(g,N))$. We first establish its Borel summability along all the rays in $\mathbb{C}_{\pi}$ using consturctive field theory techniques. We then proceed to the small-$N$ expansion of this object:
\begin{equation}
W(g,N) =  \sum_{n\ge 1}  \frac{1}{n!} \left( -\frac{N}{2} \right)^n   W_n(g) \;,
\end{equation}
and prove that this is an absolutely convergent series in a subdomain of $\mathbb{C}_{3\pi/2}$ and that both $W(g,N)$ and
$W_n(g)$ are Borel summable along all the rays in $\mathbb{C}_\pi$.
Finally, in order to obtain the transseries expansion of $W_n(g)$ and $W(g,N)$ we note that $W_n(g)$ can be written in terms of $Z_n(g)$ using the M\"oebius inversion formula relating moments and cumulants. 
Because of the absolute convergence of the small-$N$ series, it makes sense to perform the asymptoitic expansion term by term, and thus we rigorously obtain the transseries for $W(g,N)$ in a subdomain of $\mathbb{C}_{3\pi/2}$.
In the Appendices we gather some technical results, and the proofs of our propositions.

Ultimately, we obtain less information on the Stokes data for $W(g,N)$ than for $Z(g,N)$. While for $Z(g,N)$ we are able to maintain analytic control in the whole Riemann surface of $g$, the constructive field theory techniques we employ here allow us to keep control over $W(g,N)$ as an analytic function on the Riemann surface only up to $\varphi = \pm 3\pi/2$, that is past the first Stokes line, but not up to the second one. The reason for this is that close to $\varphi = \pm 3\pi/2$ there is an accumulation of Lee-Yang zeros, that is zeros of $Z(g,N)$, which make the explicit analytic continuation of $W(g,N)$ past this sector highly non trivial. New techniques are needed if one aims to recover the Stokes data for $W(g,N)$ farther on the Riemann surface: an analysis of the differential equation obeyed by $W(g,N)$ similar to the one of \cite{Baldino:2022aqm} could provide an alternative way to access it directly.

One can naturally ask what is the interplay between our results at small $N$ and the large-$N$ nonperturbative effects, first studied for the zero-dimensional $O(N)$ model in \cite{HikamiBrezin1979} (see also \cite{Marino:2012zq} for a general review, and  \cite{MarinoSerone2021JHEP} for a more recent point view).
This is a very interesting question: indeed, the relation between the two expansions is a bit more subtle than the relation between the small coupling and the large coupling expansions for instance.
The reason is that, when building the large $N$ series, one needs to use the 't Hooft coupling, which is a rescaling of the coupling constant by a factor of $N$. This changes the $N$-dependence of the partition function and free energy, making the relation between small-$N$ and large-$N$ expansions nontrivial. A good news on this front is that the analiticity domains in $g$ becomes uniform in N when recast in terms of the 't Hooft coupling \cite{Ferdinand:2022duk}. But there is still quite some work to do in order to connect the  transseries analysis at small $N$ with that at large $N$.

\paragraph{Main results.} Our main results are the following:
\begin{itemize}
    \item In Proposition~\ref{prop:Z}, we study $Z(g,N)$. While most of the results in this proposition are known for $N=1$, we recover them using the intermediate field representation (which provides a new point of view) and generalize them to arbitrary $N\in \mathbb{R}$. In particular we uncover an interplay between $Z(g,N)$ and $Z(-g,2-N)$ which captures the resurgent nature of the transseries expansion of the partition function for general $N$.

    \item Proposition~\ref{prop:Zn} deals with the function $Z_n(g)$, notably its Borel summability, transseries, and associated differential equation. To our knowledge, $Z_n(g)$ has not been studied before and all of the results presented here are new. 

   \item Proposition~\ref{prop:LVE+bounds} and~\ref{propW:BS} generalize previous results in the literature 
   \cite{Gurau:2014vwa} on the analyticity and Borel summability of $W(g,1)$ to $W(g,N)$ and furthermore derive parallel results for $W_n(g)$.

   \item Proposition \ref{prop:TSW} contains the transseries expansion of $W_n(g)$, which has not been previously considered in the literature. We also give a closed formula for the transseries expansion of $W(g,N)$.
   
   \item Lastly, in Proposition \ref{prop:DEW}, we derive the tower of recursive differential equations obeyed by $W_n(g)$. This serves as an invitation for future studies of the transseries of $W_n(g)$ from an ordinary differential equations perspective.

\end{itemize}

\section{Borel summable series and Borel summable functions}
\label{sec:BS}

When dealing with asymptotic series, a crucial notion is that of Borel summability. Less known, there exists a notion of Borel summability of \emph{functions}, intimately related to the Borel summability of series. In this section we present a brief review of these notions, which will play a central role in the rest of the paper, as well as a slight generalization of the (optimal) Nevanlinna-Sokal theorem on Borel summability \cite{Sokal:1980ey}. We will repeatedly use this theorem in this paper. 

\paragraph{Notation.} We use $K$ as a dustbin notation for irrelevant (real positive) multiplicative constants, and $R$ and $\rho$ for the important (real positive) constants. 

\paragraph{Borel summable series.}
A formal power series $A(z)=\sum_{k=0}^\infty a_k z^k$  is called a \emph{Borel summable series along the positive real axis} if the series
\begin{equation}
B (t) = \sum_{k=0}^\infty  \frac{a_k}{k!} t^k    \;,
\end{equation}
is absolutely convergent in some disk $|t|< \rho$ and $B(t)$ admits an analytic continuation in a strip of width $\rho$ along the positive real axis such that for $t$ in this strip $|B(t)| < K e^{|t|/R}$ for some real positive $R$. The function  $B(t) $ is called the \emph{Borel transform} of $A(z)$ and the \emph{Borel sum} of $A(z)$ is the Laplace transform of its Borel transform:
\begin{equation}
f(z) =  \frac{1}{z}  \int_0^{\infty} dt \; e^{-t/z} \; B (t) \; .
\end{equation} 
It is easy to check that the function $f$ is analytic in a disk of diameter $R$ tangent to the imaginary axis at the origin, ${\rm Disk}_R  = \{z\in \mathbb{C} \mid {\rm Re }{(1/z)}>1/R  \}$ . 

Clearly, if it exists, the Borel sum of a series is unique. This raises the following question: given a function $h(z)$ whose asymptotic series at zero is the Borel summable series $A(z)$, does the Borel resummation of $A(z)$ reconstruct $h(z)$? That is, is $f(z) = h(z)$? The answer to this question is \emph{no} in general: for instance the function $e^{-1/z}$ is asymptotic (along the positive real axis) at 0 to the Borel summable series $a_k=0$. It turns out that one can formulate necessary and sufficient conditions for $h(z)$ which ensure that it is indeed the Borel sum of its asymptotic series, as we now recall.

\paragraph{Borel summable functions.} A function $f:\mathbb{C} \to \mathbb{C}$ is called a \emph{Borel summable function along the positive real axis} if it is analytic in a disk ${\rm Disk}_R$ and has an asymptotic series at $0$ (which can have zero radius of convergence),
\begin{equation}
f(z) =  \sum_{k=0}^{q-1} a_k \; z^k + R_q(z)      \; ,
\end{equation}
such that the rest term of order $q$ obeys the bound:
\begin{equation}\label{eq:NS-bound}
 | R_q(z)   |   \le K   \; q! \; q^\beta \;  \rho^{- q}  \; |z|^q \;, \qquad z \in {\rm Disk}_R \;,    
\end{equation}
for some fixed $\beta\in \mathbb{R}_+$. Note that the bound in Eq.~\eqref{eq:NS-bound} is slightly weaker that the one in \cite{Sokal:1980ey}. The positive real axis is selected by the position of the center of ${\rm Disk}_R$. We call ${\rm Disk}_R  = \{z\in \mathbb{C} \mid {\rm Re }{(1/z)}>1/R  \}$ a \emph{Sokal disk}.

These two notions are intimately related: the Borel sums of Borel summable series are Borel summable functions (this is straightforward to prove). Moreover, the asymptotic series of Borel summable functions are Borel summable series.

\begin{theorem}[Nevanlinna-Sokal \cite{Sokal:1980ey}, extended]\label{thm:Sokal}
Let $f:\mathbb{C} \to \mathbb{C}$  be a Borel summable function, hence analytic and obeying the bound 
 \eqref{eq:NS-bound} with some fixed $\beta$. Then:
\begin{itemize}
    \item the Borel transform of the asymptotic series of $f$, 
\begin{equation}
B(t) = \sum_{k=0}^{\infty} \frac{1}{k! } \; a_k \; t^k \;,     
\end{equation}
is convergent in a disk of radius $\rho$ in $t$ and it defines an analytic funciton in this domain.
\item $B(t)$ can be analytically continued to the strip $ \{ t\in \mathbb{C} \mid  {\rm dist}(t,\mathbb{R}_+ ) < \rho \}$ and in this strip it obeys an exponential bound 
    $ |B(t)|<K e^{|t|/R}$. 
\item for all $z\in {\rm Disk}_R$ we can reconstruct the function $f(z)$ by the absolutely convergent Laplace transform:
\begin{equation}
f(z)=\frac{1}{z }\int_{0}^{\infty} dt \; e^{-t/z}  \; B(t)  \;  .    
\end{equation}
\end{itemize}
\end{theorem}

\begin{proof}
See Appendix~\ref{app:Sokal}.
\end{proof}

We emphasize that both for series and for functions, Borel summability is directional:
\begin{itemize}
    \item  for series, Borel summability along a direction requires the unimpeded analytic continuation of $B(t)$ in a thin strip centered on that direction.
    \item for functions, Borel summability along a direction requires analyticity and bound on the Taylor rest terms in a Sokal disk (with $0$ on its boundary) centered on that direction.
\end{itemize}

   Clearly, the singularities of the Borel transform $B(t)$ are associated to directions along which the function $f(z)$ ceases to be Borel summable. 

\section{The partition function $Z(g,N)$}
\label{sec:Z}

In this section we collect some facts about the asymptotic expansion of the partition function \eqref{eq:Z_Phi4vec}. Most of them are known, or derivable from the expression of the $Z(g,N)$ in terms of special functions, but to the best of our knowledge the expression for the transseries at general $N$ is new.

We study $Z(g,N)$ by means of the Hubbard-Stratonovich intermediate field formulation \cite{Hubbard,Stratonovich}, which is crucial to the Loop Vertex Expansion \cite{Rivasseau:2007fr} of the free energy $W(g,N)$ that we will study below.
This is based on rewriting the quartic term of the action as a Gaussian integral over an auxiliary variable $\s$ (or field, in higher dimensions):
\be
e^{-\f{g}{4!} (\phi^2)^2} = \int_{-\infty}^{+\infty} [d\sigma]\ e^{-\f12 \sigma^2 + \imath \sqrt{\f{g}{12}} \sigma \phi^2} \;,
\ee
where the Gaussian measure over $\sigma$ is normalized, i.e.\ $[d\sigma]=d\sigma/\sqrt{2\pi}$ and 
$\imath = e^{\imath \frac{\pi}{2}}$. 
Note that $\sigma$ is a real number, not a vector.
With this trick, the integral over $\phi$ becomes Gaussian and can be performed for $g>0$, leading to a rewriting of the partition function \eqref{eq:Z_Phi4vec} as:
\begin{equation}\label{eq:Zsigma}
Z(g,N)=\int_{-\infty}^{+\infty} [d\sigma]  \;e^{-\frac{1}{2} \sigma^2} \frac{1}{\left( 1 - \imath \sqrt{ \frac{g}{3} } \sigma \right)^{ N/2} }   \;.
\end{equation}

Although the original partition function is defined only for integer $N$ and we have assumed that $g>0$, in the $\s$ representation \eqref{eq:Zsigma} it becomes transparent that $Z(g,N)$ can be analytically continued both in $N$ and in $g$.

\subsection{Analytic continuation and transseries} 

As a matter of notation, we denote $\varphi\equiv \arg(g)$, and, in order to label some sets that will appear repeatedly in the rest of the paper, we define:
\begin{equation}
    \mathbb{C}_{\psi} \equiv \left\{ g\in\mathbb{C},\; g=|g| e^{\imath \varphi} :\;  \varphi \in \left(-\psi,\psi\right) \right\} \;.
\end{equation}
In particular, $\mathbb{C}_\pi=\mathbb{C} \setminus \mathbb{R}_-$ is the cut complex plane. For $\psi>\pi$, the set $ \mathbb{C}_{\psi}$ should be interpreted as a sector of a Riemann sheet, extending the principal sheet $\mathbb{C}_\pi$ into the next sheets.

Our first aim is to understand the analytic continuation of the partition function in the maximal possible domain of the Riemann surface. For later convenience, we introduce the following function, not to be confused with the partition function $Z(g,N)$:
\begin{equation}\label{eq:defZR}
    Z^{\mathbb{R}}(g,N) = \int_{-\infty}^{+\infty} [d\sigma]  \;e^{-\frac{1}{2} \sigma^2} \frac{1}{\left( 1 -  \imath  \, e^{\imath \frac{\varphi}{2}} \sqrt{ \frac{|g|}{3} } \sigma \right)^{ N/2} } \;,\qquad
    \forall \varphi \neq (2k+1)\pi \, ,\; k \in \mathbb{Z}  \;,
\end{equation}
which is an absolutely convergent integral for any $|g|$ and any $\varphi \neq (2k+1)\pi $. We have used a superscript $\mathbb{R}$ to distinguish it from $Z(g,N)$ and to emphasize that the integral is performed on the real line, irrespective of the value of $\varphi \neq (2k+1)\pi$.
Moreover, for $N/2\notin \mathbb{Z}$, the integrand is computed using the principal branch of the logarithm:
\begin{equation}
    \mathrm{ln} \left(  1 -  e^{\imath \frac{\pi + \varphi}{2}} \sqrt{ \tfrac{|g|}{3} } \sigma \right)
= \frac{1}{2} \ln \left[ 
( \cos\tfrac{\varphi}{2} )^2 + (\sin\tfrac{\varphi}{2}  +  \sqrt{ \tfrac{|g|}{3} } \sigma  )^2
\right] + \imath  \; \mathrm{Arg}(  1 -  e^{\imath \frac{\pi + \varphi}{2}} \sqrt{ \tfrac{|g|}{3} } \sigma) \;,
\end{equation}
where the $\mathrm{Arg}$ function is the principal branch of the argument, valued in $(-\pi,\pi)$. 
In particular, a change of variables $\sigma \to -\sigma$ shows that $ Z^{\mathbb{R}}(g,N) =  Z^{\mathbb{R}}( e^{2 \pi \imath} g,N)$, which is thus a single-valued function on $\mathbb{C}_{\pi}$, with a jump at $\mathbb{R}_-$.
The analytic continuation of the partition function $Z(g,N)$ will instead be a multi-valued function on $\mathbb{C}$, and thus require the introduction of a Riemann surface. We can view $Z^{\mathbb{R}}(g,N)$ as a periodic function on the same Riemann surface, with of course $Z(g,N) = Z^{\mathbb{R}}(g,N)$ for $g\in \mathbb{C}_{\pi}$, and $Z(g,N) \neq Z^{\mathbb{R}}(g,N)$ once one steps out of the principal Riemann sheet.

We collect all the relevant result concerning the partition function in Proposition~\ref{prop:Z}. For now we restrict to $N$ real, but the proposition can be extended to complex $N$ with little effort.\footnote{For $N$ complex with positive real for instance one uses the inequality $|z^{-N/2}| \le |\mathrm{Re}(z)|^{-\mathrm{Re}(N)/2} e^{\pi |\mathrm{Im}(N)|/2}$.} We will drop this assumption later.

\ 

\begin{proposition}[Properties of $Z(g,N)$] \label{prop:Z} Let $N \in \mathbb{R}$ be a fixed parameter. The partition function $Z(g,N)$ satisfies the following properties:

\begin{enumerate}

\item\label{propZ:1} for every $g\in\mathbb{C}_\pi$, the integral in Eq.~\eqref{eq:Zsigma} is absolutely convergent  and bounded from above by
\begin{equation}\label{eq:Zaprioribound}
 |Z(g,N)| \le \begin{cases}
 \left( \cos\f{\varphi}{2} \right)^{- N /2}   \;, \qquad & N \ge 0\\[5pt]
2^{|N|/2}+\f{2^{3|N|/4}}{\sqrt{\pi}} \frac{|g|^{N/4}}{3^{|N|/4}} \Gamma( \tfrac{|N|+2}{4}  ) 
 \;, \qquad & N < 0
 \end{cases} \; ,
\end{equation}
hence $Z(g,N)$ is analytic in $\mathbb{C}_\pi$.

\item\label{propZ:2} For $g\in\mathbb{C}_\pi$, the partition function is $Z(g,N) = Z^{\mathbb{R}}(g,N) $
and has the perturbative expansion:
\be \label{eq:Zpertseries}
Z(g,N) \simeq \sum_{n=0}^{\infty} \, \frac{\Gamma(2n+N/2) }{2^{2n}n! \, \Gamma(N/2) }    \; \left(- \frac{2g}{3}\right)^n \;,
\ee
where $\simeq$ means that the equation has to be interpreted in the sense of asymptotic series, i.e.:
\be
\lim_{\substack{g\to0 \\ g\in\mathbb{C}_\pi}} g^{-n_{\rm max}} \left| Z(g,N) - \sum_{n=0}^{n_{\rm max}}  \, \frac{\Gamma(2n+N/2) }{2^{2n}n! \, \Gamma(N/2) }   \; 
\left(- \frac{2g}{3}\right)^n  \right| = 0\;, \qquad \forall n_{\rm max}\geq 0\;.
\ee

\item\label{propZ:3} The function $Z(g,N)$ is Borel summable along all the directions in $\mathbb{C}_\pi$. 

\item\label{propZ:4} $Z(g,N)$ can be continued past the cut, on the entire Riemann surface. However, $\mathbb{R}_-$ is a Stokes line, that is, the anticlockwise and clockwise analytic continuations $Z_+(g,N)$ and $Z_-(g,N)$ are not equal and cease to be Borel summable at $\mathbb{R}_-$. A second Stokes line is found at $\mathbb{R}_+$ on the second sheet. For $\varphi \notin \pi\mathbb{Z}$, the analytic continuation of the partition function to the whole Riemann surface writes:
\begin{equation}
\begin{split}
& 2k\pi <  |\varphi| < (2k+1) \pi  :  \\[.5ex]
&\qquad Z(g,N )  =  \omega_{2k}  \;  Z(e^{\imath ( 2 k ) \tau \pi  } g,N) \\
&\qquad\hphantom{Z(g,N )  = } + \eta_{2k} \;   \frac{\sqrt{ 2\pi } }{\Gamma(N/2)}  \;  e^{\imath \tau \frac{\pi}{2}  } \;  e^{\frac{3}{2g}} \; \left(   e^{\imath  (2k + 1) \tau \pi  } \frac{ g} {3} \right)^{\frac{1-N}{2} }    \; Z ( e^{\imath  (2k + 1) \tau \pi  } g , 2-N) \\
&\qquad\hphantom{Z(g,N )} =  \omega_{2k}  \;  Z^{\mathbb{R}}( g,N)  + \eta_{2k} \;   \frac{\sqrt{ 2\pi } }{\Gamma(N/2)}  \;  e^{\imath \tau \frac{\pi}{2}  } \;  e^{\frac{3}{2g}} \; \left(   e^{\imath  (2k + 1) \tau \pi  } \frac{ g} {3} \right)^{\frac{1-N}{2} } \; Z^{\mathbb{R}} ( - g , 2-N)   \;,  \\[1ex]
& (2k+1)\pi < |\varphi| < (2k+2) \pi  :  \\[.5ex]
&\qquad Z(g,N ) = \omega_{2k+1}  \; Z(e^{\imath ( 2 k+2) \tau \pi  } g,N)  \\
&\qquad\hphantom{Z(g,N )  = }  +   \eta_{2k+1} \;  \frac{\sqrt{ 2\pi } }{\Gamma(N/2)}  \;  e^{\imath \tau \frac{\pi}{2}  } \;  e^{\frac{3}{2g}} \; \left(   e^{\imath  (2k + 1) \tau \pi  } \frac{ g} {3} \right)^{\frac{1-N}{2} }    \; Z ( e^{\imath (2k + 1) \tau \pi  } g , 2-N)   \\
&\qquad\hphantom{Z(g,N )} = \omega_{2k+1}  \; Z^{\mathbb{R}}( g,N) + \eta_{2k+1} \; \frac{\sqrt{ 2\pi } }{\Gamma(N/2)}  \;  e^{\imath \tau \frac{\pi}{2}  } \;  e^{\frac{3}{2g}} \; \left(   e^{\imath  (2k + 1) \tau \pi  } \frac{ g} {3} \right)^{\frac{1-N}{2} } \; Z^{\mathbb{R}} (  - g , 2-N)   
\;,
\end{split}
\end{equation}

where $\t=-\sgn(\varphi)$ and the Stokes parameters $(\omega,\eta)$ are defined recursively as:
 \begin{equation}
 (\omega_0 ,\eta_0) = (1,0)  \;,\qquad \begin{cases}
  \omega_{2k+1} &=  \omega_{2k} 
  \\
  \eta_{2k+1} & =     \eta_{2k} + \omega_{2k}
\end{cases} \; , \qquad
\begin{cases}
  \omega_{2(k+1)} &= \omega_{2k+1} + \tilde \tau \;  \eta_{2k+1}
  \\
  \eta_{2(k+1)} & =  e^{ \imath \tau \pi(N-1) }   \eta_{2k+1}
\end{cases} \; ,
 \end{equation}
 with $\tilde\tau = e^{\imath \tau \pi \frac{N +1 }{2}   } \; 2 \sin \frac{N\pi}{2}$.
The recursion gives:
\be \label{eq:rec_sol}
(\omega_{2k},\eta_{2k})=
\begin{cases}
e^{\imath\tau\pi N\f{k}{2}}\;(1,0) &,\,k\ \text{even} \\
e^{\imath\tau\pi N\f{k+1}{2}}\;(1,-1)&,\,k\ \text{odd}
\end{cases}\;.
\ee
 
The monodromy group of $Z(g,N)$ is of order 4  if $N$ is odd, and of order 2 if $N$ is even. More generally, we have a monodromy group of finite order if $N$ is a rational number, and an infinite monodromy otherwise.

\item\label{propZ:4.5} From Properties~\ref{propZ:2} and~\ref{propZ:4} we obtain that for $g$ in the sector 
$k\pi < |\varphi| < (k+1) \pi$ of the Riemann surface the partition function has the following transseries expansion:
 \begin{equation}\label{eq:fullZ}
 \begin{split}
     Z(g,N) \simeq & \omega_k \sum_{n=0}^{\infty} \, \frac{\Gamma(2n+N/2) }{2^{2n}n! \, \Gamma(N/2) }    \; \left(- \frac{2g}{3}\right)^n \crcr 
     & +  \eta_k \;  e^{\imath \tau \pi  (1- \frac{N}{2} ) }   \;  \sqrt{2\pi}  
    \left(  \frac{ g} {3} \right)^{\frac{1-N}{2} }
  e^{\frac{3}{2g}} 
    \sum_{q\ge 0} \frac{1}{ 2^{2q} q! \; \Gamma(\frac{N}{2} -2q ) } \left( \frac{2 g}{3} \right)^q \; ,
 \end{split} 
 \end{equation}
where we used:
\begin{equation}
\Gamma( 2q+ 1 - N/2)  \Gamma(N/2 - 2q) = \frac{ \pi}{ \sin( \pi \frac{N}{2} - 2\pi q ) } = \Gamma(1-N/2) \Gamma(N/2) \;.
\end{equation}
The transseries is resurgent, i.e.~the instanton series is obtained from the perturbative one by substituting $N\to 2-N$ and $g\to -g$ and vice versa.

\item\label{propZ:5} From Property~\ref{propZ:4.5}, the discontinuity of the partition function at the negative real axis:
\begin{equation}
    {\rm disc}_\pi \big(Z(g,N)\big)\equiv \lim_{ g \to \mathbb{R}_-} \bigg(Z_-(g,N) - Z_+(g,N) \bigg) \;,
\end{equation}
has the following asymptotic expansion: 
\be \label{eq:discZ}
\begin{split}
{\rm disc}_\pi \big(Z(g,N)\big) \simeq &
  \imath \sqrt{2\pi}  
    \left(  \frac{ | g | } {3} \right)^{\frac{1-N}{2} }
  e^{-\frac{3}{2|g|}} 
    \sum_{q\ge 0} \frac{1}{ 2^{2q} q! \; \Gamma(\frac{N}{2} -2q ) } \left( - \frac{2 |g|}{3} \right)^q 
 \crcr
& =  \imath \sqrt{\f{2}{\pi}}  \sin\frac{N\pi}{2} \; 
    \left(  \frac{ |g| } {3} \right)^{\frac{1-N}{2} }
  e^{- \frac{3}{2|g|}} 
    \sum_{q\ge 0} \frac{ \Gamma(2q + \frac{2-N}{2} ) }{ 2^{2q}\; q!  } \left( - \frac{2 |g|}{3} \right)^q 
\;,
\end{split}
\ee
where for $N$ even integer the sum truncates at $q=N/4-1$, if $N$ is a multiple of 4, and at $q=\lfloor N/4 \rfloor$ otherwise.

\item\label{propZ:7} The partition function obeys an homogenous linear ordinary differential equation:
\begin{equation} \label{PDEpartitionfunctionN}
	16 g^2 Z^{\prime \prime}(g,N)+\left( (8N+ 24)g+24 \right)Z^\prime(g,N)+N(N+2)Z(g,N)=0 \; ,
\end{equation}
which can be used to reconstruct the resurgent transseries expansion of $Z(g,N)$.

\end{enumerate}
\end{proposition}

\begin{proof}
See Appendix~\ref{app:proofZ}.
\end{proof}

\bigskip

The proof of this proposition is quite technical. The most interesting points come at Property~\ref{propZ:4}. While the full details can be found in Appendix~\ref{app:proofZ}, we discuss here how the Stokes phenomenon arises in the intermediate field representation. 

In order to obtain the asymptotic approximation of an integral we need to deform the integration contour to steepest-descent contours (or Lefschetz thimbles) where the Laplace method can be applied. An integration contour will in principle intersect several steepest-ascent (upwards) paths of several saddle points and it must then be deformed (i.e.~relaxed under the gradient flow) to run along the thimbles of these saddle points \cite{Witten:2010cx,Aniceto:2018bis}. When varying some parameter continuously the relevant thimbles can collide and change discontinuously leading to discontinuous changes of the asymptotic regimes at Stokes lines.
This is exactly the picture in the $\phi$ representation of the partition function, which we recall in Appendix~\ref{app:expansions}. 

In the $\sigma$ representation the picture is different.
Let us go back to Eq.~\eqref{eq:Zsigma} expressing $Z(g,N)$ as an integral over the real line.
The partition function $Z(g,N)$ is analytically continued to the extended Riemann sheet $\mathbb{C}_{3\pi/2}$ by tilting the integration contour to $e^{-\imath \theta}\sigma$ respectively $e^{\imath \theta}\sigma$ with $\theta > 0$ for its anticlockwise respectively clockwise analytic continuations $Z_+(g,N)$ and $Z_-(g,N)$. 

In this representation, the Lefshetz thimble is always the real axis, irrespective of $g$. In fact, as $g$ goes to zero, the Laplace method instructs us to look for the saddle point of the exponent in the integrand (the function $f(x)$ in Eq.~\eqref{eq:thimbles}), which in this case is a simple quadratic function,\footnote{This is perhaps more clear after rescaling $\sigma$ by $1/\sqrt{g}$ to cast our integral in the same form as Eq.~\eqref{eq:thimbles}.} while the subexponential function  (the function $a(x)$ in Eq.~\eqref{eq:thimbles}) is irrelevant for the dermination of the saddle points.
However, what happens is that the integrand has a branch point (or pole for $N$ a positive even integer) and this point crosses the thimble at the Stokes line. This is depicted in Fig.~\ref{fig:sigma-contour}.
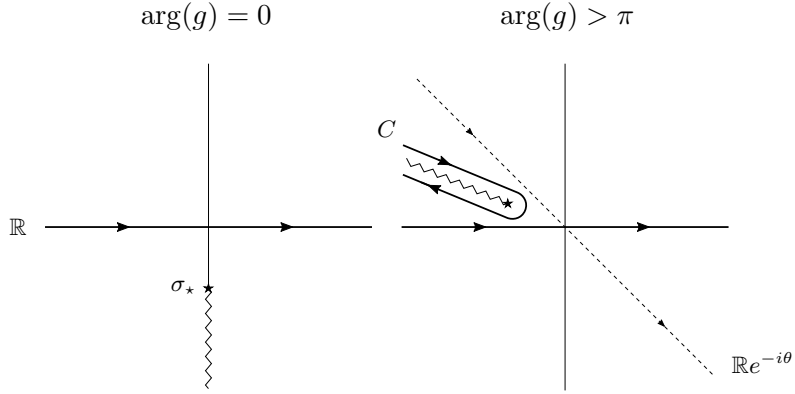
\begin{figure}[htb]
    \centering
    \begin{tikzpicture}[font=\footnotesize]
        \node at (0,0) {\includegraphics[scale=.5]{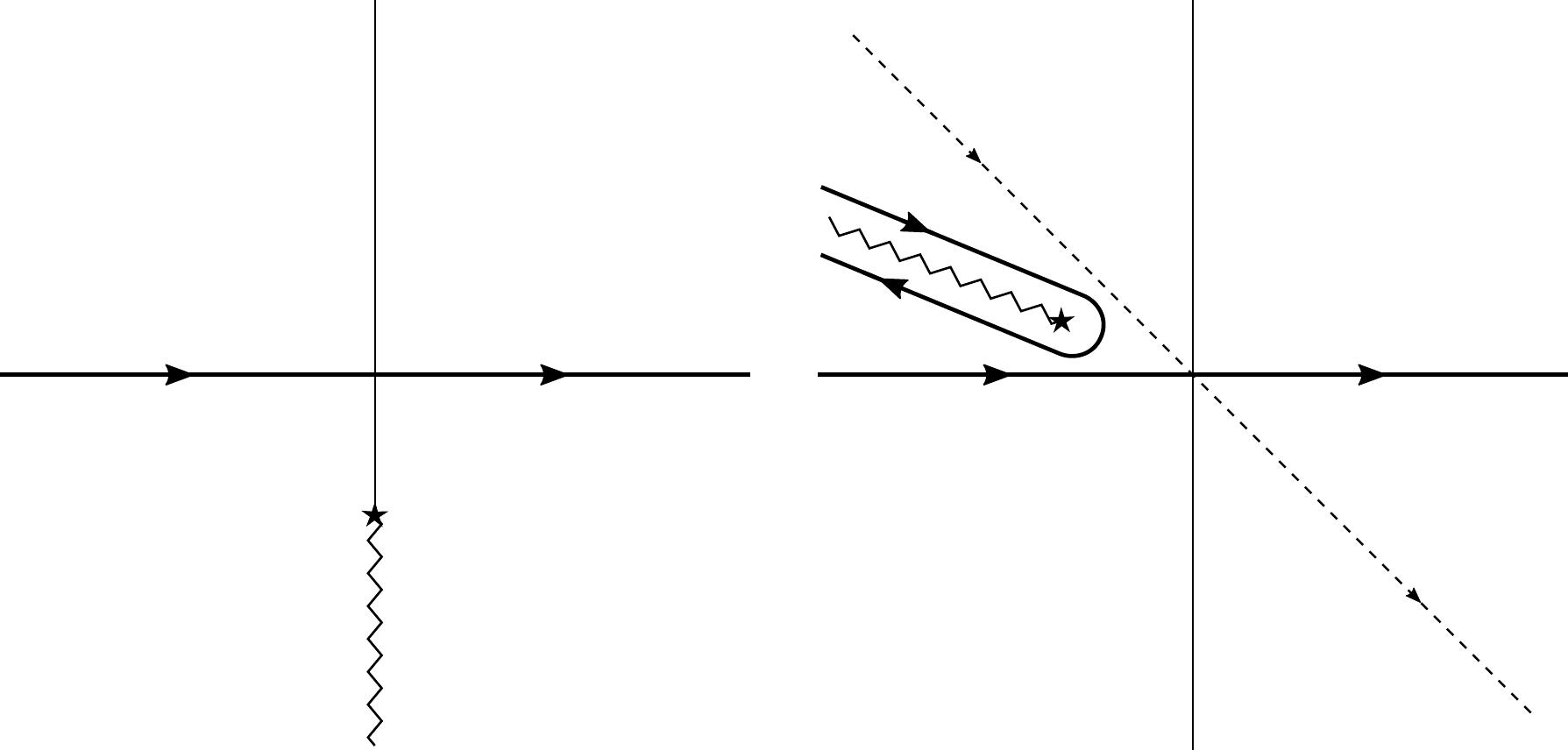}};
        \node at (-2.4,2.8) {\normalsize$\arg (g)=0$};
        \node at (2.4,2.8) {\normalsize$\arg (g) > \pi$};
        \node at (-4.9,0) {$\mathbb{R}$};       
        \node at (-2.7,-0.8) {$\sigma_\star$};
        \node at (0,1.3) {$C$};
        \node at (5,-1.8) {$\mathbb{R}e^{-i\theta}$};
    \end{tikzpicture}
    \caption{As $\arg(g)$ increases, the branch cut moves clockwise in the complex $\sigma$-plane. When $g$ crosses the negative real axis the tilted contour is equivalent to a Hankel contour $C$ plus the original contour along the real line \eqref{eq:sigma-contour1}.}
    \label{fig:sigma-contour}
\end{figure}

In detail, as long as $g = |g| e^{\imath\varphi}\in \mathbb{C}_{\pi}$, the integral converges because the branch point $\sigma_\star = -\imath \sqrt{ 3/g }  $
with branch cut $\sigma_\star \times(1,+\infty)$,
lies outside the integration contour. As $g$ approaches $\mathbb{R}_-$,  the branch point hits the contour of integration: for $\varphi \nearrow \pi$ the branch point hits the real axis at $ - \sqrt{3/|g|}$.
The analytic continuation $Z_+(g,N)$ consists in tilting the contour of integration in $\sigma$ to avoid the collision with the branch point. However,  in order to derive the asymptotic behaviour of $Z_+(g,N)$, we need to deform the integration contour back to the thimble, which is always the real axis. Once $g$ passes on the second Riemann sheet ($ \varphi > \pi$), when deforming the tilted contour to the the real axis we generate an additional Hankel contour (see Fig.~\ref{fig:sigma-contour}):
\begin{equation}\label{eq:sigma-contour1}
\begin{split}
    & Z_\pm(g,N)\big{|}_{\varphi\genfrac{}{}{0pt}{}{>\pi}{< -\pi}}=\int_{ e^{  \mp \imath \theta} \mathbb{R}} [d\sigma]  \;\frac{e^{-\frac{1}{2} \sigma^2} }{\left( 1 - \imath \sqrt{ \frac{g}{3} } \sigma\right)^{N/2} } = Z^{\mathbb{R}}(g,N) + Z^C_\pm(g,N)  \crcr
    & Z^{\mathbb{R}}(g,N) = 
    \int_{ \mathbb{R}} [d\sigma]  \;e^{-\frac{1}{2} \sigma^2} \frac{1}{\left( 1 - \imath \sqrt{ \frac{g}{3} } \sigma\right)^{N/2} } \;,\qquad  Z^C_\pm(g,N) = \int_{C} [d\sigma] \;e^{-\frac{1}{2} \sigma^2} \frac{1}{\left( 1 - \imath \sqrt{ \frac{g}{3} } \sigma\right)^{N/2} } \;,
\end{split}
\end{equation}
where the Hankel contour $C$ turns clockwise around the cut $\sigma_\star \times (1,+\infty)$, i.e.~starting at infinity with argument $ \frac{3\pi}{2} -\frac{\varphi}{2}$ and going back with argument $ -\frac{\pi}{2} -\frac{\varphi}{2}$ after having encircled the branch point $\sigma_\star$.
We kept a subscript $\pm$ for the contribution of the Hankel contour, because, even though the definition of $Z^C_\pm(g,N)$ and $C$ might suggest that it is one single function of $g$, in fact the integral around the cut is divergent for $|\varphi|<\pi/2$, and therefore the integrals at $\pi<\varphi<3\pi/2$ and  at  $-\pi>\varphi>-3\pi/2$ are \emph{not} the analytic continuation of each other. This fact is reflected in the $\t$-dependence of the asymptotic expansion in Eq.~\eqref{eq:fullZ}. 

The appearance of the Hankel contour marks a discontinuity of the contour of integration in $\sigma$ as a function of the argument of $g$, which translates into a discontinuity of the asymptotic expansion of $Z(g,N)$, that is, a Stokes phenomenon. We insist that the Hankel contour is \emph{not} a thimble for the integral in Eq.~\eqref{eq:Zsigma}, but it \emph{contributes} to the asymptotic evaluation of the integral, providing the one-instanton contribution in the transseries of $Z(g,N)$.

In order to go beyond $|\varphi|=3\pi/2$, one notices that we can analytically continue  separately $Z^{\mathbb{R}}(g,N)$ and $Z^C_\pm(g,N)$. The first is analytic in the range $|\varphi|\in(\pi,3\pi)$, where its asymptotic expansion is just the standard perturbative series.
The analytic continuation of $Z^C_\pm(g,N)$ is not immediate in the Hankel contour representation, which is only convergent for $\pi/2<|\varphi|<3\pi/2$, but after resolving the discontinuity at the cut 
and using again the Hubbard-Stratonovich trick
(as detailed in Appendix~\ref{app:proofZ}), it turns out that it can be rewritten as: 
\begin{equation}\label{eq:Zpmmain}
     Z^C_\pm(g,N)  =    e^{\imath \tau \pi (1-\frac{N}{2})  }\left(\frac{ g} {3} \right)^{\frac{1-N}{2} }
  e^{\frac{3}{2g}}  \frac{\sqrt{ 2\pi } }{\Gamma(N/2)}  \; Z ( e^{\imath \tau \pi} g , 2-N) 
  \;.
\end{equation}
with $\tau = -\sgn(\varphi)$. In this form it is manifest that $ Z^C_\pm(g,N)$ is analytic in $g$ as long as 
$ e^{\imath \tau \pi} g $ belongs to the principal sheet of the Riemann surface, that is for $\pi<|\varphi|<2 \pi$.
 We have thus shown that, when going from $|\varphi| <\pi$ to $\pi < |\varphi| < 2\pi$ our analytic continuation of $Z(g,N)$ switches:
\begin{equation}\label{eq:nicecontmain}
\begin{split}
Z(g,N) \xrightarrow[]{ |\varphi| \nearrow \pi_+} &\; Z^{\mathbb{R}}(g,N) + \frac{\sqrt{ 2\pi } }{\Gamma(N/2)}  \;  e^{\imath \tau \frac{\pi}{2}  } \;  e^{\frac{3}{2g}} \; \left( e^{\imath \tau \pi }   \frac{ g} {3} \right)^{\frac{1-N}{2} }  Z^{\mathbb{R}} ( - g , 2-N) 
  \crcr 
&  =  Z( e^{\imath (2\tau \pi) } g,N) +  
\frac{\sqrt{ 2\pi } }{\Gamma(N/2)}  \;  e^{\imath \tau \frac{\pi}{2}  } \;  e^{\frac{3}{2g}} \; \left( e^{\imath \tau \pi }   \frac{ g} {3} \right)^{\frac{1-N}{2} }  \; Z ( e^{\imath \tau \pi} g , 2-N) 
  \;,
  \end{split}
\end{equation}
where we explicitly exhibited the argument at which the switching takes place.
In the second line above,  for $\pi< |\varphi| < 2\pi$, both arguments $  e^{\imath (2\tau \pi) } g$ and $ e^{\imath \tau \pi} g$ belong to the principal sheet of the Riemann surface, where $Z(g,N)$ has already been constructed and proven to be analytic.  The first term in Eq.~\eqref{eq:nicecontmain} is regular up to $|\varphi| = 3\pi$, but the second one has a problem when $e^{\imath \tau \pi} g $ reaches the negative real axis and retracing our steps we conclude that the analytic continuation switches again:
\begin{equation}\label{eq:niceothercontmain}
\begin{split}
Z ( e^{\imath \tau \pi} g , 2-N)   \xrightarrow[]{ |\varphi| \nearrow 2\pi_+}\ & Z ( e^{\imath (3\tau \pi)} g , 2-N)  \\
& +  \frac{\sqrt{ 2\pi } }{\Gamma(1-N/2)}  \;  e^{\imath \tau \frac{\pi}{2}  } \;  e^{\frac{3}{2ge^{\imath \tau \pi}}} \; \left( e^{\imath (2\tau \pi) }   \frac{ g} {3} \right)^{\frac{N-1}{2} }   \; Z ( e^{\imath (2 \tau \pi ) } g , N) \; ,
\end{split}
\end{equation}
where this time the arguments at which $Z$ is evaluated on the right hand side stay in the principal sheet for $2\pi < |\varphi| < 3\pi$. Iterating, one obtains the analytic continuation to the whole Riemann surface.

\begin{remark} \upshape
The differential equation \eqref{PDEpartitionfunctionN} can be solved in terms of special functions.
For $N=1$, setting $Z(g,1)=\sqrt{\f{3}{2\pi g}} e^{\f{3}{4g}} f(\f{3}{4g})$, we find that the equation reduces to a modified Bessel's equation for $f(z)$:
\begin{equation} \label{eq:Bessel}
    z^2 f^{\prime \prime}(z) + z f^{\prime}(z) - (z^2 +\f{1}{16}) f(z)=0 \;.
\end{equation}
Its two lienarly independent solutions are the modified Bessel functions of the first and second kind of order $1/4$. However, only the second, $K_{1/4}(z)$, decays for $z\to\infty$, hence the initial condition $Z(0,1)=1$ fixes $f(z)=K_{1/4}(z)$.

Similarly, for general $N$, we find that setting $Z(g,N)=(\f{3}{2 g})^{N/4} f(\f{3}{2 g})$ the differential equation reduces to Kummer's equation  for $f(z)$:
\begin{equation}
    z f^{\prime \prime}(z) + (\f12-z) f^{\prime}(z) - \f{N}{4} f(z)=0 \;.
\end{equation}
With the addition of the initial condition condition $Z(0,N)=1$, we find that the solution is given by the Tricomi confluent hypergeometric function $f(z)=U(N/4,1/2,z)$.

However, for our purposes, such expressions of the partition function in terms of special functions do not gain us much, as actually such functions are defined either as integrals (so this just amounts to some manipulation and change of variables in Eq.~\eqref{eq:Zsigma}) or as a series in $z$, meaning that they are obtained from a large-$g$ expansion (see \cite{Aniceto:2018bis}). Moreover, the analytic continuation of Tricomi's function is usually given in terms of $U(a,b,z)$ itself and of Kummer's function $M(a,b,z)$, so that it is not obvious how to translate that for $Z(g,N)$.
\end{remark}

\subsection{Convergent small-$N$ series of $Z(g,N)$ and transseries of its coefficients $Z_n(g)$} 

We will now study the discontinuity of $Z(g,N)$ from a different perspective. We expand the integrand of Eq.~\eqref{eq:Zsigma} in powers of $N$ and exchange the order of summation and integration:
\be \label{eq:Z-logseries}
Z(g,N) = \sum_{n\geq 0} \f{1}{n!} \left(-\frac{N}{2}\right)^n Z_n(g)\; ,\quad Z_n(g)= \int_{-\infty}^{+\infty} [d\s] e^{-\f12 \s^2} \left( \ln(1-\imath \sqrt{\f{g}{3}} \s ) \right)^n \;.
\ee
Unlike the usual perturbative expansions in $g$, this is a convergent expansion: from the bound 
in Property~\ref{propZn:1} of Proposition~\ref{prop:Zn} below, the Gaussian integral and the sum can be commuted due to Fubini's Theorem. 
As a function of $N$, we can regard $Z(g,N)$ as a generating function of ``moments": unlike the usual moments, we are dealing with expectations of powers of the logarithm.

\begin{proposition}[Properties of $Z_n(g)$] \label{prop:Zn} 
The $Z_n(g)$, $n\in\mathbb{N}_{\ge 0}$ satisfy the following properties:
\begin{enumerate}
\item\label{propZn:1} $Z_n(g)$ is analytic in the cut plane $\mathbb{C}_\pi$. Indeed, for every $g\in\mathbb{C}_\pi$, the integral \eqref{eq:Z-logseries} is absolutely convergent and bounded from above by: 
\begin{equation}
| Z_n(g) | \le K^n \frac{ \bigg(  |\ln ( \cos\frac{\varphi}{2}) |+1 \bigg)^n   }{ \ve^n  }
\bigg(1 + |g|^{\frac{n\ve}{2}} \Gamma(\tfrac{n\ve + 1}{2}  ) \bigg) \; ,
\end{equation}
for any $\ve >0$ and with $K$ some $g$-independent constant. Using this bound with some fixed $\ve<2$ shows that, $\forall g \in \mathbb{C}_\pi$ the series in Eq.~\eqref{eq:Z-logseries} has infinite radius of convergence in $N$.

\item\label{propZn:2} For $g\in\mathbb{C}_\pi$, $Z_n(g)$ has the perturbative expansion:
\be \label{eq:Znpertseries}
Z_n(g)\simeq   \sum_{m\geq n/2} \left( -\f{2g}{3}\right)^m  \, \f{(2m)!}{2^{2m} m!}
 \sum_{\substack{m_1,\ldots, m_{2m-n+1} \ge 0  \\ \sum k  m_k=2m , \; \sum m_k=n}} \f{(-1)^n n!}{\prod_k k^{m_k}m_k!} \, \equiv \, Z_n^{{\rm pert.}}(g) \; .
\ee

\item\label{propZn:3} The functions $Z_n(g)$ are Borel summable along all the directions in $\mathbb{C}_\pi$.

\item\label{propZn:4} $Z_n(g)$ can be continued past the cut on the extended Riemann sheet $\mathbb{C}_{3\pi/2}$,
and the small-$N$ series has infinite radius of convergence in $N$ in this domain. 
However, $\mathbb{R}_-$ is a Stokes line and the anticlockwise and clockwise analytic continuations $Z_{n+}(g)$ and $Z_{n-}(g)$ are not equal and cease to be Borel summable at $\mathbb{R}_-$.

\item\label{propZn:5} 
For $g\in\mathbb{C}_{3\pi/2}$, $Z_n(g)$ has a the following transseries expansion:
 \begin{equation}\label{eq:transsereiesZn}
     Z_n(g) \simeq Z_n^{{\rm pert.}}(g) + \eta \, e^{\frac{3}{2g}} Z_n^{(\eta)}(g) \;, 
\end{equation}
with $Z_n^{{\rm pert.}}(g)$ as in Eq.~\eqref{eq:Znpertseries}, and:
 \begin{equation}
 \begin{split}
 Z_n^{(\eta)}(g) = &    \f{\imath }{\sqrt{2\pi}}\sqrt{ \f{  g}{3}}
\sum_{q=0}^{\infty} \sum_{p=0}^{n} \f{1}{q!} \left(\f{g}{6}\right)^q
\binom{n}{p} \f{d^p \G(z)}{dz^p}\Big|_{z=2q+1}  \crcr
& \qquad \qquad \qquad \qquad \qquad \left[\left(\ln( e^{\imath \tau \pi } \f{g}{3}) -i\pi \right)^{n-p}
-\left(\ln( e^{\imath \tau \pi } \f{g}{3})  +i\pi \right)^{n-p}\right]\; ,
\end{split} 
\end{equation}
with $\tau = -\sgn(\varphi)$ and $\eta$ a transseries parameter which is zero on the principal Riemann sheet and one if $|\varphi|>\pi$.
Proceeding in parellel to Proposition~\ref{prop:Z}, one can study the full monodromy of $Z_n(g)$.

\item\label{propZn:6} The discontinuity on the negative axis has the following asymptotic expansion:
\be \label{eq:discZn} \begin{aligned}
{\rm disc}_\pi \big(Z_n(g)\big) \simeq\f{e^{-\f{3}{2|g|}}}{\sqrt{2\pi}}\sqrt{\f{|g|}{3}}
\sum_{q=0}^{\infty} \sum_{p=0}^{n} \f{1}{q!} &\left(-\f{|g|}{6}\right)^q
\binom{n}{p} \f{d^p \G(z)}{dz^p}\Big|_{z=2q+1} 
\\ & \left[\left(\ln|\f{g}{3}| -i\pi \right)^{n-p}
-\left(\ln|\f{g}{3}| +i\pi \right)^{n-p}\right]\; .
\end{aligned}\ee
Summing over $n$, the discontinuity of the partition function \eqref{eq:discZ} is recovered. 

\item\label{propZn:7}   The functions $Z_n(g)$  obey a tower of linear, inhomogenous ordinary differential equations:
\begin{equation} \label{PDEZn}
\begin{split}
    & Z_0(g)=1 \;,\\
    & 4 g^2 Z_1^{\prime\prime}(g)+6 \left(g+1 \right) Z_1^{\prime}(g)  = 1 \;,\\
    & 4 g^2 Z_n^{\prime\prime}(g)+6 \left(g+1 \right) Z_n^{\prime}(g)  =  n \left( 4 g Z_{n-1}^{\prime}(g)
    + Z_{n-1}(g) \right) -  n (n-1) Z_{n-2}(g) \;.
\end{split}
\end{equation}
which can be used to reconstruct the resurgent transseries expansion of $Z_n(g)$.
\end{enumerate}
\end{proposition}

\begin{proof}
See Appendix~\ref{app:proofZn}
\end{proof}

As with Proposition~\ref{prop:Z}, the most interesting points are Properties~\ref{propZn:4} and ~\ref{propZn:5}. Again the analytic continuations $Z_{n\pm}(g)$ of $Z_n(g)$ to the extended Riemann sheet $\mathbb{C}_{3\pi/2}$ are obtained by tilting the integration contour to $e^{\mp \imath \theta}\sigma$ with $\theta > 0$.  The branch point $\sigma_{\star}$ of the integrand in Eq.~\eqref{eq:Z-logseries} crosses the real axis when $g$ reaches $\mathbb{R}_-$, and deforming the tilted contours back to the real axis detaches Hankel contours around the cut $ \sigma_\star \times (1,+\infty)$:
\begin{equation}\label{eq:sigma-contour12}
\begin{split}
    & Z_{n\pm}(g)\big{|}_{\varphi\genfrac{}{}{0pt}{}{>\pi}{< -\pi}}=\int_{ e^{  \mp \imath \theta} \mathbb{R}} [d\sigma]  \; e^{-\frac{1}{2} \sigma^2} \left( \ln( 1 - \imath \sqrt{ \frac{g}{3} }  \sigma) \right)^{n}  = Z^{\mathbb{R}}_n(g) + Z^C_{n\pm}(g)  \crcr
   & Z^{\mathbb{R}}_n(g) = 
    \int_{ \mathbb{R}} [d\sigma]  \;e^{-\frac{1}{2} \sigma^2}  \left( \ln( 1 - \imath \sqrt{ \frac{g}{3} }  \sigma ) \right)^{n}  \;,\qquad  Z^C_{n\pm}(g) = \int_{C} [d\sigma] \;e^{-\frac{1}{2} \sigma^2} \left( \ln(1 - \imath \sqrt{ \frac{g}{3} } \sigma)\right)^{n}  \; .
\end{split}
\end{equation}
The transeries of $Z_n(g)$ is obtained by summing the asymptotic expansions of the two pieces:
\[
Z_{n}^{\mathbb{R}}(g) \simeq Z_n^{{\rm pert.}}(g) \;,\qquad Z_{n\pm}^C(g) \simeq e^{\frac{3}{2g}} Z^{(\eta)}_{n}(g)\big{|}_{\tau = \mp 1} \;.
\]

\bigskip

Notice that the homogeneous equation in Property~\ref{propZ:7} in Proposition~\ref{prop:Zn} is the same for all $n\geq 1$, and it admits an exact solution in the form of a constant plus an imaginary error function:
\begin{equation}\label{eq:homDiffEqZ0}
    4g^2 Z_1^{\prime \prime}(g) + 6(g+1)Z_1^\prime(g)=0 \qquad \Rightarrow \qquad    Z_1(g)=c_1+ c_2  \int_0^{\imath \sqrt{ \frac{3}{2g} } }e^{-t^2} dt \; .
\end{equation}
The asymptotic expansion of the error function reproduces the one-instanton contribution of Eq.~\eqref{eq:transsereiesZn} for $n=1$.
For $n>1$, instead, this is only part of the instanton contribution, the rest being generated by the recursive structure of the inhomogeneous equations.
Similarly, the perturbative expansion comes from the special solution to the inhomogeneous equation, even at $n=1$, as for those we cannot match exponential terms with the right-hand side.
For $n>1$ the homogeneous equation remains the same, but the inhomogeneous part depends on the solutions to previous equations, and thus it can also contain exponential terms.

\section{The free energy $W(g,N)$}
\label{sec:W}

We now turn to the free energy $W(g,N) = \ln Z(g,N)$. 
Our aim is find the equivalent of the results listed in Proposition \ref{prop:Z}, in the case of $W(g,N)$. Taking the logarithm has drastic effects: the nonperturbative effects encountered in $W(g,N)$ are significantly more complicated than the ones encountered for $Z(g,N)$. 
One can understand this from the fact that the linear differential equation satisfied by $Z(g,N)$ translates into a nonlinear one for $W(g,N)$, leading to an infinite tower of multi-instanton sectors in the transseries \cite{Aniceto:2018bis}.
Here we will follow a different route, based on the small-$N$ expansion.

Much like the partition function $Z(g,N)$, its logarithm $W(g,N)$ can also be expanded in $N$:
\begin{equation}\label{eq:ZW}
     W(g,N) = \ln(Z(g,N)) \equiv \sum_{n\ge 1}\frac{1}{n!} \left( -\frac{N}{2}\right)^n W_n(g) \;.
\end{equation}

The coefficients $W_n(g)$ can be computed in terms of $Z_n(g)$. As already mentioned, $Z_n(g)$ are the moments of the random variable $\ln(1-\imath \sqrt{g/ 3} \, \sigma)$, hence $W_n$ are the cumulants of the same variable and can be computed in terms of $Z_n(g)$ by using the M\"obius inversion formula (which in this case becomes the moments-cumulants formula). Let us denote $\pi$ a partition of the set $\{1,\dots n\}$, $b\in \pi$ the parts in the partition, $|\pi|$ the number of parts of $\pi$ and $|b|$ the cardinal of $b$. Then:
\begin{equation}
  Z_n (g)= \sum_{\pi} \prod_{b\in \pi} W_{|b|} (g)   \;,\qquad W_n (g) = \sum_{\pi} \lambda_{\pi} \prod_{b\in \pi}  Z_{|b|} (g) \;,
\end{equation}
where $\lambda_{\pi} = (-1)^{|\pi|-1} (|\pi|-1)!$ is the M\"oebius function on the lattice of partitions.
Grouping together the partitions having the same number $n_i$ of parts of size $i$ this becomes:\footnote{This can also be obtained directly from Fa{\`a} di Bruno's formula:
\[
    \frac{d^n}{du^n} \ln Z|_{u=0}
  =\sum_{\substack{n_1,\dots n_n \geq 0\\\sum i n_i = n }}
     \frac{n!}{\prod_i n_i! (i!)^{n_i}} \; (n_1 + \dots +n_n-1)! \;
      \frac{(-1)^{n_1 + \dots +n_n-1}}{ Z^{ n_n + \dots +n_n } } \prod_{i\ge 1} [ Z^{(i)} ]^{n_i}|_{u=0} \; ,
\]
noticing that for $u=-N/2$, we have $Z^{(i)}=Z_i$.
}
\begin{align}\label{eq:Moeb}
Z_n(g) & = \sum_{\substack{ n_1,\dots,n_n \ge 0\\ \sum in_i = n } } \frac{n!}{\prod_i n_i! (i!)^{n_i}}
      \prod_{i=1}^{n} W_i(g)^{n_i} \;,\crcr
W_n(g) & = \sum_{k= 1}^{n} (-1)^{k-1} (k-1)! \sum_{ \substack{n_1,\dots, n_{n-k+1}\ge 0 \\ \sum in_i = n ,\, \sum n_i = k }}  \frac{n!}{\prod_i n_i! (i!)^{n_i}} \prod_{i=1}^{n-k+1} Z_i(g)^{n_i} \;.
\end{align}

Eq.~\eqref{eq:Moeb} relates $W_n(g)$ and $Z_n(g)$ as analytic functions of $g$. However, this translates into a relation between $W(g,N)$ and $Z(g,N)$ which holds only \emph{in the sense of formal power series in $N$}. Even though $Z(g,N)$ is analytic in some domain, one can not conclude that $W(g,N)$ is also analytic in the same domain: convergence of the series defining $Z(g,N)$ does not imply convergence of the series defining $W(g,N)$ in Eq.~\eqref{eq:ZW}. This can most readily be seen at the zeros of the partition function, the so-called Lee-Yang zeros, which are singular points for the free energy. In order to study the analyticity properties of $W(g,N)$ one needs to use a completely different set of tools. However, as we will see below, the M\"obius inversion has its own uses: it is the most direct way to access the transseries expansion of $W(g,N)$. 

\subsection{Constructive expansion}

The following Proposition~\ref{prop:LVE+bounds} is a slight variation on the Loop Vertex Expansion (LVE) introduced in \cite{Rivasseau:2007fr} (see also \cite{Gurau:2014vwa} for more details). 
It gives an integral representation for $W_n(g)$ in Eq.~\eqref{eq:ZW} which allows us to prove that $W(g,N)$ is convergent (hence analytic) in a bounded domain on the extended Riemann sheet $\mathbb{C}_{3\pi/2}$, wrapping around the branch point at the origin. 

In Proposition~\ref{prop:Z} we fixed $N$ to be a real parameter. However, Eq.~\eqref{eq:Z-logseries} writes $Z(g,N)$ as an expansion in $N$ with a nonzero (infinite!) radius of convergence in $N$, as long as 
$|\varphi| < 3\pi/2$ (note that the bound in Property \ref{propZn:1} of Proposition~\ref{prop:Zn} suffices only for $|\varphi|< \pi$; in order to reach $|\varphi| < 3\pi/2$, one needs to use the improved bound in Eq.~\eqref{eq:problemeq}). We can therefore extend $N$ to a larger domain in the complex plane. 
As the following proposition shows, something similar applies also to $W(g,N)$, but with a finite radius of convergence.

\paragraph{Notation.}
Let us denote $T_n$ the set of combinatorial trees with $n$ vertices labeled $1,\dots n$. There are $\frac{ (n-2)! }{ \prod_{i=1}^n (d_i-1)!}$ trees over $n$ labelled vertices with coordination $d_i$ at the vertex $i$ and $\sum_i d_i = 2(n-1)$. The total number of trees in $T_n$ is $n^{n-2}$. 
Let $\cT\in T_n$ be such a tree. We denote $P_{k-l}^{\cT} $ the (unique) path in the tree $\cT$ connecting the vertices $k$ and $l$. If we associate to each edge $(k,l)\in \cT$ a variable $u_{kl}$ between 0 and 1, we can define the $n\times n$ matrix $w^\cT$:
\begin{equation}
      w^\cT_{kl} \equiv \begin{cases}
                     1 \;,\qquad & \text{if} \;\; k=l \\
                     \inf_{(i,j)\in   P_{k-l}^{\cT} } \{ u_{ij} \}  \;, \qquad &\text{else} 
                  \end{cases}  \; .
\end{equation}
The matrix $w^\cT$ is a positive matrix for any choice of $u$ parameters, and is stirctly positive outside a set of measure $0$ (see Appendix~\ref{app:BKAR} for more details). Of course the matrix $w$ depends on $u$ but we suppress this in order to simplify the notation. 

\begin{proposition}[The LVE expansion, analyticity]\label{prop:LVE+bounds}
Let $N$ be a fixed \emph{complex} parameter and let us denote $g= |g|e^{\imath \varphi}$. The cumulants $W_n(g)$ can be written as:
\be \label{eq:W_BKARterm}
\begin{split}
 W_1(g) & =Z_1(g)= \int_{-\infty}^{+\infty}  [d\sigma]  \; e^{-\frac{1}{2} \sigma^2 } 
   \ln\Big[ 1 -\imath \sqrt{\tfrac{g}{3}}   \sigma \Big] \;, \crcr
W_n(g) &  = - \left( \frac{g}{3} \right)^{n-1} \sum_{ \cT \in T_n } 
 \int_{0}^1 \prod_{ (i,j) \in \cT } du_{ij} 
\int_{-\infty}^{+\infty} \frac{\prod_i   [d\sigma_i] }{\sqrt{ \det w^\cT }}  \; e^{ - \frac{1}{2} \sum_{i,j} \sigma_i ( w^{\cT})_{ij}^{-1}  \sigma_j } 
\;\; \prod_{i} \frac{(d_i-1)!}{ \left(   1 -\imath \sqrt{\tfrac{g}{3}}   \sigma_i \right)^{d_i} }  \;,
\end{split}
\ee
where we note that the Gaussian integral over $\sigma$ is well defined, as $w^\cT$ is positive, and normalized. Furthermore:
\begin{enumerate}
\item\label{propLVE:1} The functions $W_n(g), n\ge 2$ are bounded by:
        \begin{equation}
        |W_n(g)| \le  \frac{(2n-3)!}{(n-1)!}
        \left| \frac{g}{3 (\cos\frac{\varphi}{2} )^{2} } \right|^{n-1} \;.
        \end{equation}
        Therefore, they are analytic in the cut plane $\mathbb{C}_{\pi}$.
\item\label{propLVE:2} The series
        \begin{equation} \label{eq:W_BKAR}
        W(g,N)  = \sum_{n\ge 1} \frac{1}{n!} \left( -\frac{N}{2} \right)^{n} W_n(g) 
        \end{equation}
        is absolutely convergent in the following cardioid domain:
        \be \label{eq:cardioid}
        \mathbb{D}_0 = \left\{ g\in\mathbb{C},\; g=|g| e^{\imath \varphi} :\;  
        |g|<\frac{1}{|N|} \; \frac{3}{2}  (\cos\frac{\varphi}{2} )^{2}\right\}
        \;. 
        \ee
\item\label{propLVE:3} $W_n(g)$ can be analytically continued to a subdomain of the extended Riemann sheet $\mathbb{C}_{3\pi/2}$ by tilting the integration contours to $\sigma\in e^{-\imath \theta} \mathbb{R}$:
    \be\begin{split}\label{eq:Wncontint}
    W_{1\theta}(g) = &   e^{-\imath \theta}\int_{-\infty}^{+\infty} [d\sigma]  \; e^{-\frac{1}{2} e^{-2\imath \theta} \sigma^2 } 
    \ln\Big( 1 -\imath \sqrt{\tfrac{g}{3}}   e^{-\imath \theta }\sigma \Big) \;, \crcr
    W_{n\theta}(g) = &  - \left( \frac{g}{3} \right)^{n-1}  \sum_{ \cT \in T_n  } 
    \int_{0}^1  \prod_{ (i,j) \in \cT } du_{ij} \crcr
    & \qquad \qquad \int_{\mathbb{R}} \frac{\prod_i e^{-\imath \theta} [ d\sigma_i]  }{\sqrt{ \det w^\cT }}  \; e^{ - \frac{1}{2} e^{-2\imath \theta}\sum_{i,j} \sigma_i ( w^\cT)_{ij}^{-1}  \sigma_j } 
    \;\; \prod_{i} \frac{(d_i-1)!}{ \left(   1 -\imath \sqrt{\tfrac{g}{3}}   e^{-\imath \theta}\sigma_i \right)^{d_i} }  \;.
    \end{split}\ee
\item\label{propLVE:4} For $n\ge 2$ we have the following bound:
    \begin{equation}\label{eq:boundW}
    |W_{n\theta}(g)|  \le \frac{(2n-3)!}{(n-1)!}  \frac{1}{\sqrt{ \cos(2\theta)  }} \left| \frac{g}{3
    \sqrt{\cos(2\theta)} \left( \cos\frac{\varphi-2\theta}{2} \right)^2
    } \right|^{n-1}    \;.    
    \end{equation}
\item\label{propLVE:5} The series
     \begin{equation}
     W_{\theta}(g,N)  = \sum_{n\ge 1} \frac{1}{n!} \left( -\frac{N}{2} \right)^{n} W_{n\theta}(g)   \;,
     \end{equation}
    is absolutely convergent in the following domain:
    \be \label{eq:W-domain}
    \mathbb{D}_{\theta} = \left\{ g\in\mathbb{C},\; g=|g| e^{\imath \varphi} :\;  
    |g| < \frac{1}{|N|} \; \frac{3}{2} \;  \left(\cos\frac{\varphi-2 \theta}{2} \right)^{2}  \; \sqrt{ \cos(2\theta) } \right\} \;.
    \ee
\end{enumerate}

Consequently,  $W_n(g)$ and $W(g,N)$ can be analytically extended to the following respective domains:
\begin{align}
 W_{n}(g) : \qquad  &   |2\theta| < \frac{\pi}{2} \; , \;\; |\varphi - 2\theta| < \pi    \;, \crcr
  W(g,N) : \qquad  &   |2\theta| < \frac{\pi}{2} \; , \;\; |\varphi - 2\theta| < \pi      \;,
\qquad   |g| < \frac{1}{|N|} \; \frac{3}{2} \;  \left(\cos\frac{\varphi-2 \theta}{2} \right)^{2}  \; \sqrt{ \cos(2\theta) }  \; .
\end{align}

Pushing $\theta \to \pm \pi/4$ allows us to write a convergent expansion for all $|\varphi|< \frac{3\pi}{2}$.
\end{proposition}

\begin{proof}
See Appendix~\ref{app:LVE}
\end{proof}

\begin{figure}[!htb]
	\centering
	\includegraphics[width=0.3\linewidth]{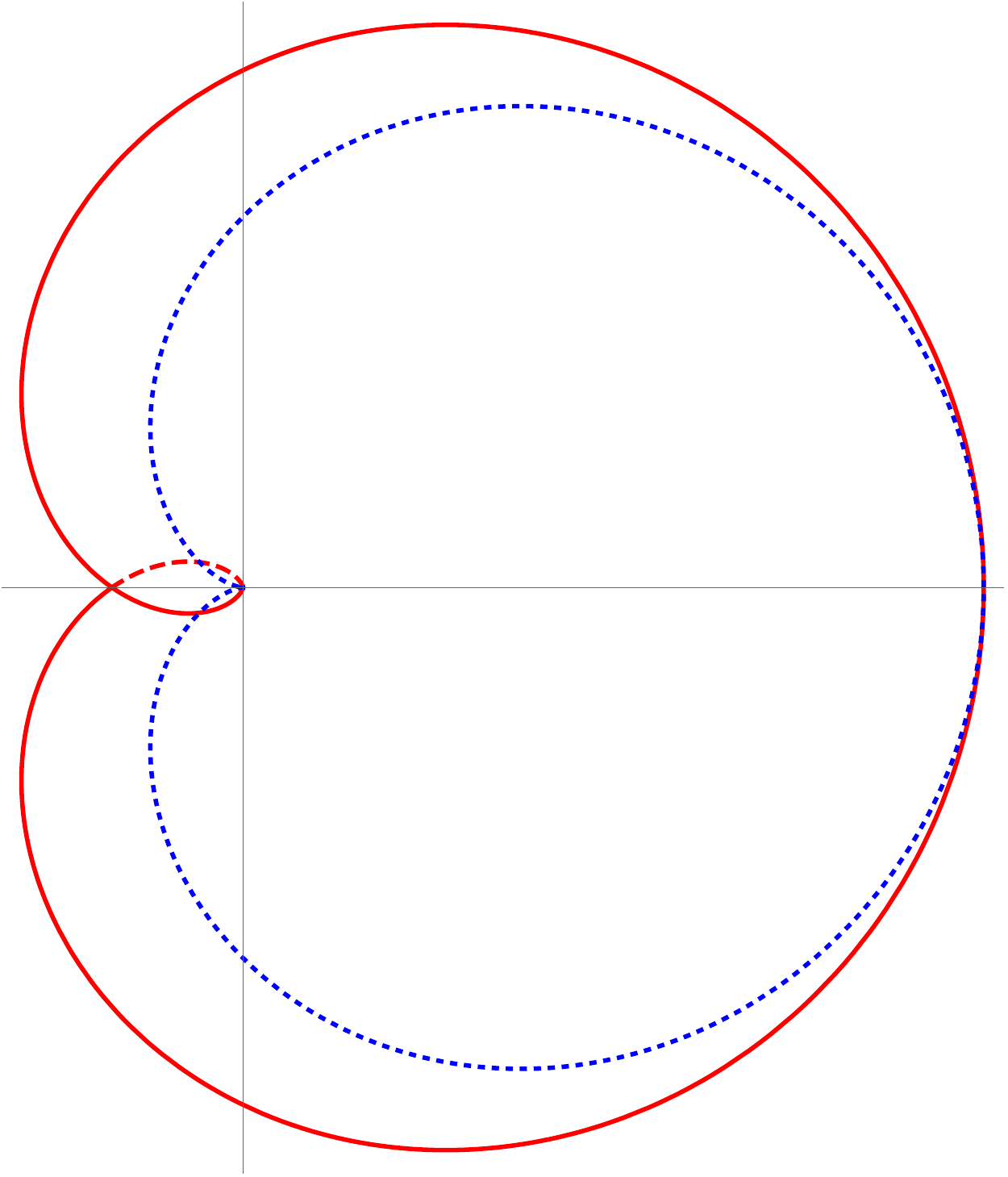} 
	\caption{The cardioid domain $\mathbb{D}_0$ of Eq.~\eqref{eq:cardioid} (dotted blue line) and the extended cardioid $\mathbb{D}_{\theta}$ of Eq.~\eqref{eq:W-domain} (red line), for $\th=\varphi/6$, in the complex $g$-plane. The branch cut is on the negative real axis, thus the portions of $\mathbb{D}_{\theta}$ going beyond it are to be understood as being on different Riemann sheets.}
	\label{fig:W-domain}
\end{figure}

The main point of the proposition is that by constructive methods we can prove analyticity of $W(g,N)$ in a nontrivial domain. In a first step, without touching the integration contours, we prove that such domain is the cardioid of Eq.\eqref{eq:cardioid}. However, the cardioid does not allow us to reach (and cross) the branch cut.
Tilting the integration contours by $\th$, we are able to extend the original cardioid domain to the larger domain of Eq.~\eqref{eq:W-domain} (see  Fig.~\ref{fig:W-domain}), going beyond the cut on a subdomain of the extended Riemann sheet $\mathbb{C}_{3\pi/2}$. The optimal domain can be found by maximizing the right-hand side of Eq.~\eqref{eq:W-domain} with respect to $\th$, at fixed $\varphi$, but a simpler and qualitatively similar choice is to take $\th=\varphi/6$.

Note that the domain of analyticity of $W(g,N)$, Eq.~\eqref{eq:W-domain}, depends on $N$ and shrinks to zero for $N\to \infty$. Results uniform in $N$ can only be established if one keeps the 't Hooft coupling $g_{t} = g N$ fixed \cite{Gurau:2014vwa}. On the other hand, for any $g$ on the extended Riemann sheet $\mathbb{C}_{3\pi/2}$, the radius of convergence of the LVE expansion in $N$ is nonzero.

\begin{remark} \upshape
It is also worth noticing that the explicit expressions for the partition function in terms of special functions, discussed around Eq.~\eqref{eq:Bessel}, provide us with some useful information about the  zeros of $Z(g,N)$ (Lee-Yang zeros), and hence about the singularities of $W(g,N)$.
For example, in the case $N=1$, the partition function is expressed in terms of a modified Bessel function of the second kind, whose zeros have been studied in some depth. In particular, from what is known about $K_{\nu}(z)$  (e.g.\ \cite{Watson-book}) we deduce that $Z(g,1)=\sqrt{\f{3}{2\pi g}} e^{\f{3}{4g}} K_{1/4}(\f{3}{4g})$  has no zeros in the principal sheet $\mathbb{C}_{\pi}$, while on each of the two following sheets it has an infinite sequence of zeros approaching the semiaxis at $|\varphi|=3\pi/2$ from the left, and accumulating towards $g=0$ (see Fig.~\ref{fig:LeeYang-zeors}).
Therefore, it should come as no surprise that $W(g,N)$ cannot be analytically continued around the origin beyond $|\varphi|=3\pi/2$.
\end{remark}

\begin{figure}[!htb]
	\centering
	\includegraphics[width=0.6\linewidth]{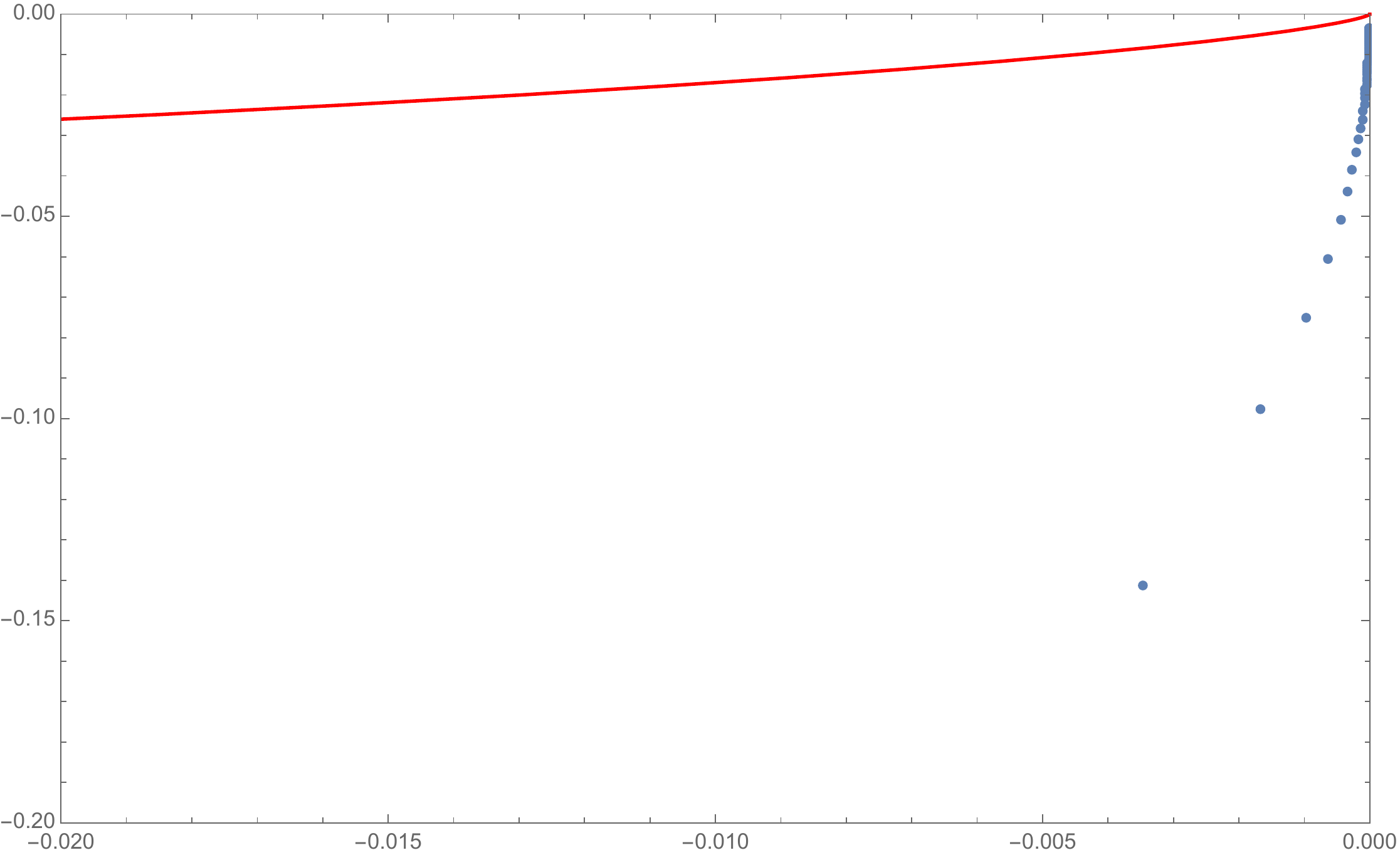}
	\caption{Approximate location (see \cite{Watson-book}) of the Lee-Yang zeros of $Z(g,1)$ (blue dots) in the quadrant $\pi<\varphi<3\pi/2$ of $\mathbb{C}_{3\pi/2}$, together with the boundary of the domain $\mathbb{D}_{\theta}$ (in red).}
	\label{fig:LeeYang-zeors}
\end{figure}

\begin{remark}\label{rem:Wnintegral} \upshape 
Integrating out the $u$ parameters and performing the sum over trees, one should be able to prove that integral expressions~\eqref{eq:Wncontint} reproduce the moment-cumulant relation in Eq.~\eqref{eq:Moeb}. In particular this would provide an alternative proof that the moment cumulant relation holds in the sense of analytic functions on the Riemann surface. The proof that this indeed happens is involved as the summation over trees requires the use of combinatoiral techniques similar to the ones discussed in Appendix~\ref{app:BKAR}. We postpone this for future work.
\end{remark}

\bigskip

In \cite{Rivasseau:2007fr} the LVE is used to prove the Borel summability of $W(g,1)$ along the positive real axis. Building on the techniques introduced in \cite{Rivasseau:2007fr}, we now generalize this result. 

\begin{proposition}[Borel summability of $W_n(g)$ and $W(g,N)$ in $\mathbb{C}_\pi$]\label{propW:BS}
The cumulants $W_n(g)$ and the free energy $W(g,N)$ at any fixed complex $N$ are Borel summable along all the directions in the cut complex plane $\mathbb{C}_\pi$.
\end{proposition}

\begin{proof}
See Appendix~\ref{app:BSW}.
\end{proof}

\subsection{Transseries expansion}
\label{sec:AsymptoticSeries}

It is well known that the (perturbative) asymptotic series of $W(g,N)$ at $g=0$ is a sum over connected Feynam graphs. The connection between the LVE expansion of $W(g,N)$ presented in 
Proposition~\ref{prop:LVE+bounds} and the Feynman graphs is discussed in Appendix~\ref{app:FeynGr}. 
On the other hand, the power series in each multi-instanton sector of the transseries of $W(g,N)$ has no simple diagrammatic interpretation; they can be constructed from the nonlinear differential equation obeyed by $W(g,N)$, or more formally by expanding the logarithm of the transseries expansion of $Z(g,N)$ in powers of the transseries monomial $\exp\{3/(2g)\}$ (e.g.\ \cite{Aniceto:2018bis}). The latter is however only a meaningful operation in the sense of formal power series.

In this section we take a different route and derive rigorously the transseries expansion of $W(g,N)$ by exploiting the analytical control we have on the small-$N$ expansion. We first notice that, from Propositions \ref{prop:Zn} and \ref{prop:LVE+bounds}, $W_n(g)$ and $Z_n(g)$ are analytic functions on the extended Riemann sheet $\mathbb{C}_{3\pi/2}$. Next, we use Eq.~\eqref{eq:Moeb} to construct $W_n(g)$ as a finite linear combination of finite products of $Z_i(g)$'s. Each such product is in fact a (factored) multidimensional integral, hence we can apply to it the steepest descent method to obtain its asymptotic expansion. In $\mathbb{C}_\pi$, the asymptotic expansion of each factor $Z_i(g)$ is of the perturbative type, Eq.~\eqref{eq:Znpertseries}, and $W_n(g)$ is just a finite linear combination of Cauchy products of such series.

When turning $g$ past the negative real axis, each integration contour in this multidimensional integral must be deformed past a cut and each $Z_{i\pm}(g) = Z_i^{\mathbb{R}}(g) + Z^C_{i\pm}(g)$ (see the disucussion below Proposition \ref{prop:Zn}).
It follows that $W_i(g)$ is a linear combination of products involving $Z_i^{\mathbb{R}}(g)$'s and $Z^C_{i\pm}(g)$, and this representation holds in the sense of analytic functions on the Riemann surface. In order to obtain the transseries of $W_n(g)$, one needs to build the transseries expansion of each of the terms in the linear combinations. As the multidimensional integrals are factored, this is just the Cauchy product of the transseries $Z_i^{\rm pert.}(g)$ and $e^{\frac{3}{2g}}Z_i^{(\eta)}(g)$ corresponding to $Z_i^{\mathbb{R}}(g)$ and  $Z^C_{i\pm}(g)$, respectively.

The summation over $n$ is more delicate, as it is an infinite series.
As we have seen in Proposition~\ref{prop:LVE+bounds}, the small-$N$ series of $W(g,N)$ converges in the domain $\mathbb{D}_0$ of Eq.~\eqref{eq:W-domain}, thus yielding $W(g,N)$ in terms of $W_n(g)$ as an analytic function on such domain. Therefore, we can apply the steepest descent method term by term to the small-$N$ series, and hence
the transseries of $W(g,N)$ is rigorously reconstructed  by substituting the transseries for $W_n(g)$ in Eq.~\eqref{eq:ZW}. 

 Unsurprisingly, at the end we recover the formal transseries of $W(g,N)$ which can be obtained by direct substitution of the transseries expansion of $Z(g,N)$,  taking formally its logarithm, and then expanding in powers of $Z_i^{(\eta)}(g)$ and $Z_i^{\rm pert.}(g)-1$. What we gained in the process is that we replaced a formal manipulation on transseries with a rigorous manipulation on analytic functions.

\begin{proposition}\label{prop:TSW}
The cumulant $W_n(g)$ and the full free energy $W(g,N)$ have transseries expansions, that can be organized into instanton sectors. The instanton counting parameter is denoted by $p$.
\begin{enumerate}
\item\label{propTSW:1} For $g\in\mathbb{C}_{3\pi/2}$,, the cumulant $W_n(g)$ has the transseries expansion:
\begin{equation} \label{eq:Wn-trans}
    W_n(g)= \sum_{p=0}^n e^{\frac{3}{2g}p}
    \;\Big(\eta\sqrt{2\pi} \sqrt{\f{g}{3}}\Big)^{p}\; \sum_{l'= 0}^{n-p} \left(\ln\left(\tfrac{g}{3}\right)\right)^{l'} \sum_{l\geq 0} g^l\, 
    W^{(p)}_{n;l,l'} \,,
\end{equation}
where $\mathbb{R}_-$ is a Stokes line, $\tau = -\sgn(\varphi)$ and $\eta$ is a transseries parameter which is zero on the principal Riemann sheet and is one when $|\varphi|>\pi$.
The $g$-independent coefficient $W^{(p)}_{n;l,l'}$ is given by the following nested sum:
\begin{equation}\label{eq:Prop5coeff}\begin{aligned}
W^{(p)}_{n;l,l'} = &
\sum_{\substack{k = p \\ k+p\geq 1}}^{n} (-1)^{k-1} (k-1)! \sum_{ \substack{n_1, \ldots, n_{n-k+1} \ge 0 \\ \sum in_i = n ,\, \sum n_i = k }} \ \sum_{\substack{\{0\leq p_i\leq n_i\} \\ _{i=1,\dots,n-k+1} \\ \sum p_i=p}}
\frac{n!}{\prod_i (n_i-p_i)!p_i! (i!)^{n_i}}
\\ & \quad
\sum_{\substack{\{a^i_j\geq 0\}^{i=1,\dots,n-k+1}_{j=1,\dots,n_i} \\ \sum_i\sum_j a^i_j=l}} \ \sum_{\substack{\{0\leq c^i_j\leq i-1\}^{i=1,\dots,n-k+1}_{j=1,\dots,p_i} \\ \sum_i\sum_j c^i_j=l'}}
\left(\f{1}{6}\right)^{\sum_{i=1}^{n-k+1}\sum_{j=1}^{p_i}a^i_j} \left(-\f{2}{3}\right)^{\sum_{i=1}^{n-k+1}\sum_{j=p_i+1}^{n_i}a^i_j} 
\\ & \qquad\qquad\qquad\qquad\qquad\qquad
\prod_{i=1}^{n-k+1} \Bigg(\prod_{j=1}^{p_i} G(a^i_j,c^i_j;i)\Bigg)\Bigg(\prod_{j=p_i+1}^{n_i} G(a^i_j;i)\Bigg)
\; ,
\end{aligned}\end{equation}
with
\begin{equation}
\begin{aligned}
    & G(a;i) = \f{(2a)!}{2^{2a} a!} \sum_{\substack{a_1,\ldots,a_{2a-i+1} \ge 0  \\ \sum k  a_k=2a , \; \sum a_k=i}} \f{(-1)^i i!}{\prod_k k^{a_k}a_k!} \; , \\
    & G(a,c;i) = \sum_{b=0}^{i-1} \left(\imath\tau 2\pi \right)^{i-1-b-c} \f{i!}{a!\,b!\,c!\,(i-b-c)!} \f{d^b \G(z)}{dz^b}\Big|_{z=2a+1} \;.
\end{aligned}
\end{equation}

\item\label{propTSW:2}  For $g\in\mathbb{D}_{\theta}$, the full free energy $W(g,N)$ has the transseries expansion:
\begin{equation}
\begin{split}
    W(g,N) & = \sum_{n\ge 1} \frac{1}{n!} \left( -\frac{N}{2}\right)^n W_n(g) \crcr 
    & = \sum_{p\ge 0} e^{\frac{3}{2g}p} \; \Big(\eta\sqrt{2\pi}e^{\imath\tau\f{\pi}{2}} \left(\frac{e^{\imath\tau\pi} g}{3}\right)^{\frac{1-N}{2}}\Big)^{p} \; \sum_{l\ge 0} \left(- \frac{2g}{3}\right)^{l} W^{(p)}_{l}(N) \; ,
\end{split}
\end{equation}
where
\begin{equation}
\begin{split}
    W^{(p)}_{l}(N) =& \sum_{\substack{q\geq 0\\ p+q\geq 1}} (-1)^{p+q-1} \frac{(p+q-1)!}{p!q!} \\ & \qquad
\sum_{\substack{n_1,\dots,n_q\geq 1 \\ m_1,\dots,m_p\geq 0 \\ \sum n_i+\sum m_j=l }}
\Bigg( \prod_{i=1}^{q} \frac{\Gamma(2n_i+N/2) }{2^{2n_i}n_i! \, \Gamma(N/2) }\Bigg)
\Bigg( \prod_{j=1}^{p} \frac{(-1)^{m_j}}{ 2^{2m_j} m_j! \; \Gamma(\frac{N}{2} -2m_j)}\Bigg) \;.
\end{split}
\end{equation}
\end{enumerate}
\end{proposition}
\begin{proof} See Appendix~\ref{app:Trans}
\end{proof}

While the expressions in Proposition \ref{prop:TSW} are not the most amenable to computations, one feature is striking. Expanding the cumulant $W_n(g)$ into $p$ instanton sectors, we observe that only the first $n$ instantons contribute to $W_n$, that is the sum in Eq.~\eqref{eq:Wn-trans} truncates to $p=n$. The $n$ instanton contribution to $W_n$ comes from $n=p=k$ in Eq.~\eqref{eq:Prop5coeff} which implies $n_1=n$ and all the others $0$, hence:\footnote{This formula is most easily derived from Eq.~\eqref{eq:Prop5proof1}.}
\begin{equation}
    W_n^{(n)}(g) \simeq e^{\f{3}{2g} n } \left(\eta\sqrt{2\pi} \sqrt{\f{g}{3}}\right)^n (-1)^{n-1} (n-1)! \left( \sum_{q=0}^{\infty}  \f{(2q)!}{q!} \left(\f{g}{6}\right)^q \right)^n \;.
\end{equation}

This is genuinely new phenomenon. Usually, for quantities that are interesting for physics, one either deals with functions having only one instanton, like $Z(g,N)$ (or $Z_n(g)$) or with function receiving contributions from all the instanton sectors, like $W(g,N)$. This is, to our knowledge, the first instance when some physically relevant quantity receiving contributions from a finite number of instantons strictly larger than one is encountered.

\subsection{Differential equations}

The exotic behaviour of $W_n(g)$ can also be understood in terms of differential equations. By rewriting the partition function as $Z(g,N)=e^{W(g,N)}$ it is straightforward to turn (\ref{PDEpartitionfunctionN}) into a differential equation for $W(g,N)$, which in turn implies a  tower of equations for $W_n(g)$.

\begin{proposition}\label{prop:DEW}
The function $W(g,N)$ obeys the non linear differential equation:
\begin{equation}
    16 g^2 W^{\prime \prime}(g,N) + 16g^2\left(W^{\prime}(g,N)\right)^2+ \left((8N+24)g+24 \right)W^\prime(g,N)+N(N+2)=0 \; .
\end{equation}

The functions $W_n(g)$ obey the tower of differential equations:
\begin{equation}\label{eq:TowerDiffEqWn}
    \begin{split}
        & 4 g^2 W_1^{\prime\prime}(g)+6(g+1)W_1^\prime(g)-1=0, \\
        & 4 g^2 W_2^{\prime\prime}(g)+6(g+1)W_2^\prime(g)+8g^2\left(W_1^\prime(g)\right)^2 -8g W_1^\prime(g)+2=0, \\
        & 4 g^2 W_n^{\prime\prime}(g)+6(g+1)W_n^\prime(g)+4 g^2 \sum_{k=1}^{n-1}\binom{n}{k}W_{n-k}^\prime W_{k}^{\prime}-4n g  W_{n-1}^\prime(g)=0 \;.
    \end{split}
\end{equation}

\end{proposition}

The differential equation for $W_{1}(g)$ is, unsurprisingly, identical to the one for $Z_1(g)$ (the connected $1-$point function equals to the full $1-$point function). Note that although the differential equation for $W(g,N)$ is non linear, the one for $W_n(g)$ is linear. In fact, since $W_0(g)=0$, the non linear term $\left(W^\prime(g,N)\right)^2$ produces only source terms in (\ref{eq:TowerDiffEqWn}). The linearity of the equations provides another point of view on why only a finite number of instantons arise in each $W_n(g)$.

\section*{Acknowledgements}
R. G., H. K. and D. L. have been supported by the European Research
Council (ERC) under the European Union's Horizon 2020 research and innovation program (grant agreement No818066) and by the Deutsche Forschungsgemeinschaft (DFG, German Research Foundation) under Germany's Excellence Strategy EXC--2181/1 -- 390900948 (the Heidelberg STRUCTURES Cluster of Excellence). We thank Ricardo Schiappa for comments on an early version of this work.

\newpage

\appendix

\section{Asymptotic expansions}
\label{app:expansions}

Analytic continuations of exponential integrals like $Z(g,N)$ can be carried out by deforming the original real integration cycles in the complex plane. Complex Morse theory (Picard-Lefschetz-Theory) provides a systematic framework for decomposing the original integration cycle $\cC$ into a sum of more convenient cycles $\cJ_i$ called Lefschetz thimbles:
\be \label{eq:thimbles}
I(g)=\int_{\cC}dx\; e^{ \frac{1}{g} f(x)} a(x) = \sum_{i} \int_{\cJ_i}dx\; e^{\frac{1}{g} f(x)} a(x)\; .
\ee 
Generically, each thimble intersects one critical (or saddle) point $x^\ast_i\in\mathbb{C}^m$ and consists of the union of downward flows with respect to the real part of $f(x)$ originating at $x^\ast_i$. From a topological point of view $\{\cJ_i\}$ generate the $m$'th relative homology group of the underlying $2m$-dimensional space. Crucially, the imaginary part of $f(x)$ is constant along each $\cJ_i$ and the integral along a thimble is absolutely convergent. Since the individual integrals are non-oscillating, it is possible to apply Laplace's method to each term, expanding the integrand around the critical points:
\[
I(g)=\sum_{i} e^{\frac{1}{g} f(x^\ast_i)} \Phi^{(i)}(g) \; ,
\] 
where $\Phi^{(i)}(g)$ is an asymptotic series, possibly containing logarithms and non-integer powers of $g$.
As $g$ is varied, the thimbles are deformed and the number of thimbles appearing in the decomposition of the original contour $\cC$ may vary discontinuously. These discrete changes, happening at values of $g$ for which different thimbles intersect each other, are connected to the so called Stokes jumps.

Lefschetz thimble techniques are standard tool in resurgence analysis \cite{Sternin1997,Dunne:2015eaa,Aniceto:2018bis} and have many other applications to path integrals \cite{Witten:2010zr,Fujii:2013sra,Aarts:2014nxa,Bluecher:2018sgj,Tanizaki:2014tua,Tanizaki:2017yow}. More details and a nice pedagogical introduction can be found in \cite{Witten:2010cx}.

In the following we review the derivation of the asymptotic expansions around the critical points of the zero-dimensional $\phi^4$ theory for $N=1$.

\subsection{The $\phi$ representation of the partition function}
\paragraph{The vacuum expansion.}
Our starting point is the partition function of the model in the $\phi$ representation. We set $N=1$ and consider:
\begin{equation} \label{eq:Z_Phi4N1}
Z(g)=\int_{-\infty}^{+\infty} [d\phi]  \; e^{-S[\phi] } \;, \qquad S[\phi]=\frac{1}{2}\phi^2+\frac{g}{4!} \phi^4 \; ,
\end{equation}
where again we set $[d\phi]=d\phi/\sqrt{2\pi}$.
There are three solutions of the equations of motions, i.e.~critical points $S'[\phi_\star]=0$, namley $\phi_0=0$ and $\phi_{\pm}=\pm i \sqrt{\frac{6}{g}}$, and each of them has an attached thimble, $\cJ_0$ and $\cJ_\pm$ (see Fig.~\ref{fig:phi-thimble}). For symmetry reasons it is clear that $Z_{\mathcal{J}_+}(g)=Z_{\mathcal{J}_-}(g)$, thus there are only two asymptotic expansions.

Eq.~\eqref{eq:Z_Phi4N1} is absolutely convergent if $g$ is in the right half complex plane $\mathrm{Re}(g)>0$, hence in this domain it defines an analytic function $Z(g)$. In order to analytically continue this function we turn in the complex plane and parameterize $g=|g|e^{i \varphi}$ but we tilt the contour of integration by $e^{-i \theta}$ to compensate. In detail, we define:
\begin{equation} \label{eq:Zphi-tilt}
Z_\theta(g)=\int_{ \mathbb{R} e^{-i \theta} } [d\phi] \  e^{-\frac{1}{2}\phi^2-\frac{|g| e^{i \varphi}}{4!}\phi^4}=e^{-i\theta}\int_{ \mathbb{R} }[d\phi] \  e^{-\frac{1}{2}\phi^2 e^{-2 i \theta}-\frac{|g| e^{i (\varphi-4\theta)}}{4!}\phi^4} \; ,
\end{equation}
which is absolutely convergent if both $\varphi -4\theta \in (-\pi/2,\pi/2)$ and $-2\theta \in (-\pi/2,\pi/2)$, and is independent of $\theta$ as long as it converges. As $Z_0(g) = Z(g)$, it follows that $Z_\theta(g)$ is the analytic continuation of $Z(g)$ and it is easy to check that 
the integral in Eq.~\eqref{eq:Zphi-tilt} defines it for all $-3\pi/2\leq \varphi<3\pi/2$. We denote:
\begin{equation}
Z_+(g) = Z_{\theta}(g) \text{ for }\theta>0 \;,\qquad Z_-(g) = Z_{\theta}(g) \text{ for }\theta < 0 \; ,
\end{equation}
the anticlockwise respectively clockwise analytic continuations. Both $Z_+(g)$ and $Z_-(g)$ are defined for any $g$ with $\mathrm{Re}(g)<0$, but they are not equal. That is $Z_\theta(g)$ is a multi-valued function in the complex $g$-plane, with a branch point at the origin. We chose the range $-\pi < \varphi<\pi$ for the principle Riemann sheet, with a cut along the negative real axis. 

For $|\varphi|<\pi$ the integration contour in Eq.~\eqref{eq:Zphi-tilt} is homotopic to just the perturbative thimble $\mathcal{J}_0$, which at the origin is tangent to the real axis. In this case, the Laplace method applied to $Z_{\mathcal{J}_0}(g)$ amounts to Taylor expanding the quartic interaction and computing the Gaussian integral:\footnote{As explained for example in \cite{Bender:1999}, in the Laplace method we restrict the integration to a small neighborhood of the saddle point, we expand  the integrand, keeping only the first nontrivial term in the exponent while expanding the rest, and lastly exchange sum and integral, extending the integration domain to an infinite line along which the integrals are convergent and computable.}
\begin{equation}\label{eq:Zpert_phi}
\begin{aligned}
Z(g)&=Z_{\mathcal{J}_0}(g)=\int_{\mathcal{J}_0} [d\phi] e^{-\frac{1}{2}\phi^2-\frac{g}{4!}\phi^4}= \sum_{n=0}^{\infty}  \frac{1}{n!} \left(-\frac{g}{4!}\right)^n \int_{-\infty}^{\infty} [d\phi] e^{-\frac{1}{2}\phi^2} \phi^{4n} \\
&=\sum_{n=0}^{\infty}\left(-\frac{2}{3}\right)^n \frac{(4n)!}{2^{6n}(2n)!n!}g^n\equiv\sum_{n=0}^{\infty}A_n^{\rm pert.}g^n \;.
\end{aligned}
\end{equation}

\paragraph{The instanton sector.}
At $\varphi=\pm\pi$, i.e.\ at $g<0$, the perturbative thimble intersects the instanton thimbles. At $|\varphi|>\pi$ they split again, but on the opposite side, so that the perturbative thimble effectively has a jump at $|\varphi|=\pi$, leading to a jump in the decomposition of the original contour $\cC$. 

As the evaluation of of the integrals along the instanton thimbles, $Z_{\cJ_\pm}$, is well behaved and continuous across the Stokes line at $\varphi=\pi$,
for the asymptotic expansion around the instanton, we consider the case $g<0$. In this case the thimbles $\cJ_\pm$ are described by $\mathrm{Re}(\phi)=\pm \sqrt{(\mathrm{Im}(\phi))^2+\frac{6}{|g|}}$. Since we know the analytic expression of the thimbles, we take an explicit parametrization of the curve and use it to compute the integral. For example we can chose:
\begin{equation}
\gamma_\pm: \ t\in(-\infty, \infty) \rightarrow \left(\pm \sqrt{t^2+\frac{6}{|g|}}\right)+i \, t \;.
\end{equation}
and the integral reads:
\begin{equation}
\begin{split}
	Z_{\mathcal{J}_+} &=\int_{\mathcal{J_+}}[d\phi] e^{-S(\phi)}=\int_{-\infty}^{\infty}[dt]\; \dot{\gamma}_+(t)e^{-S\left(\gamma_+(t)\right)} \\
	&=\int_{0}^{\infty}[dt]\; \dot{\gamma}_+(t)e^{-S\left(\gamma_+(t)\right)}+\int_{0}^{\infty}[dt]\; \dot{\gamma}_+(-t)e^{-S\left(\gamma_+(-t)\right)} \;.
\end{split}
\end{equation}
The imaginary part of the action is zero along the thimbles. Also we have that $\dot{\gamma}_\pm(-t)=\frac{\mp t}{\sqrt{t^2+\frac{6}{|g|}}}+i$, and thus the real parts of the two integrals cancel. In the end we find:
\begin{equation}\label{eq:niceresult}
		Z_{\mathcal{J}_+}=2i\int_{0}^{\infty} [dt]\, e^{-\frac{3}{2|g|}-t^2-\frac{|g|}{6}t^4}
		=i e^{\frac{3}{2 g}} \int_{-\infty}^{+\infty} [dt]\, e ^{-t^2+\frac{g}{6}t^4} \;.
\end{equation}
Now we can rescale the $t$ by $\frac{1}{\sqrt{2}}$ and find the same integral as for $\cJ_0$ but with the opposite sign for $g$. In the end we have:
\begin{equation}\label{eq:niceresult1}
		Z_{\mathcal{J}_\pm}=\frac{i}{\sqrt{2}} e^{\frac{3}{2g}} \sum_{p=0}^{\infty}(-1)^p A_p^{\rm pert.}g^p \;.
\end{equation}

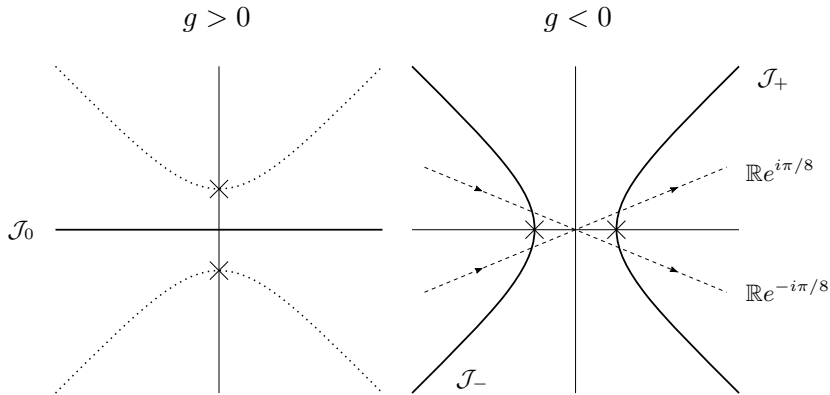
\begin{figure}
    \centering
    \begin{tikzpicture}[font=\footnotesize]
        \node at (0,0) {\includegraphics[scale=.5]{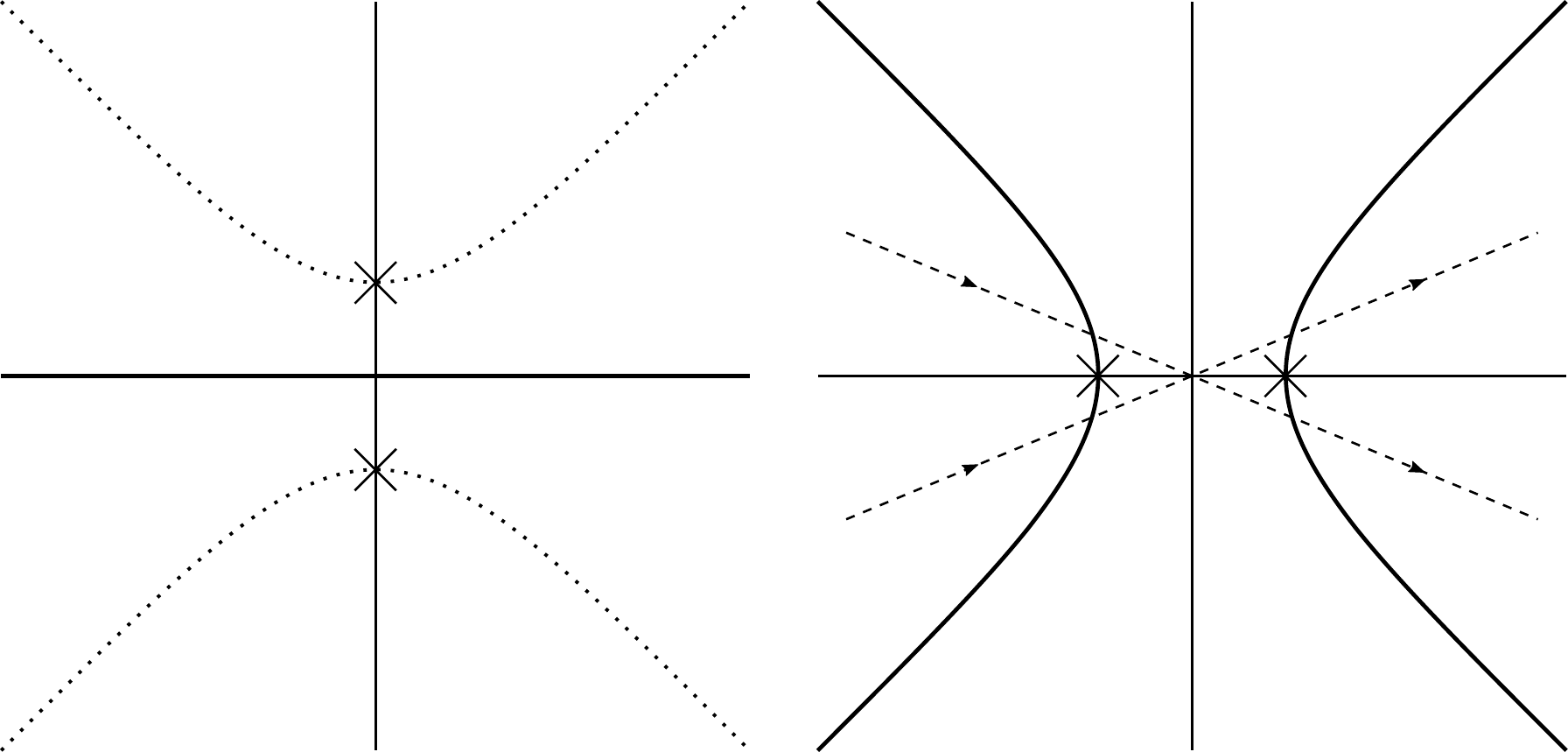}};
        \node at (-2.4,2.8) {\normalsize$g>0$};
        \node at (2.4,2.8) {\normalsize$g<0$};
        \node at (-5,0) {$\cJ_0$};
        \node at (5,2) {$\cJ_+$};
        \node at (1,-2) {$\cJ_-$};
        \node[anchor=west] at (4.5,0.8) {$\mathbb{R}e^{i\pi/8}$};
        \node[anchor=west] at (4.5,-0.8) {$\mathbb{R}e^{-i\pi/8}$};
    \end{tikzpicture}
    \caption{Critical points and thimbles (thick lines) in the complex $\phi$-plane. The crosses mark the positions of the instantons and the dashed lines are the tilted contours of Eq.~\eqref{DiscontinuityNegativeAxis}. 
    }
    \label{fig:phi-thimble}
\end{figure}

\paragraph{Discontinuity.}

The discontinuity at the cut can be computed as:
\begin{equation}\label{DiscontinuityNegativeAxis}
Z_-(-|g|)-Z_+(-|g|)=\int_{\mathbb{R}e^{i\pi/8}}\f{d\phi}{\sqrt{2\pi}} \  e^{-\frac{1}{2}\phi^2+\frac{|g|}{4!}\phi^4}-\int_{\mathbb{R}e^{-i\pi/8}} \f{d\phi}{\sqrt{2\pi}} \  e^{-\frac{1}{2}\phi^2+\frac{|g|}{4!}\phi^4} \;,
\end{equation}
where we tilted the contours by the $\theta$ of minimal absolute value which ensures convergence. Each of the two contours can be deformed to the perturbative thimble alone for $|\varphi|<\pi$.
However, in the limit of $|\varphi|\nearrow \pi$ the difference of the two perturbative thimbles approaches the sum of the two instanton thimbles (see Fig.~\ref{fig:HalfThimbles}), leading to the asymptotic expansion in Eq.~\eqref{eq:discZ}, with $N=1$.

\begin{figure}[!htb]
	\centering
	\begin{tikzpicture}[font=\footnotesize]
        \node at (0,0) {\includegraphics[scale=0.5]{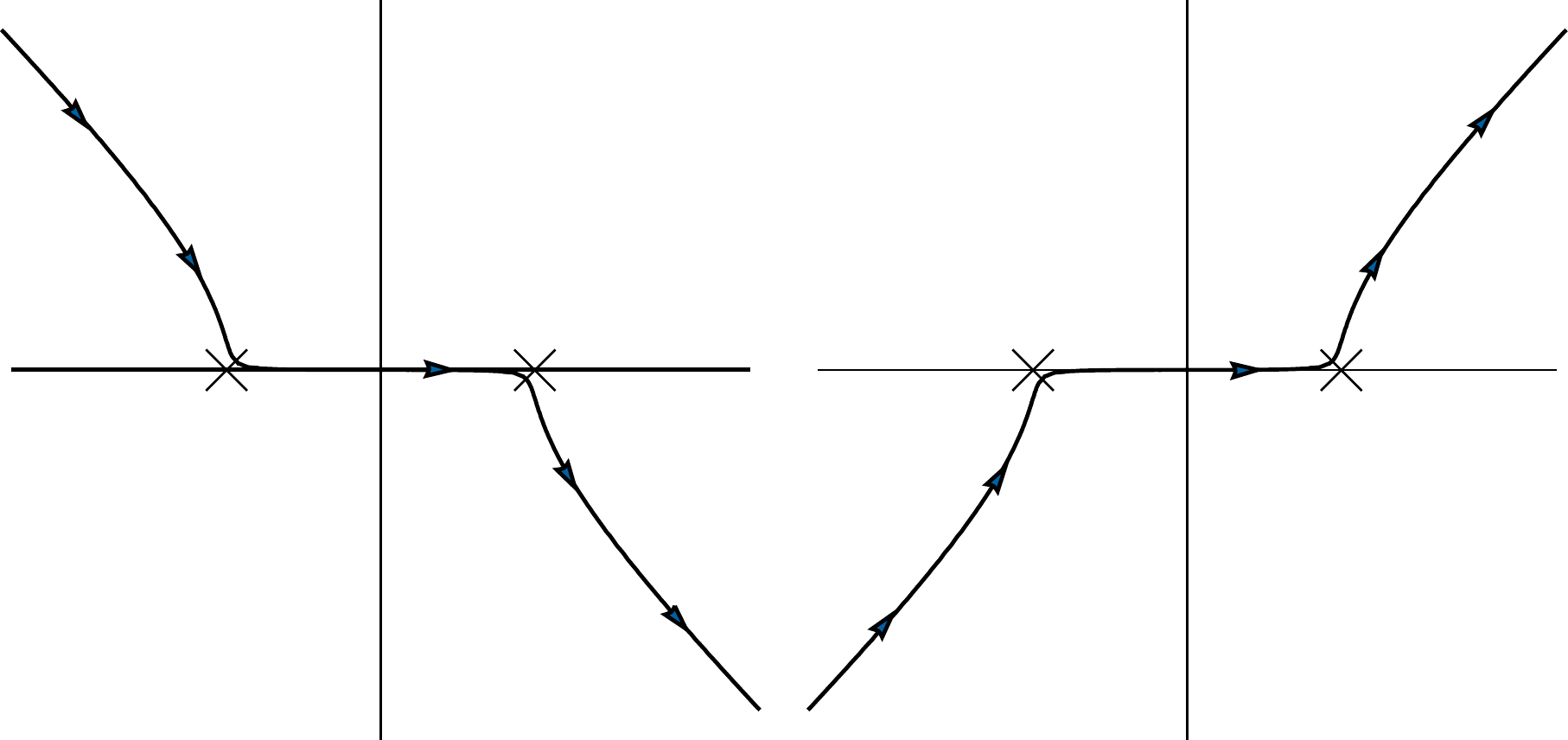}};
        \node at (-2.4,2.8) {\normalsize$\varphi\nearrow\pi$};
        \node at (2.4,2.8) {\normalsize$\varphi\searrow -\pi$};
    \end{tikzpicture}
	\caption{The thimble $\mathcal{J}_{0}$ for $Z_{+}(-|g|)$ (left) and $Z_{-}(-|g|)$ (right) as $|\varphi|\nearrow \pi$.}
	\label{fig:HalfThimbles}
\end{figure}

\section{A simple generalization of the Nevanlinna-Sokal theorem}
\label{app:Sokal}

\begin{proof}[Proof of Theorem~\ref{thm:Sokal}] The proof is an infinitesimal variation of the proof of \cite{Sokal:1980ey} which deals with the case 
$\beta=0$. We parallel the notation of \cite{Sokal:1980ey}. 

\bigskip

We first observe that:
\begin{equation}
|a_q | = \lim_{z\to 0, z\in {\rm Disk}_R} \left| \frac{ a_q z^q +  R_{q+1} (z) }{z^q} \right| = \lim_{z\to 0, z\in {\rm Disk}_R} \frac{ | R_q(z) |}{|z|^q} \le
K \, q! \;q^\beta \;  \rho^{-q}  \;, 
\end{equation}
hence $B(t)$ is an integer power series which converges in the disk $|t|<\rho$ and defines an analytic function in this domain.

\bigskip

Let us recall Hankel's contour integral representation of the inverse of the Gamma function:
\begin{equation}
    \frac{ 1 }{2\pi \imath} \oint_{\mathrm{Re}(z^{-1}) =r^{-1}} dz  \; e^{x/z} z^{ k - 1}
     = \frac{ 1 }{2\pi \imath} \int_{r^{-1}-\imath \infty}^{r^{-1}+\imath \infty} dw \; e^{x w} \; w^{-k -1}
      = \frac{x^k} {\Gamma(k+1)} \; ,
\end{equation}
which holds for any $r,x\in \mathbb{R}_+$ and $k\in \mathbb{C}$. We define for $x\in \mathbb{R}_+$ the function:
\begin{equation}
 b_0(x)  =  \frac{ 1  }{2\pi \imath} \oint_{\mathrm{Re}(z^{-1}) =r^{-1}} dz \; e^{x/z} \; z^{ -1} \; f(z) \;,
\end{equation}
which is an $r$ independent function for $r<R$ as long as the integral converges because the integration contour is fully contain in ${\rm Disk}_R$. Substituting the asymptotic expansion of $f(z)$ up to order $q$ we obtain:
\begin{equation}\label{eq:curest}
    b_0 (x) = \sum_{k=0}^{q-1} \frac{a_k}{ k ! } x^k + \frac{ 1 }{2\pi \imath} \oint_{\mathrm{Re}(z^{-1}) = r^{-1}} dz \; e^{x/z} \; z^{-1}  R_q(z) \; .
\end{equation}
Changing variables to $w=1/z$ and using the bound on $R_q$, the reminder term above is bounded by:
\begin{equation}
\begin{split}
& K  \; q! \; q^{\beta} \; \rho^{-q} \; e^{x/r}  \; \int_{-\infty}^{\infty} dv \; \frac{1}{  | r^{-1} +\imath v| ^{q + 1} }  \le K   \; q! \; q^\beta \; \rho^{-q} \; e^{x/r} \; r^{q} 
 \; \int_{-\infty}^{\infty} d v  \frac{1}{ (1+v^2)^{\frac{q+1}{2}} }\;  ,
 \end{split}
\end{equation}
and the integral in the last line is always bounded by the case $q=1$ in which case it is $\pi$ and can be absorbed in $K$. Choosing $r = x/ q$ (which is possible for $q$ large enough $q> x/R $) and using the Stirling upper bound on the Gamma function, the reminder is finally bounded by $K \;  q^{\beta + 1/2} (x/\rho)^q  $ hence goes to zero in the $q\to \infty$ limit as long as $x<\rho$. It follows that for $0<x<\rho$, $b_0(x) = B(t)|_{t=x}$.

\bigskip

We now define $b_m(x) = \frac{d^m}{dx^m} b_0(x) $, and using Eq.~\eqref{eq:curest} with $q=m+1$ we get:
\begin{align}
 b_m(x) =a_m  +     \frac{ 1 }{2\pi \imath} \oint_{\mathrm{Re}(z^{-1}) = r^{-1}} dz \; e^{x/z} z^{-m-1}  R_{m+1}(z) \;, 
 \end{align}
and we have the bound $ |b_m(x)|\le K \; (m+1)! \; (m+1)^\beta  \; \rho^{-m-1} e^{x/r} $ (note that this bound covers also the term $|a_m|< K \; m! \; m^{\beta} \; \rho^{-m}$).

 \bigskip
 
The sum $B_x(t)=\sum_{m\ge 0}\frac{(t-x)^m}{m!}  b_m(x)$ defines an analytic function in $t$ as long as it converges. As:
\begin{equation}
|B_x(t)| \le K e^{x/r} \sum_{m\ge 0} (m+1)^{\beta+1} ( |t-x| / \rho)^m \;,
\end{equation}
we conclude that $B_x(t)$ is analytic in a disk or radius $\rho$ centerd at $x$. It is immediate to check that $B_x(t)=B_{x'}(t)$ as long as they both converge, hence $B_x(t)$ is the analytic continuation of $B(t)$ to the strip $ \{ t\in \mathbb{C} \mid  {\rm dist}(t,\mathbb{R}_+ ) < \rho \}$ and it obeys the appropriate exponential bound. The last point follows by noting that for $z'\in {\rm Disk}_r$:
\begin{equation}
 \frac{1}{z' }\int_{0}^{\infty} dx \; e^{-x/z'} \; \frac{ 1 }{2\pi \imath} \oint_{\mathrm{Re}(z^{-1}) =r^{-1}} dz \; e^{x/z} \; z^{-1} \; f(z)  = \frac{ 1 }{ 2\pi \imath} \oint_{\mathrm{Re}(z^{-1}) =r^{-1}} dz \; \frac{1}{z-z'}\; f(z)  = f(z') \;,
\end{equation}
by Cauchy's theorem.
\end{proof}

\section{Proofs of Propositions}

In this appendix we gather the proofs of various Propositions in the main body of the paper.

\subsection{Properties of $Z(g,N)$}
\label{app:proofZ}

\begin{proof}[Proof of Proposition~\ref{prop:Z}]The proof of this proposition is linear.

\paragraph{Property~\ref{propZ:1}.} This follows by bounding the square root and taking into account that the Gaussian integral is normalized to 1. For $N\ge 0$ and $g = |g| e^{\imath\varphi }$ with $-\pi < \varphi < \pi$, we have the uniform bound: 
\begin{equation}
|1 -\imath \sqrt{\tfrac{g}{3}}   \sigma  |  = | e^{-\imath \frac{\varphi}{2} } - \imath  \sqrt{\tfrac{|g|}{3}}   \sigma | \ge \cos \frac{ \varphi}{2} \; .
\end{equation}
For $N<0$ we use $|1 -\imath \sqrt{\tfrac{g}{3}}   \sigma  | \le 1 +  \sqrt{ \frac{|g|}{3}} |\sigma| $ and splitting the integration interval in regions where $ \sqrt{ \frac{|g|}{3} }  |\sigma| \le 1$ respectively $\sqrt{ \frac{|g|}{3} } |\sigma| \ge 1$ and then re-extending the integration intervals to cover $(-\infty,\infty)$ we get: 
\begin{equation}
|Z(g,N)|
\le  \int_{ -\infty}^{\infty} [d\sigma]  \;e^{-\frac{1}{2} \sigma^2} (1 + \sqrt{ \tfrac{|g|}{3} }  |\sigma| )^{|N|/2} \le 
2^{|N|/2}+\f{2^{3|N|/4}}{\sqrt{\pi}} \frac{|g|^{N/4}}{3^{|N|/4}} \Gamma( \tfrac{|N|+2}{4}  ) \;.
\end{equation}

\paragraph{Property~\ref{propZ:2}.}
The perturbative expansion is obtained using $(1- x)^{-N/2} =\sum_{q\ge 0} \binom{q+N/2-1}{q}x^q$ and commuting (formally) the sum and integral:
\begin{equation}
Z(g,N)  =\sum_{n=0}^{\infty}\binom{2n + \frac{N}{2}-1}{2n} \left( - \f{g}{3} \right)^n \int_{\mathbb{R}} [d\sigma]  \;e^{-\frac{1}{2} \sigma^2} \sigma^{2n}= \sum_{n=0}^{\infty} \;  \frac{\Gamma(2n+N/2) }{2^{2n}n! \; \Gamma(N/2) }    \; \left( - \frac{2g}{3}\right)^n  \;.
\end{equation}
For the case $N=1$ for instance we have $\frac{ \Gamma(2n+ 1/2)}{\Gamma(1/2)} = \frac{(4n)!}{4^{2n}(2n)!}$ (see also Eq.~\eqref{eq:Zpert_phi} in Appendix~\ref{app:expansions}). 

\paragraph{Property~\ref{propZ:3}.}
Properties~\ref{propZ:3},~\ref{propZ:4},~\ref{propZ:4.5} and~\ref{propZ:5} are closely related and require that we deal carefully with the integration contour. We define:
\begin{align}
Z_\theta(g,N)=\int_{ e^{  -\imath \theta} \mathbb{R}} [d\sigma]  \;e^{-\frac{1}{2} \sigma^2} \frac{1}{\left( 1 - \imath \sqrt{ \frac{g}{3} } \sigma \right)^{N/2} }       = 
\int_{ \mathbb{R}} e^{-\imath \theta}[d\sigma]  \;e^{-\frac{1}{2}e^{-2\imath\theta} \sigma^2} \frac{1}{\left( 1 - \imath \sqrt{ \frac{|g|}{3} }e^{\imath \frac{\varphi-2\theta}{2}} \sigma \right)^{N/2} }    \;,
\end{align}
which is absolutely convergent if both $\varphi -2\theta \in (-\pi,\pi)$ and $-2\theta \in (-\pi/2,\pi/2)$. Moreover,  as long as it converges, it is independent of $\theta$, as can be verified by noticing that the derivative with respect to $\theta$ can be rewritten as the integral of a total derivative in $\sigma$. Thus $Z_{\theta}(g,N)$ is the analytic continuation of $Z(g,N)$ and optimizing on $\theta$ the partition function can be continued to the extended Riemann sheet $\mathbb{C}_{3\pi/2}$ with a branch point at $0$.

\bigskip

Using a Taylor formula with integral rest we have $ Z_\theta(g,N) - \sum_{k=0}^{q-1} \frac{1}{k!} Z_{\theta}^{(k)}(0,N) g^k  = R_{\theta}^q(g,N) $ with:
\begin{align}
R_{\theta}^q(g,N)   & =  \int_0^1 du \; \frac{(1-u)^{2q-1}}{(2q-1)!}  \;
 \int_{ \mathbb{R}} e^{-\imath \theta}[d\sigma]  \;e^{-\frac{1}{2}e^{-2\imath\theta} \sigma^2} 
  \left(\frac{d}{du} \right)^{2q}\left( \frac{1}{\left( 1 - \imath \sqrt{ \frac{|g|}{3} }e^{\imath \frac{\varphi-2\theta}{2}} \sigma u \right)^{N/2}  }  \right) \\
 & = \int_0^1 du \; \frac{(1-u)^{2q-1}}{(2q-1)!} 
 \int_{ \mathbb{R}} e^{-\imath \theta} [d\sigma]  \;e^{-\frac{1}{2}e^{-2\imath\theta} \sigma^2}  
 \frac{ \left( - \imath \sqrt{ \frac{|g|}{3} }e^{\imath \frac{\varphi-2\theta}{2}} \sigma  \right)^{n} 
   (-1)^{n} \frac{ \Gamma(n + N/2  ) }{  \Gamma(N/2)}  
 }
 { \left(1 - \imath \sqrt{ \frac{|g|}{3} }e^{\imath \frac{\varphi-2\theta}{2}} \sigma u \right)^{ 
 \frac{ 2n+N }{ 2 } } } \Bigg{|}_{n=2q}  \; , \nonumber
\end{align}
where for $N<0$ we need to chose $q > -N/4$.
Using
$| 1 - \imath \sqrt{ \frac{|g|}{3} }e^{\imath \frac{\varphi-2\theta}{2}} \sigma u | \ge \cos\frac{\varphi-2\theta}{2} $ the rest term is bounded as:\footnote{The factor $1/(\cos(2\theta))^{q}$ can be improved to $1$ by Taylor expanding in the Gaussian measure along the lines of the proof of Proposition \ref{propW:BS}.}
\begin{equation}
 | R_\theta^q(g,N) |  \le  \frac{1}{ \left( \cos (2\theta) \right)^{\frac{1}{2} + q} } \; \frac{
  \left( \frac{|g|}{3}\right)^{q}
 }{ \left( \cos\frac{\varphi-2\theta}{2}  \right)^{ 2q + N/ 2 } }  \; \frac{1}{(2q)!} \; \frac{ \Gamma(2q + N/2  ) }{  \Gamma(N/2)} \;\frac{ (2q)! }{ 2^{q} q !} \; ,
\end{equation}
and using the Stirling formula as upper/lower bound for the $\Gamma$ function\footnote{In detail:
\begin{align}\label{eq:Stibound}
    1 \le \frac{\Gamma(x+1)}{ \sqrt{2\pi x} (x/e)^x } \le e^{\frac{1}{12 x}} \le e^{1/12}\; ,\qquad x \in [1, \infty) \;.
\end{align}} we have for $q$ large enough (larger that $\max\{1,N\}$):
\begin{equation}
\begin{split}
\frac{ \Gamma(2q+\tfrac{N}{2}) }{2^q q!} & \le K \frac{ (2q+\tfrac{N}{2}-1)^{\frac{1}{2} + 2q+\tfrac{N}{2}-1 } e^{-(  2q+\tfrac{N}{2}-1) } }{2^q q^{\frac{1}{2} + q} e^{-q}} \crcr 
& \le K \; 2^q \;  q^{q +\frac{N}{2} - 1 } e^{-q} \;
  \left(1 +  \frac{ \frac{N}{2} -1 }{ 2q }\right)^{\frac{4q + N-1}{2}}
\; 
\le K \; q! \; q^{\frac{N -3 }{2} } \; 2^q \; ,
\end{split}
\end{equation}
for $K$ some $q$ independent constant.\footnote{Note that $\left(1 +  \frac{ \frac{N}{2} -1 }{ 2q }\right)^{\frac{4q + N-1}{2}} \le \exp\{  \frac{4q + N-1}{2} \ln( 1 +  \frac{ \frac{N}{2} -1 }{ 2q } ) \} \le \exp\{ \frac{ ( 4q + N-1) ( \frac{N}{2} -1 )  }{4q} \} \le K$ for $q\ge 1$.} Conveniently choosing $\theta = \varphi/6$, we get the following bounds on $Z_{\varphi/6}(g,N)$ and $R_{\varphi/6 }^q(g,N)$: 
\begin{equation}
    | R_{\varphi/6 }^q(g,N) | 
     \le  \frac{K}{ \left(\cos\frac{\varphi}{3}\right)^{\frac{ N+1}{2}  }  } \; q! \; q^{\frac{N-3}{2}} \; \left( \frac{1}{ \frac{3}{2}   \left( \cos\frac{\varphi}{3} \right)^3 } \right)^q \; |g|^q \;, 
    \qquad
|Z_{ \varphi / 6} (g,N) | \le \f{1 }{  \left(\cos\f{\varphi}{3}\right)^{\f{ N+1}{2}  }} \;.
\end{equation}
Observe that $Z_{\theta}(g,N)$ is independent on $\theta$ only as long as $\varphi$ and $\theta$ are independent, but the choice $\theta = \varphi/6$ fixes $\theta$ in terms of the argument of $g$ and $Z_{\varphi/6}(g,N)$ depends on $\varphi$.

We are now in the position to prove that $Z(g,N)$ is Borel sumable along all the directions in the cut plane
$\mathbb{C}_\pi$ by verifying the conditions of Theorem~\ref{thm:Sokal}, Appendix~\ref{app:Sokal}. This comes about as follows:
\begin{itemize}
    \item let us fix some $\alpha\in (-\pi,\pi)$. As already noted in Properties~\ref{propZ:1} and~\ref{propZ:2}, $Z(g,N)$ is analytic in $\mathbb{C}_\pi$, hence in particular at $|g|e^{\imath \alpha}$ and its asymptotic expansion at 0 is known.
    \item $Z(g,N)$ is analytically continued to any $g$ in a Sokal disk (with $0$ on its boundary) tilted by $\alpha$, that is $g \in {\rm Disk}^{\alpha}_R = \{ z \mid  \mathrm{Re}(e^{\imath \alpha} / z)  > 1/R  \} $ via $Z_{\varphi/6}(g,N)$. Note that this Sokal disk extends up to $g$ with argument $\varphi = \alpha \pm \pi/2$. In the entire Sokal disk the rest term obeys the bound:
    \begin{equation}
| R_{\varphi/6 }^q(g,N) |  \le K \; q! \; q^{\frac{N-3}{2}} \; |g|^q \;
    \max_{\pm} \bigg\{ \f{1}{\left(\cos \f{\alpha \pm \frac{ \pi}{2} }{3}\right)^{\f{ N+1}{2}  } } \; \left( \f{1}{ \f{3}{2}   \left( \cos\f{\alpha \pm \frac{\pi}{2}}{3} \right)^3 } \right)^q  \bigg\}
      \; .
    \end{equation}
    For any fixed $\alpha\in(-\pi,\pi)$, $ \min_{\pm} \big\{ \frac{3}{2} \left( \cos\frac{\alpha \pm \frac{\pi}{2}}{3} \right)^3 \big\} = \rho>0$ for some $\rho$ hence the Taylor rest obeys the bound in Eq.~\eqref{eq:NS-bound}.
\end{itemize}

\paragraph{Property~\ref{propZ:4}.} Although this point is discussed in the main body of the paper, we include it also here for completeness. We denote the analytic continuation of $Z(g,N)$ to the extended Riemann sheet $\mathbb{C}_{3\pi/2}$ by:
\begin{align}
\qquad \theta >0 : \qquad  Z_+(g,N) = Z_{\theta}(g,N)  \;, \qquad \qquad
 Z_-(g,N) = Z_{ - \theta}(g,N) \; . 
\end{align}
Observe that the factorial bound on the Taylor rest term can not be satisfied (for any choice of $\theta$) when $\varphi \to \pm 3\pi/2$. As Borel summability along a direction $\alpha$ requires analytic continuation and bound on the rest term in a Sokal disk centred on that direction, hence extending up to $\alpha\pm \pi/2$, $Z(g,N)$ looses Borel summability at $g\in \mathbb{R}_-$.

From Eq.~\eqref{eq:Zsigma} we have $Z(g,N)$:
\begin{equation} 
Z(g,N)=\int_{-\infty}^{+\infty} [d\sigma]  \;e^{-\frac{1}{2} \sigma^2} \frac{1}{\left( 1 - \imath \sqrt{ \frac{g}{3} } \sigma\right)^{N/2} }  =\int_{-\infty}^{+\infty} [d\sigma]  \;e^{-\frac{1}{2} \sigma^2} \frac{1}{\left( 1 - e^{ \imath \frac{\pi}{2} }  \sqrt{ \frac{g}{3} } \sigma\right)^{N/2} }  \; ,
\end{equation}
and the Lefschetz thimble of the integral is the real axis irrespective of $g$ (comparing to Eq.~\eqref{eq:thimbles}, we see that in the present case thimbles are defined by the saddle points of $\s^2$). As long as $g = |g| e^{\imath\varphi}\in \mathbb{C}_{\pi}$, the integral converges as the singularity at
\begin{equation}
\sigma_\star = -\imath \sqrt{ \f{3}{g} } =  e^{- \imath\frac{\pi}{2} - \imath  \frac{ \varphi}{2}} \sqrt{\f{3}{|g|} } \;,
\end{equation}
lies outside the integration contour.
Notice that the singularity is a pole for $N$ even, and otherwise it is a branch point with branch cut $\sigma_\star \times (1,+\infty)$.
We will discuss in detail the general case with a branch cut, as the case of even $N$ turns out to be readable as a special case of the results at general $N$.

As $g$ approaches $\mathbb{R}_-$ the branch point hits the contour of integration: for $\varphi \nearrow \pi$ (that is we approach the cut in the $g$-plane counterclockwise) the branch point hits the real axis at $ - \sqrt{3/|g|}$ while for $\varphi \searrow -\pi$ (that is we approach the cut in the $g$-plane clockwise) the branch point hits the real axis at $ \sqrt{3/|g|}$. 
The analytic continuation $Z_+ (g,N)$ (resp. $Z_-(g,N)$) consists in tilting the contour of integration in $\sigma$ by some clockwise roatation $-\theta < 0$ (resp. counterclockwise, $\theta >0$) to avoid the collision with the branch point. However, once $g$ passes on the second Riemann sheet $ \varphi > \pi$ (resp. $\varphi < -\pi$) the tilted contour is no longer a thimble and in order to derive the asymptotic behaviour of $Z_\pm(g,N)$ we need to  to rotate it back to the real axis. This costs us a Hankel contour $C$ along the cut (see Fig.~\ref{fig:sigma-contour}):
\begin{equation}\label{eq:sigma-contour}
\begin{split}
    & Z_\pm(g,N)\big{|}_{\varphi\genfrac{}{}{0pt}{}{>\pi}{< -\pi}}=\int_{ e^{  \mp \imath \theta} \mathbb{R}} [d\sigma]  \;\frac{e^{-\frac{1}{2} \sigma^2} }{\left( 1 - \imath \sqrt{ \frac{g}{3} } \sigma\right)^{N/2} } = Z^{\mathbb{R}}(g,N) + Z^C_\pm(g,N)  \crcr
    & 
    Z^C_\pm(g,N) = \int_{C} [d\sigma] \;e^{-\frac{1}{2} \sigma^2} \frac{1}{\left( 1 - \imath \sqrt{ \frac{g}{3} } \sigma\right)^{N/2} } \;.
\end{split}
\end{equation}

The integral $Z^{\mathbb{R}}(g,N)$, defined in Eq.~\eqref{eq:defZR}, is absolutely convergent, and hence analytic, in the range $|\varphi|\in(\pi,3\pi)$, where it is bounded from above as in Eq.~\eqref{eq:Zaprioribound}.

The Hankel contour $C$ turns clockwise around the cut $\sigma_\star\times (1,+\infty)$, i.e.~starting at infinity with argument $ \frac{3\pi}{2} -\frac{\varphi}{2}$ and going back with argument $ -\frac{\pi}{2} -\frac{\varphi}{2}$ after having encircled the branch point $\sigma_\star$.
We kept a subscript $\pm$ for the contribution of the Hankel contour, because, even though the definition of $Z^C_\pm(g,N)$ and $C$ might suggest that it is one single function of $g$, in fact the integral around the cut is divergent for $|\varphi|<\pi/2$ and therefore the integrals at $\pi<\varphi<3\pi/2$ and  at  $-\pi>\varphi>-3\pi/2$ are not the analytic continuation of each other.

We will now rewrite $Z^C_\pm(g,N)$ in a more useful form. With the change of variables $\sigma = e^{- \imath\frac{\pi}{2} -\imath \frac{\varphi}{2} }  \sqrt{ 3/|g| }   \; \sigma' $ the contour $C$ becomes a Hankel contour $C'$ turning clockwise around $(1,+\infty)$ and a shift to $\sigma' = 1+t$ brings $C'$ to $C''$, a clockwise oriented Hankel contour around the positive real axis, starting at infinity with argument $2\pi$ and going back with vanishing argument after having encircled the origin:
\begin{equation}
\begin{split}
Z^C_\pm(g,N) & = \int_{C} [d\sigma] \;e^{-\frac{1}{2} \sigma^2} \frac{1}{\left( 1 - e^{\f{\imath \pi}{2} } \sqrt{ \frac{g}{3} } \sigma\right)^{N/2} }
= \left(  \frac{3}{ e^{\imath \pi} g} \right)^{1/2}  \int_{C'} [d\sigma'] \;e^{\frac{3}{2g} (\sigma')^2} \frac{1}{( 1 -  \sigma' )^{N/2} }  \crcr 
& = \frac{1}{\sqrt{2\pi}}  \left(  \frac{3}{ e^{\imath \pi} g} \right)^{1/2}    \int_{C''} dt \; ( e^{-\imath \pi} t)^{-\frac{N}{2}} \; e^{\frac{3}{2g} (1+t)^2 } 
\; ,
\end{split}
\end{equation}
 where we have made explicit the choice of branch by expressing minus signs as phases.
Notice that the integral converges because for $ \pi <|\varphi| <  \frac{3\pi}{2} $ the exponent has a negative real part.

Next, we make the change of variables $t =  e^{\imath \tau \pi} \frac{g}{3}  u$, with $\tau = -\sgn(\varphi)$ (i.e.~$\tau = -$ for $\varphi>\pi$, that is for $Z_{+}^{C}(g)$ and $\tau=+$ for $\varphi<-\pi$ i.e.\ for $Z_{-}^{C}(g)$), obtaining:
\begin{equation} \label{eq:ZCpm-intermediate}
  Z^C_\pm(g,N)  =\frac{e^{\imath\tau \pi (1-\frac{N}{2}) } }{\imath \sqrt{2\pi}}  \left(  \frac{ g} {3} \right)^{\frac{1-N}{2} }
  e^{\frac{3}{2g}}  \int_{e^{-\imath \tau \pi-\imath\varphi} C''} du \; 
  ( e^{-\imath \pi} u)^{ -\frac{N}{2}}  e^{-u+\f{g}{6}u^2} \;.
\end{equation}

The two choices of $\tau$ are dictated by the fact that for $\pi<|\varphi|<3\pi/2$ the contour of integration should stay in the domain of convergence continuously connected to $C''$. The fact that this entails two different choices of $\tau$ reflects what we anticipated about the need of keeping a $\pm$ subscript.
Lastly, the contour of integration can be deformed back to $C''$, where we easily evaluate the discontinuity (for $N<2$) as:\footnote{This is obtained by writing, for $\mathrm{Re}(z)<1$,
\[
\int_{C''} du \; e^{-S(u)}\, (e^{-\imath\pi} u)^{-z} = 
\int_{+\infty}^0  du \; e^{-S(u)}\, e^{-z (\ln|u|+\imath \pi)} + \int_0^{+\infty}  du \; e^{-S(u)}\, e^{-z (\ln|u|-\imath \pi)} = 
2\imath \sin(\pi z) \int_0^{+\infty}du \;e^{-S(u)}\, u^{-z} \;.
\]
}
\begin{equation} \label{eq:ZC-nonpert}
\begin{split}
  Z^C_\pm(g,N) & =\frac{e^{\imath\tau \pi (1-\frac{N}{2}) } }{ \sqrt{2\pi}}  \left(  \frac{ g} {3} \right)^{\frac{1-N}{2} }
  e^{\frac{3}{2g}} \;  2 \sin(\pi\f{N}{2})  \int_0^{+\infty} du \;
    e^{-u+\f{g}{6}u^2} u^{ -\frac{N}{2}} \\
    &= \frac{e^{\imath\tau \pi (1-\frac{N}{2}) } }{ \sqrt{2\pi}}  \left(  \frac{ g} {3} \right)^{\frac{1-N}{2} } e^{\frac{3}{2g}} 
  \; 2^{1+N/2} \sin(\pi\f{N}{2})  \int_0^{+\infty} d\r \;
    e^{-\f12 \r^2+\f{g}{24}\r^4} \r^{1 -N}
    \;.
\end{split}
\end{equation}
In the last step, we performed the change of variables $u=\r^2/2$ in order to make explicit that for $N=1$ we reproduce Eq.~\eqref{eq:niceresult} (times two, because the Hankel contour is only one, while there are two instanton thimbles in the $\phi$ representation). For general $N$, the integral resembles that of the $O(N)$ model in polar coordinates, except that in that case we would have the opposite sign for the power of the insertion, i.e.\ $\r^{N-1}$. This also explains the relation between the coefficients of the perturbative and nonperturbative series in Eq.~\eqref{eq:fullZ}, which are related by the transformation $N\to 2-N$ (up to an area of $S^{N-1}$ of the missing angular integration in Eq.~\eqref{eq:ZC-nonpert}).

The integral over $u$ (or $\r$) in Eq.~\eqref{eq:ZC-nonpert} is convergent as long as $\mathrm{Re}(g)<0$, i.e.~for $\pi/2<|\varphi|<3\pi/2$. One can use again the Hubbard-Stratonovich  trick to write (for $N<2$):
\begin{equation}
\begin{split}
    \int_0^{+\infty} du \;
    e^{-u+\f{g}{6}u^2} u^{ -\frac{N}{2}} 
    &=\int_{-\infty}^{+\infty} [d\s] \; e^{-\f{\s^2}{2}} \int_0^{+\infty} du \;
    e^{-u( 1  +  \sqrt{\f{g}{3}}\s)} u^{ -\frac{N}{2}}  \\
    &= \G(1-N/2)\int_{-\infty}^{+\infty} [d\s] \; e^{-\f{\s^2}{2}} \left( 1 +  \sqrt{\f{g}{3}}\s \right)^{\f{N}{2}-1}  \;,
\end{split}
\end{equation}
where, in order to ensure uniform convergence of the $u$ integral, we keep $|\varphi|=\pi$. Note that the integral is independent of choice of branch of $\sqrt{g}$, as a sign can be absorbed in $\s$. The reader will note that this is proportional to our integral in  Eq.~\eqref{eq:defZR} with arguments $Z^{\mathbb{R}}(-g,2-N)$, which is unambiguous as $Z^{\mathbb{R}}(g,N)$ is periodic with period $2\pi$ in the argument of $g$, hence one can chose any determination of $-g$.
The advantage of the manipulation above is that the integral over $\sigma$ now converges for $ 0< |\varphi|< 2 \pi$, hence it allows us to analytically continue $Z^C_{\pm}(g,N)$ beyond $|\varphi| = 3\pi/2$ up to $|\varphi| \nearrow 2\pi$.

Using Euler's reflection formula, $\Gamma( 1 - z) \sin(\pi z  ) =  \pi / \Gamma(z) $, we get:\footnote{Notice that this allows us also to analytically continue the result to $N\geq 2$.}
\begin{equation}\label{eq:Zpm}
\begin{split}
     Z^C_\pm(g,N)  & = \frac{e^{\imath \tau \pi (1-\frac{N}{2}) } }{ \sqrt{2\pi}}  \left(  \frac{ g} {3} \right)^{\frac{1-N}{2} }
  e^{\frac{3}{2g}}  \frac{2\pi}{\Gamma(N/2)}\int_{-\infty}^{+\infty} [d\s] \; e^{-\f{\s^2}{2}} \left( 1 - \imath \sqrt{\f{ -g}{3}}\s \right)^{\f{N}{2}-1}
  \crcr
  & =    e^{\imath \tau \pi (1-\frac{N}{2})  }\left(\frac{ g} {3} \right)^{\frac{1-N}{2} }
  e^{\frac{3}{2g}}  \frac{\sqrt{ 2\pi } }{\Gamma(N/2)}  \; Z^{\mathbb{R}} (- g , 2-N) 
  \crcr
  & =    e^{\imath \tau \pi (1-\frac{N}{2})  }\left(\frac{ g} {3} \right)^{\frac{1-N}{2} }
  e^{\frac{3}{2g}}  \frac{\sqrt{ 2\pi } }{\Gamma(N/2)}  \; Z ( e^{\imath \tau \pi} g , 2-N) 
  \;.
 \end{split}
\end{equation}
In the last line above we have chosen a determination of $-1$ such that $e^{\imath \tau \pi} g$ belongs to the principal sheet of the Riemann surface, where $Z=Z^{\mathbb{R}}$, hence $ Z ( e^{\imath \tau \pi} g , 2-N)   = 
Z^{\mathbb{R}} ( e^{\imath \tau \pi}  g , 2-N) = Z^{\mathbb{R}} (- g , 2-N)   $, where in the last equality we used the fact that $Z^\mathbb{R}$ is single-valued.
 We have thus shown that, when going from $|\varphi| <\pi$ to $\pi < |\varphi| < 2\pi$ our analytic continuation of $Z(g,N)$ switches:
\begin{equation}\label{eq:nicecont}
\begin{split}
Z(g,N) \xrightarrow[]{ |\varphi| \nearrow \pi_+}\   & Z^{\mathbb{R}}(g,N) + \frac{\sqrt{ 2\pi } }{\Gamma(N/2)}  \;  e^{\imath \tau \frac{\pi}{2}  } \;  e^{\frac{3}{2g}} \; \left( e^{\imath \tau \pi }   \frac{ g} {3} \right)^{\frac{1-N}{2} }  Z^{\mathbb{R}} ( - g , 2-N) 
  \crcr 
&  =  Z( e^{\imath (2\tau \pi) } g,N) +  
\frac{\sqrt{ 2\pi } }{\Gamma(N/2)}  \;  e^{\imath \tau \frac{\pi}{2}  } \;  e^{\frac{3}{2g}} \; \left( e^{\imath \tau \pi }   \frac{ g} {3} \right)^{\frac{1-N}{2} }  \; Z ( e^{\imath \tau \pi} g , 2-N) 
  \;,
  \end{split}
\end{equation}
where $ |\varphi| \nearrow \pi_+$ signifies that the switching takes place when $|\varphi|$ crosses the value $\pi$ coming from below.
In the second line above,  for $\pi< |\varphi| < 2\pi$, both arguments $  e^{\imath (2\tau \pi) } g$ and $ e^{\imath \tau \pi} g$ belong to the principal sheet of the Riemann surface , where $Z(g,N)$ has already been constructed and proven to be analytic. 

The first term in Eq.~\eqref{eq:nicecont} is regular up to $|\varphi| = 3\pi$, but the second one has a problem when 
$e^{\imath \tau \pi} g $ reaches the negative real axis (which is $\varphi \to - 2\tau \pi$). Note that 
 $ Z( e^{\imath \tau \pi} g , 2-N)$ approaches the cut singularity in the principal sheet of the Riemann surface, as its argument is $ e^{\imath \tau \pi} g $. But we already know what happens with $Z(g',N')$ when $g'$ traverses the cut singularity in the principal Riemann sheet: a branch point crosses the integration contour, one detaches a Hankel contour and the analytic continuation switches again:
\begin{equation}\label{eq:niceothercont}
\begin{split}
     Z ( e^{\imath \tau \pi} g , 2-N)  \xrightarrow[]{ |\varphi| \nearrow 2 \pi_+}\
     & Z ( e^{\imath (3\tau \pi)} g , 2-N) \crcr
     & +  \frac{\sqrt{ 2\pi } }{\Gamma(1-N/2)}  \;  e^{\imath \tau \frac{\pi}{2}  } \;  e^{\frac{3}{2ge^{\imath \tau \pi}}} \; \left( e^{\imath (2\tau \pi) }   \frac{ g} {3} \right)^{\frac{N-1}{2} }   \; Z ( e^{\imath (2 \tau \pi ) } g , N) \; ,
\end{split}
\end{equation}
where this time the arguments at which $Z$ is evaluated on the right hand side stay in the principal sheet for $2\pi < |\varphi| < 3\pi$. 

We iterate this and build the analytic continuation of $Z(g,N)$ in terms of $Z^{\mathbb{R}}$ on the whole Riemann surface. The first few steps in this continuation are:
\begin{equation}
 \begin{split}
 |\varphi| <\pi : & \qquad  Z(g,N) =  Z^{\mathbb{R}}(g,N)  \;, \crcr 
 \pi < |\varphi| < 2\pi: & \qquad Z(g,N) =  Z(e^{\imath (2\tau \pi) } g,N) + \frac{\sqrt{ 2\pi } }{\Gamma(N/2)}  \;  e^{\imath \tau \frac{\pi}{2}  } \;  e^{\frac{3}{2g}} \; \left( e^{\imath \tau \pi }   \frac{ g} {3} \right)^{\frac{1-N}{2} } \; Z( e^{\imath \tau \pi} g , 2-N)   \crcr
& \hphantom{\qquad Z(g,N)} = Z^{\mathbb{R}}(g,N) + \frac{\sqrt{ 2\pi } }{\Gamma(N/2)}  \;  e^{\imath \tau \frac{\pi}{2}  } \;  e^{\frac{3}{2g}} \; \left( e^{\imath \tau \pi }   \frac{ g} {3} \right)^{\frac{1-N}{2} } \; Z^{\mathbb{R}}(- g , 2-N) 
\;, \\
2\pi < |\varphi| < 3\pi : & \qquad  Z(g,N)
 =\left( 1 + \tilde \tau   \right)  Z(e^{\imath ( 2\tau \pi ) } g,N) \crcr
& \hphantom{\qquad Z(g,N) =}   + e^{\imath \tau \pi(N-1)}\frac{\sqrt{ 2\pi } }{\Gamma(N/2)}  \;  e^{\imath \tau \frac{\pi}{2}  } \;  e^{\frac{3}{2g}} \; \left( e^{\imath (3\tau \pi) }   \frac{ g} {3} \right)^{\frac{1-N}{2} }    \; Z( e^{\imath (3\tau \pi)  } g , 2-N)      \crcr
& \hphantom{\qquad Z(g,N) } 
 =\left( 1 + \tilde \tau   \right) Z^{\mathbb{R}}( g,N)  \crcr
& \hphantom{\qquad Z(g,N) =}   + e^{\imath \tau \pi(N-1)}\frac{\sqrt{ 2\pi } }{\Gamma(N/2)}  \;  e^{\imath \tau \frac{\pi}{2}  } \;  e^{\frac{3}{2g}} \; \left( e^{\imath (3\tau \pi) }   \frac{ g} {3} \right)^{\frac{1-N}{2} }    \; Z^{\mathbb{R}}( - g , 2-N)  \;,     
 \end{split}
\end{equation}
where we denoted:
\begin{equation}
\begin{split}
\tilde \tau = \frac{\sqrt{ 2\pi } }{\Gamma(N/2)}  \;  e^{\imath \tau \frac{\pi}{2}  } \;  e^{\frac{3}{2g}} \; \left( e^{\imath \tau \pi }   \frac{ g} {3} \right)^{\frac{1-N}{2} } \; 
\frac{\sqrt{ 2\pi } }{\Gamma(1-N/2)}  \;  e^{\imath \tau \frac{\pi}{2}  } \;  e^{\frac{3}{2ge^{\imath \tau \pi}}} \; \left( e^{\imath (2\tau \pi) }   \frac{ g} {3} \right)^{\frac{N-1}{2} } = 
2\sin (\tfrac{N\pi}{2} ) \; e^{\imath \tau \pi \frac{ N +1 }{2} } \;.
\end{split}
\end{equation}

In order to iterate  Eq.~\eqref{eq:nicecont}, we must make sure that at each step the arguments of the functions $Z$ involved in the analytic continuation are brought back to the principal sheet. 
We denote the analytic continuation of the partition function to the Riemann surface by:
\begin{equation}
\begin{split}
& 2k\pi < |\varphi| < (2k+1) \pi  :  \\
&\qquad Z(g,N ) =  \omega_{2k}  \;  Z(e^{\imath ( 2 k ) \tau \pi  } g,N)  \\
&\qquad\hphantom{Z(g,N ) =} + \eta_{2k} \;  \frac{\sqrt{ 2\pi } }{\Gamma(N/2)}  \;  e^{\imath \tau \frac{\pi}{2}  } \;  e^{\frac{3}{2g}} \; \left(   e^{\imath  (2k + 1) \tau \pi  } \frac{ g} {3} \right)^{\frac{1-N}{2} }    \; Z ( e^{\imath  (2k + 1) \tau \pi  } g , 2-N)   \;,  \\[1ex]
& (2k+1)\pi < |\varphi|  < (2k+2) \pi  :  \\
&\qquad Z(g,N ) = \omega_{2k+1}  \; Z(e^{\imath ( 2 k+2) \tau \pi  } g,N)  \\
&\qquad\hphantom{Z(g,N ) =}  +   \eta_{2k+1} \;  \frac{\sqrt{ 2\pi } }{\Gamma(N/2)}  \;  e^{\imath \tau \frac{\pi}{2}  } \;  e^{\frac{3}{2g}} \; \left(   e^{\imath  (2k + 1) \tau \pi  } \frac{ g} {3} \right)^{\frac{1-N}{2} }    \; Z ( e^{\imath (2k + 1) \tau \pi  } g , 2-N)   \;,
\end{split}
\end{equation}
a general recursion relation for the $(\omega_q,\eta_q)$ is obtained from Eq.~\eqref{eq:nicecont} and Eq.~\eqref{eq:niceothercont} generalized to the Riemann surface:\footnote{In detail, we use: 
\[
\begin{split}
&  Z( e^{\imath ( 2 k ) \tau \pi  } g ,N)
\xrightarrow[]{ |\varphi| \nearrow (2k +1 ) \pi_+}   
Z( e^{\imath (2k+2) \tau \pi } g,N) +  
\tfrac{\sqrt{ 2\pi } }{\Gamma(N/2)}  \;  e^{\imath \tau \frac{\pi}{2}  } \;  e^{\frac{3}{2g}} \; \left( e^{\imath (2k+1) \tau \pi }   \tfrac{ g} {3} \right)^{\frac{1-N}{2} }  \; Z (  e^{\imath (2k+1) \tau \pi } g , 2-N)   \;, \crcr
 &  Z (   e^{\imath (2k +1) \tau \pi  } g  , 2-N)  
\xrightarrow[]{ |\varphi | \nearrow  (2k+2)\pi_+ }  Z( e^{\imath (2k +3) \tau \pi  } g , 2-N)  +  \tfrac{\sqrt{ 2\pi } }{\Gamma(1-N/2)}  \;  e^{\imath \tau \frac{\pi}{2}  } \; 
e^{ - \frac{3}{2g }} \; \left(  e^{\imath (2k +2) \tau \pi  }   \tfrac{ g} {3} \right)^{\frac{N-1}{2} }   \; Z( e^{\imath (2k +2) \tau \pi  } g , N) \; , \crcr
&   \tfrac{\sqrt{ 2\pi } }{\Gamma(N/2)}  \;  e^{\imath \tau \frac{\pi}{2}  } \;  e^{\frac{3}{2g}} \; 
\left(   e^{\imath  (2k + 1) \tau \pi  } \tfrac{ g} {3} \right)^{\frac{1-N}{2} } \;\;
\left[  \tfrac{\sqrt{ 2\pi } }{\Gamma(1-N/2)}  \;  e^{\imath \tau \frac{\pi}{2}  } \; 
     e^{ - \frac{3}{2g }} \; \left(  e^{\imath (2k +2) \tau \pi  }   \tfrac{ g} {3} \right)^{\frac{N-1}{2} }   \right] = \tilde \tau \;.
\end{split}
\] 
}
 \begin{equation}
 (\omega_0 ,\eta_0) = (1,0)  \;,\qquad \begin{cases}
  \omega_{2k+1} &=  \omega_{2k} 
  \\
  \eta_{2k+1} & =     \eta_{2k} + \omega_{2k}
\end{cases} \; , \qquad
\begin{cases}
  \omega_{2(k+1)} &= \tilde \tau \;  \eta_{2k+1} + \omega_{2k+1}
  \\
  \eta_{2(k+1)} & =  e^{ \imath \tau \pi(N-1) }   \eta_{2k+1}
\end{cases} \;.
 \end{equation}
The recursion can easily been solved by introducing a transfer matrix:
\begin{equation}\label{eq:propZmonodromy}
    \begin{pmatrix}
    \omega_{2k}\\ \eta_{2k}
    \end{pmatrix}
    =
    A^k
    \begin{pmatrix}
    1\\ 0
    \end{pmatrix} \;,
    \qquad
    A= \begin{pmatrix}
    1+ \tilde{\tau} & \tilde{\tau} \\ e^{ \imath \tau \pi(N-1) } & e^{ \imath \tau \pi(N-1) }
    \end{pmatrix} \;,
\end{equation}
which leads to Eq.~\eqref{eq:rec_sol}.
Since the eigenvalues of $A$ are $\pm e^{\imath \t\pi \frac{ N}{2} }$, we have that $A^k$ equals the identity matrix for $k=4$ if $N$ is odd, and for $k=2$ if $N$ is even. Therefore, in these two cases we have a monodromy group of order 4 and 2, respectively.
More generally, we have a monodromy group of finite order if $N$ is a rational number, and an infinite monodromy otherwise.

\paragraph{Property~\ref{propZ:4.5}.} This follows by combining Porperty~\ref{propZ:2}, which gives the asymptotic expansion of $Z^{\mathbb{R}}(g,N)$, with Eq.~\eqref{eq:Zpm} 
using $ (1+x)^{N/2-1} = \sum_{q\ge 0}  \frac{ \Gamma(N/2) }{ q! \Gamma(N/2-q) } x^{q}  $:
\begin{equation}
     Z^C_\pm(g,N)  =  e^{ \imath \tau \pi (1-\frac{N}{2}  ) }  \sqrt{2\pi}  \left(  \frac{ g} {3} \right)^{\frac{1-N}{2} }
  e^{\frac{3}{2g}} \sum_{q\ge 0} \frac{1}{ 2^{2q} q! \; \Gamma(\frac{N}{2} -2q ) } \left( \frac{2 g}{3} \right)^q \;.
\end{equation}

\paragraph{Property~\ref{propZ:5}.} In order to compute the discontinuity of the partition function, let us consider $g$ in the complex plane of the coupling constant slightly below the negative real axis, $\mathrm{Re}(g),\mathrm{Im}(g)<0$. This $g$ can be reached either counterclockwise with $Z_+$ or clockwise with $Z_-$. There is a subtlety here: if we denote this point as $g = |g|e^{\imath \varphi}$ with $\varphi \in (\pi, 3\pi/2)$ in the clockwise direction, it corresponds to $g = |g|e^{\imath (\varphi - 2\pi ) } $ in the counterclockwise direction. While only the real axis contributes to $Z_-(g,N)$ (as turning clockwise we do not cross the cut to reach it), $Z_+(g,N)$ has the additional contribution of the Hankel contour (as turning counterclockwise we cross the cut):
\begin{equation}
    Z_+(g,N) = Z^{\mathbb{R}}(g,N) + Z^C_+(g,N) \;,\qquad       Z_-(g,N) = Z^{\mathbb{R}} (e^{-2\pi \imath }g,N)  =Z^{\mathbb{R}}(g,N) \; ,
\end{equation}
where we used the fact that $Z^{\mathbb{R}}$ is single-valued. When taking the difference, the $Z^{\mathbb{R}}$ pieces cancel and the discontinuity is given by the Hankel contour contribution,
$ Z_-(g,N) -Z_+(g,N) = -  Z^C_+(g,N) $. In particular, at the negative real axis we have:
\begin{equation}
{\rm disc}_\pi \big(Z(g,N)\big) = \lim_{ \varphi \searrow \pi } \left(Z_-(e^{-2\pi \imath }g,N) - Z_+(g,N) \right) = -  Z^C_+(e^{\imath \pi}|g|,N) \;.
\end{equation}
Lastly, notice that the computation done on the other side of the cut, $\lim_{ \varphi \nearrow -\pi } \left(Z_-(g,N) - Z_+(e^{2\pi \imath }g,N) \right)$, gives the same result, because $Z^C_-(e^{-\imath \pi}|g|,N) = -  Z^C_+(e^{\imath \pi}|g|,N)$.

\paragraph{Property~\ref{propZ:7}.} 
From the intermediate field expression Eq.~\eqref{eq:Zsigma} of the partition function we find straightforwardly that $ (N+4 g\p_g) Z(g,N) = N Z(g,N+2)$.
Applying this twice, we have:
\be
(N+2+4 g\p_g) (N+4 g\p_g) Z(g,N) = N (N+2) Z(g,N+4) \;.
\ee
Integrating by parts in the right-hand side, we find that $N (N+2) Z(g,N+4)=-4! Z^\prime(g,N)$, and thus we arrive at the equation:
\begin{equation} \label{DifferentialEqVector}
	N(N+2)Z(g,N)+\left( (8N+ 24)g+24 \right)Z^\prime(g,N)+16 g^2 Z^{\prime \prime}(g,N)=0 \; .
\end{equation}
This concludes the proof of Proposition~\ref{prop:Z}.
\end{proof}

\subsection{Properties of $Z_n(g)$}
\label{app:proofZn}

\begin{proof}[Proof of Proposition~\ref{prop:Zn}]The proof is similar to the one of Proposition~\ref{prop:Z}.

\paragraph{Property~\ref{propZn:1}.} We start from:
\begin{equation}
	Z_n(g)= \int_{-\infty}^{+\infty} [d\s] e^{-\f12 \s^2} \left( \ln(1-\imath \sqrt{\f{g}{3}} \s ) \right)^n = \int_{-\infty}^{+\infty} [d\s] e^{-\f12 \s^2} \left( \ln(1-\imath e^{\imath\frac{\varphi}{2}}\sqrt{\f{|g|}{3}} \s ) \right)^n  \; ,
\end{equation}	
where we parameterized $g = |g| e^{\imath \varphi}$ with $\varphi \in (-\pi,\pi)$. 
For any real $x$ we have $ 1+|x| \ge | 1 - \imath e^{\imath \frac{\varphi}{2}} x  |  \ge \cos\frac{\varphi}{2} $, hence we can bound the logarithm as:
\begin{equation}
    \begin{split}
 |\ln (1 - \imath e^{\imath \frac{\varphi}{2}} x )|^2  
 & \le
 \pi^2 + \max\bigg\{ [ \ln ( \cos\tfrac{\varphi}{2}) ]^2,  \bigg[ \ln(1 + |x|) \bigg]^2 \bigg\} \\
& \le 2 \bigg[ \left( | \ln ( \cos\tfrac{\varphi}{2})|+1\right) \ln(e^\pi + |x| )\bigg]^2  \;,
    \end{split}
\end{equation}
which implies:
\begin{align}
 |Z_n(g)| \le \bigg| 2^{1/2} \left( |\ln ( \cos\tfrac{\varphi}{2})| +1\right) \bigg|^n  \int [d\sigma] e^{-\frac{1}{2} \sigma^2 } 
     \bigg[ \ln( e^{\pi} + \sqrt{\frac{|g |}{3}} |\sigma|) \bigg]^{n} \;.
\end{align}
We have $ \ln( e^{\pi} +  \sqrt{\frac{|g |}{3}} |\sigma|) \le \frac{1}{\ve} (e^{\pi} +  \sqrt{\frac{|g |}{3}} |\sigma| )^\ve$ for any $\ve>0$.
Cutting the integration interval into regions where 
$ e^{\pi} <  \sqrt{\frac{|g |}{3}}|\sigma|$ and regions where $ e^{\pi}> \sqrt{\frac{|g |}{3}} |\sigma|$ and extending each piece back to the entire real line we get a bound: 
\begin{equation}
\begin{split}
    | Z_n(g) | & \le \frac{ \bigg| 2^{1/2} \left( | \ln ( \cos\frac{\varphi}{2}) |+1\right) \bigg|^n   }{ \ve^n  }
     \left[ (2e^\pi)^{n\ve} + \left( \frac{4|g |}{3}\right)^{\frac{n\ve}{2}} \int [d\sigma] e^{-\frac{1}{2}  \sigma^2 }  
    |\sigma|^{n\ve}  \right] \crcr 
    & \le K^n \frac{ \bigg(  | \ln ( \cos\frac{\varphi}{2})| +1 \bigg)^n   }{ \ve^n  }
     \bigg(1 + |g|^{\frac{n\ve}{2}} \Gamma(\tfrac{n\ve+1}{2}  ) \bigg) \; .
\end{split}
\end{equation}

\paragraph{Property~\ref{propZn:2}.}
In order to derive the asymptotic expansion of $Z_n(g)$, we use Fa{\`a} di Bruno's formula:\footnote{Namely, we evaluate the $q$-th term of the Taylor expansion using the following formula (at $u=0$)
\[
    \frac{d^q}{du^q} [\ln(1 - u x)]^n 
  =\sum_{\substack{m_1,\dots m_q\\\sum km_k = q, \, \sum m_k \le n}}
     \frac{q!}{\prod_k m_k! (k!)^{m_k}} \; \frac{ n !}{ ( n- \sum_k m_k)!  } \;
     [\ln(1 - u x)]^{n - \sum_k m_k}  \prod_k \left( \frac{ -  (k-1)! x^{k} }{ (1 - ux)^k } \right)^{m_k} \;.
\]
} to expand:
\be 
\left(\ln(1-\imath \sqrt{\f{g}{3}} \s )\right)^n = \sum_{m\geq n} \left( \imath \sqrt{\f{g}{3}}\right)^m \sigma^m \sum_{\substack{m_1,\ldots, m_{m} \ge 0 \\ \sum  k m_k=m , \; \sum m_k=n}} \f{(-1)^n n!}{\prod _k k^{m_k}m_k!} \; ,
\ee
and integrate the Gaussian term by term. As a cross check, we can (formally) resum:\footnote{
In the last step, we use: $\sum_{m_1,m_2,\ldots \ge 0} \prod_{k\ge 1} \f{1}{m_k!} \left(\f{Nx^k}{2k}\right)^{m_k} 
= \exp\{\sum_{k\ge 1} \f{Nx^k}{2k} \} = (1-x)^{-N/2} = \sum_{m=0}^\infty \f{\Gamma(m+N/2)}{\Gamma(N/2)m!} x^m$.
}
\be
\begin{aligned}
&\sum_{n\ge 0}\frac{1}{n!} \left(-\frac{N}{2}\right)^n Z_n(g) \simeq \sum_{n\geq 0} \f{1}{n!} \left(-\frac{N}{2}\right)^n \sum_{m\geq n/2} \left( -\f{2g}{3}\right)^m \, \f{(2m)!}{2^{2m} m!}
\sum_{\substack{m_1,\ldots, m_{2m} \ge 0 \\ \sum k m_k=2m , \sum  m_k=n}} \f{(-1)^n n!}{\prod_{k}k^{m_k}m_k!} \crcr
&\qquad= \sum_{m=0}^\infty \left( -\f{2g}{3}\right)^m \, \f{(2m)!}{2^{2m} m!}
\sum_{\substack{m_1,\ldots, m_{2m} \ge 0\\ \sum k m_k=2m}} \prod_{k=1}^{2m}\left(\f{N}{2k}\right)^{m_k}\f{1}{m_k!}
= \sum_{m=0}^\infty \f{\Gamma(2m+N/2)}{2^{2m} m!\Gamma(N/2)} \left( -\f{2g}{3}\right)^m   \; ,
\end{aligned}
\ee
reproducing the asymptotic expansion of $Z(g,N)$ in Eq.~\eqref{eq:Zpertseries}. 

\paragraph{Property~\ref{propZn:3}.} We analytically continue $Z_n(g)$ to the extended Riemann sheet $\mathbb{C}_{3\pi/2}$ as in Property \ref{propZn:3}, Proposition \ref{prop:Zn}, by turning the integration contour by $e^{-\imath\theta}$:
\begin{equation}
    Z_{n\theta}(g)= \int_{-\infty}^{+\infty} e^{-\imath \theta}[d\s] e^{-\f12 e^{-2 \imath\theta} \s^2} \left( \ln(1-\imath e^{ \imath\frac{\varphi-2\theta}{2} }\sqrt{\f{|g|}{3}} \s ) \right)^n  \;.
\end{equation}

Using a Taylor formula with integral rest, the reminder $ R_{n\theta}^q(g) $ writes:
\begin{align}
R_{n\theta}^q(g)   & =  \int_0^1 du \; \frac{(1-u)^{2q-1}}{(2q-1)!}  \;
 \int_{ \mathbb{R}} e^{-\imath \theta}[d\sigma]  \;e^{-\frac{1}{2}e^{-2\imath\theta} \sigma^2} 
  \left(\frac{d}{du} \right)^{2q} \left( \ln(1-\imath e^{ \imath\frac{\varphi-2\theta}{2} }\sqrt{\f{|g|}{3}} \s u) \right)^n  \\
 & = \int_0^1 du \; \frac{(1-u)^{2q-1}}{(2q-1)!} 
 \int_{ \mathbb{R}} e^{-\imath \theta} [d\sigma]  \;e^{-\frac{1}{2}e^{-2\imath\theta} \sigma^2}  
 \frac{ \left(   \imath e^{ \imath\frac{\varphi-2\theta}{2} } \sqrt{\f{|g|}{3}} \s \right)^{2q} }{ \left( 1-\imath e^{ \imath\frac{\varphi-2\theta}{2} }\sqrt{\f{|g|}{3}} \s u \right)^{2q} }\crcr
&\qquad \sum_{\substack{m_1,\dots m_q\\\sum km_k = 2q, \, \sum m_k \le n}}
    \frac{ (2q)!}{\prod_k m_k! k^{m_k}} \; \frac{ (-1)^{\sum m_k} n !}{ ( n- \sum_k m_k)!  } \;
     \left( \ln(1-\imath e^{ \imath\frac{\varphi-2\theta}{2} }\sqrt{\f{|g|}{3}} \s u) \right)^{n - \sum_k m_k} 
\; . \nonumber
\end{align}
Now $| 1 - \imath \sqrt{ \frac{|g|}{3} }e^{\imath \frac{\varphi-2\theta}{2}} \sigma u | \ge \cos\frac{\varphi-2\theta}{2} $ and: 
\begin{equation}
    \frac{ \left( \ln(1-\imath e^{ \imath\frac{\varphi-2\theta}{2} }\sqrt{\f{|g|}{3}} \s u) \right)^{n - \sum_k m_k} }{   \left( 1-\imath e^{ \imath\frac{\varphi-2\theta}{2} }\sqrt{\f{|g|}{3}} \s u \right)^{2q} } \le \frac{ \bigg| 2^{1/2}( | \ln (\cos \frac{\varphi-2\theta}{2})| +1) \ln(e^{\pi} + |g| |\sigma| )\bigg|^{n-\sum m_k}  } { \left( \cos\frac{\varphi-2\theta}{2} \right)^{2q} } \;,
\end{equation}
therefore:
\begin{equation}\label{eq:bounrestZn}
\begin{split}
    |R_{n\theta}^q(g)| & \le \frac{ n! 2^{\frac{n}{2}}\bigg(| \ln (\cos \frac{\varphi-2\theta}{2})|+1 \bigg)^n }{\left( \cos\frac{\varphi-2\theta}{2} \right)^{2q} }
    \left(\frac{|g|}{3} \right)^q \frac{\Gamma(2q+\frac{1}{2})}{(2q)!} \int[d\s] e^{-\frac{1}{2} \cos (2\theta) \sigma^2 } \sigma^{2q} [\ln(e^{\pi} + |g| |\sigma| )]^n \crcr
    & \le K \frac{ \bigg( | \ln (\cos \frac{\varphi-2\theta}{2})| +1 \bigg)^n } 
    { \left( \cos\frac{\varphi-2\theta}{2} \right)^{2q} }
     \left( \frac{ |g|}{3} \right)^q \left( \frac{(2q)!}{ 2^q q! \left( \cos(2\theta)  \right)^{q+1/2}  } + 2^q \frac{  |g|^{\frac{n\ve}{2}} \Gamma(q +  \frac{n\ve+1}{2} ) }{     \left( \cos(2\theta)  \right)^{q+(n\ve+1)/2} }   \right)\;,
\end{split}
\end{equation}
for any $\ve$ with ($K$ depending on $n$ and $\ve$). Fixing for example $\ve=1$ and $\theta = \frac{\varphi}{6}$ yields the desired result as in Property \ref{propZ:3}, Proposition \ref{prop:Z}.

\paragraph{Property~\ref{propZn:4}.}
A slight variation on the bound in Property~\ref{propZn:1} yields: 
\begin{equation}\label{eq:problemeq}
    | Z_{n\theta}(g) | \le K^n \frac{ \bigg(  | \ln ( \cos\frac{\varphi-2\theta}{2})| +1 \bigg)^n   }{ \ve^n  }
     \bigg(1  + \frac{  |g|^{\frac{n\ve}{2}} \Gamma( \frac{n\ve+1}{2}  ) }{     \left( \cos(2\theta)  \right)^{(n\ve+1)/2} }  \bigg) \; .
\end{equation}
Fixing for convenience $\theta = \frac{\varphi}{6}$, one observes that $\sum_{n\ge 0}\frac{1}{n!} (-N/2)^n Z_{n\theta}(g)$ has infinite radius of convergence in $N$ for any $|\varphi|<3\pi/2$. Denoting, similarly to the notation we used for $Z(g,N)$, the analytic continuation of $Z_n(g)$ to the extended Riemann sheet $\mathbb{C}_{3\pi/2}$ by:
\begin{align}
\theta> 0 \; ,\qquad  Z_{n+}(g) = Z_{n\theta}(g) \;, \qquad \qquad
 Z_{n-}(g) = Z_{n - \theta}(g) \; ,
\end{align}
we note that the bound on the Taylor rest term in Eq.~\eqref{eq:bounrestZn} is lost for $\varphi \to \pm3\pi/2$, hence Borel summability is lost for $\varphi \to \pm\pi$. As in the case of $Z(g,N)$, after $g$ crosses the cut at $\mathbb{R}_-$, say counterclockwise, in deforming the contour of integration $e^{-\imath \theta} \mathbb{R}$ back to the steepest-descent contour along the real line, we generate a clockwise-oriented Hankel contour around $ \sigma_{\star} \times (1,+\infty)$:
\begin{equation}
\begin{split}
    & Z_{n\pm}(g)\big{|}_{\varphi\genfrac{}{}{0pt}{}{>\pi}{< -\pi}}=\int_{ e^{  \mp \imath \theta} \mathbb{R}} [d\sigma]  \; e^{-\frac{1}{2} \sigma^2} \left( \ln( 1 - \imath \sqrt{ \frac{g}{3} }  \sigma) \right)^{n}  = Z^{\mathbb{R}}_n(g) + Z^C_{n\pm}(g)  \crcr
   & Z^{\mathbb{R}}_n(g) = 
    \int_{ \mathbb{R}} [d\sigma]  \;e^{-\frac{1}{2} \sigma^2}  \left( \ln( 1 - \imath \sqrt{ \frac{g}{3} }  \sigma ) \right)^{n}  \;,\qquad  Z^C_{n\pm}(g) = \int_{C} [d\sigma] \;e^{-\frac{1}{2} \sigma^2} \left( \ln(1 - \imath \sqrt{ \frac{g}{3} } \sigma)\right)^{n}  \;.
\end{split}
\end{equation}

\paragraph{Property~\ref{propZn:5}.}

We now compute $Z_{n\pm}^{C}(g)$. As in Proposition~\ref{prop:Z} this is given by integrating along the clockwise orientated Hankel contour $\cC$. First we change variable to $\sigma = e^{-\imath \frac{\pi}{2}-\imath\frac{\varphi}{2}} \sqrt{3/ |g| } \sigma'$ and the contour becomes $C'$ clockwise oriented around $(1,\infty)$ and a shift $\sigma'=1+t$ brings $C'$ to a clockwise oriented contour around the positive real axis:
\begin{equation}
    Z_{n\pm}^{C}(g)=  \left( \frac{3}{e^{\imath \pi} g }\right)^{1/2} \int_{C'} [d\sigma']
     \; e^{\frac{3}{2g} (\sigma')^2} [ \ln(1 - \sigma') ]^n
     = \frac{1}{\sqrt{2\pi}} \left( \frac{3}{e^{\imath \pi} g }\right)^{1/2} \int_{C''} dt \; e^{\frac{3}{2g}(1+t)^2} [\ln (-t)]^n \;.
\end{equation}
A further change of variables $t = e^{\imath \tau\pi}\frac{g}{3} u$ with 
$\tau = -\sgn(\varphi)$ yields:\footnote{See comments below Eq.~\eqref{eq:ZCpm-intermediate} for an explanation of the $\t$-dependence.}
\begin{equation}\label{eq:Zn-Hankel}
\begin{split}
 Z_{n\pm}^{C}(g)= & \frac{e^{\imath \tau \pi}}{\imath} \, \frac{1}{\sqrt{2\pi}}   \sqrt{\frac{g}{3}} e^{\frac{3}{2g} } \int_0^{\infty} du \; e^{-u+\f{g}{6}u^2} \left[ ( \ln(\tfrac{e^{\imath \tau \pi}  g}{3}u) -\imath \pi )^{n} -( \ln( \tfrac{e^{\imath \tau \pi} g}{3}u) + \imath \pi )^{n} \right]
\\
=& \frac{e^{\imath \tau \pi}}{\imath} \, \frac{1}{\sqrt{2\pi}}   \sqrt{\frac{g}{3}} e^{\frac{3}{2g} }
\sum_{q\ge 0} \frac{1}{q!} \left( \frac{g}{6}\right)^q  \sum_{p=0}^n \binom{n}{p} 
\\
&\qquad \int_0^{\infty} du \; e^{-u} \; u^{2q}  (\ln u)^p \left[ ( \ln(\tfrac{e^{\imath \tau \pi}  g}{3}) -\imath \pi )^{n-p} -( \ln( \tfrac{e^{\imath \tau \pi} g}{3}) + \imath \pi )^{n-p} \right] \;,
\end{split}
\end{equation}
and integrating over $u$ we get 
\begin{equation}
\begin{split}
Z_{n\pm}^{C}(g) \simeq & \frac{e^{\imath \tau \pi}}{\imath} \, \, \f{e^{\f{3}{2g}}}{\sqrt{2\pi}} \sqrt{\f{ g}{3}}
\sum_{q=0}^{\infty} 	\f{1}{q!}\left(\f{g}{6}\right)^q \sum_{p=0}^{n}
\binom{n}{p} \f{d^p \G(z)}{dz^p}\Big|_{z=2q+1} \crcr 
& \qquad \qquad \left[\left(\ln\left(\tfrac{e^{\imath\tau\pi}g}{3}\right) - \imath \pi \right)^{n-p}
-\left(\ln\left(\tfrac{e^{\imath\tau\pi}g}{3}\right) +\imath \pi \right)^{n-p}\right] \;,
\end{split}
\end{equation}
which combined with Property \ref{propZn:2} implies the full transseries expansion Eq.~\eqref{eq:transsereiesZn}.

A good cross check of the results consist in summing over $n$:
\be\begin{aligned}
& \sum_{n\geq 0} \f{1}{n!} \left(-\frac{N}{2}\right)^n Z_{n\pm}^{C}(g)
=\frac{e^{\imath \tau \pi }}{\imath}  \f{e^{\f{3}{2g}}}{\sqrt{2\pi}}\sqrt{\f{g}{3}}
\sum_{q=0}^{\infty} 	\f{1}{q!}\left(\f{g}{6}\right)^q  \sum_{p,k=0}^{\infty} \f{1}{p!\,k!}\left(-\frac{N}{2}\right)^{p+k}  \f{d^p \G(z)}{dz^p}\Big|_{z=2q+1} \crcr
& \qquad \qquad \qquad \qquad \qquad \qquad \left[\left(\ln\left(\tfrac{e^{\imath \tau \pi}g}{3}\right) - \imath \pi \right)^{k}
-\left(\ln\left(\tfrac{e^{\imath \tau \pi}g}{3}\right) + \imath \pi \right)^{k}\right]
\\
& \qquad = \frac{e^{\imath \tau \pi }}{\imath}  \f{e^{\f{3}{2g}}}{\sqrt{2\pi}}\sqrt{\f{g}{3}}
\sum_{q=0}^{\infty} \frac{1}{q!} \left(\f{g}{6}\right)^q \Gamma(2q+\f{2-N}{2}) 
\left[  e^{\imath \frac{N\pi}{2}}  \left(\tfrac{e^{\imath \tau \pi}g}{3}  \right)^{-\f{N}{2}}
-e^{-\imath \frac{N\pi}{2}}  \left(\tfrac{e^{\imath \tau \pi}g}{3}  \right)^{-\f{N}{2}}\right]
		\\
& \qquad =  \frac{  e^{\imath \tau \pi ( 1- \frac{N}{2} ) }  }{\imath} \f{e^{\f{3}{2g}}}{\sqrt{2\pi}} \left(\frac{g}{3} \right)^{\frac{1-N}{2}}
   \sum_{q\ge 0} \frac{ \Gamma(2q+ \frac{2-N}{2}) \frac{ \sin\frac{N\pi}{2} }{\pi} 2\pi \imath }{ 2^{2q}q! } \left( \frac{2g}{3}\right)^q \;,
\end{aligned}\ee
Which coincides with the instanton part in Eq.~\eqref{eq:fullZ}.

\paragraph{Property~\ref{propZn:6}.} This follows from  $    {\rm disc}_\pi \big(Z_n(g)\big) = - Z_{n+}^C(e^{\imath \pi}|g|)$ combined with Properties \ref{propZn:4} and \ref{propZn:5}.

\paragraph{Property~\ref{propZn:7}.} The derivation of Eq.~\eqref{PDEZn} is straightforward. 
We substitute the small-$N$ expansion (\ref{eq:Z-logseries}), with the condition $Z(g,0)=1$, in  the partial differential equation for $Z(g,N)$, Eq.~(\ref{PDEpartitionfunctionN}), and collect the terms with the same powers of $N$. 

This concludes the proof of Proposition~\ref{prop:Zn}.
\end{proof}

\subsection{The LVE expansion, analyticity}
\label{app:LVE}

\begin{proof}[Proof of Proposition~\ref{prop:LVE+bounds}.]

We review here briefly the proof the LVE formula for the free energy. For more details, see \cite{Rivasseau:2007fr} and followup work. We use the notation of the Gaussian integral as a differential operator (see  \cite{bauerschmidt2019introduction} for a detailed discussion of this notation), as a convenient bookkeeping device for the action of derivatives with respect to matrix elements of the covariance of a Gaussian measure and we denote $V(\sigma) = \ln(1 - \imath \sqrt{ g/3  } \, \sigma)  $. Eq.~\eqref{eq:Z-logseries} becomes with this notation: 
\begin{equation}
     Z(g,N) =  \sum_{n\ge 0}\frac{1}{n!} \left( -\frac{N}{2}\right)^n Z_n \;,\quad
      Z_n(g) = \int[d\sigma]   \;  e^{ -  \frac{1}{2} \sigma^2 } [\ln(1 - \imath \sqrt{ g/3  } \, \sigma)]^n
      \equiv \left[ e^{ \frac{1}{2} \frac{\d  }{\d  \sigma }  \frac{\d }{\d \sigma } }    V(\sigma)^n  \right]_{\sigma=0} \;,
\end{equation}
where $\d / \d \sigma  $ denotes the derivative with respect to $\sigma$. 

Let us take aside a term $Z_n(g)$. We introduce copies with degenerate covariance and we introduce fictitious interpolating link parameters $x_{k l} = x_{lk} =1$:
\begin{equation}
Z_n(g) = \left[ e^{ \frac{1}{2} \sum_{k,l=1}^n \frac{\d  }{\d  \sigma_k }  \frac{\d }{\d \sigma_l } }  \prod_{i=1}^n V(\sigma_i)  \right]_{\sigma_i=0} = \left[ e^{ \frac{1}{2} \sum_{k,l=1}^n x_{kl}\frac{\d  }{\d  \sigma_k } \frac{\d }{\d \sigma_l } }  \prod_{i=1}^n V(\sigma_i  ) \right]_{\sigma_i=0, x_{ij}=1}   \;.
\end{equation}
We fix the diagonal elements $x_{ii}=1$ and use symmetric interpolations $x_{ij}=x_{ji}$ where $x_{ij}$ with $i<j$ are independent parameters. For all $i\neq j$:
\begin{equation}
\frac{\partial }{ \partial x_{ij}} 
 \left[ e^{ \frac{1}{2} \sum_{k,l=1}^n x_{kl}\frac{\d}{\d  \sigma_k }  \frac{\d}{\d \sigma_l } }  \right] 
 =  e^{ \frac{1}{2} \sum_{k,l=1}^n x_{kl}\frac{\d}{\d  \sigma_k }  \frac{\d}{\d \sigma_l } }   
  \;  \left( \frac{\d}{\d \sigma_i }   \frac{\d  }{\d \sigma_j } \right) \; ,
\end{equation}
and using Appendix~\ref{app:BKAR}, we get:
\begin{equation} 
Z(g,N)  = \sum_{n\ge 0} \frac{\left( -\frac{N}{2}\right)^n}{n!}  \sum_{\cF \in F_n } \int_{0}^1  \prod_{ (i,j) \in \cF } du_{ij}   \Bigg[  e^{ \frac{1}{2} \sum_{k,l} w_{kl}^{\cF} \frac{\d }{\d \sigma_k }   
  \frac{\d }{\d \sigma_l } }   
\left( \prod_{ (i,j)\in \cF } \frac{\d}{\d \sigma_i } \frac{\d}{\d \sigma_j} 
    \right) \prod_{i=1}^n V( \sigma_i )  \Bigg]_{\sigma_i = 0 } \; ,    
\end{equation}
where $F_n$ is the set of all the forests over $n$ labeled vertices. 

Observing that the contribution of a forest factors over the trees (connected components) in the forest, the logarithm is trivial: it comes to restricting the sum above to trees over $n$ vertices (hence with $n-1$ edges):  
\begin{equation}
W(g,N)  =  \sum_{n\ge 1} \frac{\left( -\frac{N}{2}\right)^n}{n!}   \sum_{ \cT \in T_n }
 \int_{0}^1 \prod_{ (i,j) \in \cT } du_{ij} 
\Bigg[ e^{ \frac{1}{2} \sum_{k,l} w_{kl}^{\cT} \frac{\d }{\d \sigma_k } \frac{\d }{\d \sigma_l } } 
\left( \prod_{ (i,j)\in \cT } \frac{\d }{\d \sigma_i } \frac{\d }{\d \sigma_j  }  \right) \prod_{i=1}^n  V(\sigma_i )   \Bigg]_{\sigma_i =0 } \; .    
\end{equation}
Taking into account that the action of the derivatives on the interaction is:
\begin{equation}
 \frac{\d^d}{\d \sigma^d}    \ln\big( 1 -\imath \sqrt{\tfrac{g}{3}}   \sigma \big) = (-1) \; \frac{ (d-1)!\left( \imath \sqrt{\tfrac{g}{3}  }\right)^d  } { \left(  1 -\imath \sqrt{\tfrac{g}{3}} \sigma \right)^d } \; ,  
\end{equation}
denoting $d_i$ the degree of the vertex $i$ in $\cT$, recalling that $\sum_i d_i = 2(n-1)$ and observing that the terms $n=1$ is special we get:
\be
\begin{split}
& W (g,N) =  - \frac{N}{2} \Bigg[  e^{ \frac{1}{2}  \frac{\d }{\d \sigma} \frac{\d }{\d \sigma} } 
   \ln\Big( 1 -\imath \sqrt{\tfrac{g}{3}}   \sigma \Big) \Bigg]_{\sigma =0 } \\
& \; - \sum_{n\ge 2} \frac{1}{n!} \left( -\frac{N}{2} \right)^n \left( \frac{g}{3} \right)^{n-1} \sum_{ \cT \in T_n } 
 \int_{0}^1  \prod_{ (i,j) \in \cT } du_{ij} 
\Bigg[  e^{ \frac{1}{2} \sum_{i,j} w_{ij}^{\cT}  \frac{\d }{\d \sigma_i} \frac{\d }{\d \sigma_j} } 
\;\; \prod_{i} \frac{(d_i-1)!}{ \left(   1 -\imath \sqrt{\tfrac{g}{3}}   \sigma_i \right)^{d_i} } \Bigg]_{\sigma_i =0 }\; .
\end{split}
\ee
Reinstating the notation of the normalized Gaussian measure as a probability density, we obtain the series in Eq.~\eqref{eq:W_BKAR}, with coefficients \eqref{eq:W_BKARterm}.

\bigskip

\paragraph{Property~\ref{propLVE:1} and ~\ref{propLVE:2}.}
In order to determine the domain of convergence of 
\eqref{eq:W_BKAR}, let us denote as usual $g = |g| e^{\imath\varphi }$, and take $-\pi < \varphi < \pi$. In this region, we can use the uniform bound
$|1 -\imath \sqrt{\tfrac{g}{3}}   \sigma_i  |  \ge \cos \frac{ \varphi}{2}$, remaining with Gaussian integrals over $\sigma_i$ and integrals over $u_{ij}$, both bounded by 1. Therefore, for $n\ge 2$, a crude bound on $W_n(g)$ is:
\begin{equation}
  |W_n(g)|   \le   \left| \frac{g}{3} \right|^{n-1} \frac{ (n-2)! }{ ( \cos\frac{\varphi}{2} )^{2(n-1)} } \sum_{d_i\ge 1}^{\sum_{i=1}^n (d_i-1)  =n-2} 1 = \frac{(2n-3)!}{(n-1)!}
 \left| \frac{g}{3 (\cos\frac{\varphi}{2} )^{2} } \right|^{n-1}    \;,          
 \end{equation}
where we used the fact that there are $(n-2)!/\prod_i (d_i-1)!$ trees over $n$ labeled vertices with degree $d_i$ at the vertex $i$, and that the sum over $d_i$ is the coefficient of $x^{n-2}$ in the expansion of $(1-x)^{-n}  =\sum_{q\ge 0} \binom{n-1+q}{n-1} \; x^q  $. Using Stirling's formula for the asymptotics of the factorials, it follows that $W(g,N)$ is convergent in a cardioid domain $|gN|<\frac{3}{2}  (\cos\frac{\varphi}{2} )^{2}$. 

\bigskip

\paragraph{Property~\ref{propLVE:3}.} 
The domain of convergence can be enlarged by turning $\sigma\to e^{-\imath \theta} \sigma$, which yields a convergent expansion in a subdomain of the extended Riemann sheet $\mathbb{C}_{3\pi/2}$. In order to make this precise, let us define:
\be\label{eq:contsig}
\begin{split}
&W_\theta(g,N) = - \frac{N}{2} \int  e^{-\imath \theta} [d\sigma]  \; e^{-\frac{1}{2} e^{-2\imath \theta} \sigma^2 } 
   \ln\Big( 1 -\imath \sqrt{\tfrac{g}{3}}   e^{-\imath \theta }\sigma \Big) - \sum_{n\ge 2} \frac{1}{n!} \left( -\frac{N}{2} \right)^{n}\left( \frac{g}{3} \right)^{n-1}  \\
 &\qquad \times \sum_{ \cT \in T_n  } 
 \int_{0}^1  \prod_{ (i,j) \in \cT } du_{ij} 
\int_{\mathbb{R}} \frac{\prod_i e^{-\imath \theta} [ d\sigma_i]  }{\sqrt{ \det w^\cT }}  \; e^{ - \frac{1}{2} e^{-2\imath \theta}\sum_{i,j} \sigma_i ( w^\cT)_{ij}^{-1}  \sigma_j } 
\;\; \prod_{i} \frac{(d_i-1)!}{ \left(   1 -\imath \sqrt{\tfrac{g}{3}}   e^{-\imath \theta}\sigma_i \right)^{d_i} }  \;,
\end{split}
\ee
and implicitly $W_{n\theta}(g)$. It is easy to check that, for $g$ in the cut plane $\mathbb{C}_{\pi}$ both $W_{\theta}(g,N)$ and $W_{n\theta}(g)$ are independent of $\theta$ as long as they converge. As $\theta$ can be chosen such that the integrals converge for $|\varphi|>\pi$, $W_{\theta}(g,N)$ and $W_{n\theta}(g)$ analytically extend $W(g,N)$ and $W_n(g)$ to some maximal domain.

\paragraph{Property~\ref{propLVE:4} and~\ref{propLVE:5}.} We now have the bound $ |  1 -\imath \sqrt{\tfrac{g}{3}}   e^{-\imath \theta}\sigma_i | \ge \cos\frac{\varphi - 2\theta}{2} $ and the Gaussian integral is bounded by:
\begin{equation}
\left| \int \frac{\prod_i e^{-\imath \theta} [d\sigma_i]  }{\sqrt{ \det w^\cT }}  \; e^{ - \frac{1}{2} e^{-2\imath \theta}\sum_{i,j} \sigma_i ( w^\cT)_{ij}^{-1}  \sigma_j }   \right| \le \frac{1}{[ \cos(2\theta)]^{\frac{n}{2} } }     \;,
\end{equation}
and the combinatorics runs as before leading to:
\begin{equation}
  |W_{n\theta}(g)|  \le \frac{(2n-3)!}{(n-1)!}  \frac{1}{\sqrt{ \cos(2\theta)  }} \left| \frac{g}{3
   \sqrt{\cos(2\theta)} \left( \cos\frac{\varphi-2\theta}{2} \right)^2
  } \right|^{n-1}    \;.
\end{equation}
Using again Stirling's formula, we find that $W(g,N)$  is convergent for
$|g| <  \frac{1}{|N|} \frac{3}{2}\;  \sqrt{ \cos(2\theta) }  \; \left( \cos\frac{\varphi-2\theta}{2} \right)^{2} $. This concludes the proof of  Proposition~\ref{prop:LVE+bounds}.
\end{proof}

\subsection{Borel summability of $W_n(g)$ and $W(g,N)$ in $\mathbb{C}_\pi$}
\label{app:BSW}

\begin{proof}[Proof of Proposition~\ref{propW:BS}.] We focus on $W(g,N)$: for $W_n(g)$ one follows the same steps without summing over $n$.
Our starting point is the expression:
\be
\begin{split}
W_\theta(g,N) = & - \frac{N}{2} \Bigg[  e^{ \frac{1}{2} e^{2\imath \theta} \frac{\d }{\d \sigma} \frac{\d }{\d \sigma} } 
   \ln\Big( 1 -\imath \sqrt{\tfrac{g}{3}}  e^{-\imath \theta} \sigma \Big) \Bigg]_{\sigma =0 } - \sum_{n\ge 2} \frac{1}{n!} \left( -\frac{N}{2}\right)^{n} \left( \frac{g}{3} \right)^{n-1}   \\
&\qquad \qquad \times \sum_{ \cT \in T_n }
 \int_{0}^1 \Bigl( \prod_{ (i,j) \in \cT } du_{ij} \Bigr)  
\Bigg[  e^{ \frac{1}{2} e^{2\imath \theta}\sum_{i,j} w_{ij}^{\cT}  \frac{\d }{\d \sigma_i} \frac{\d }{\d \sigma_j} } 
\;\; \prod_{i} \frac{(d_i-1)!}{ \left(   1 -\imath \sqrt{\tfrac{g}{3}}   e^{-\imath \theta}\sigma_i \right)^{d_i} } \Bigg]_{\sigma_i =0 }\; .
\end{split}
\ee
for the analytical continuation $W_\theta(g,N)$ of the free energy.

We are interested in finding a good bound on the Taylor rest term of order $q$ of the expansion of 
$W(g,N)$. This Taylor rest term consists in two pieces. We denote $Q^q_\theta(g)$ the sum over all the trees having at least $q$ edges: 
\begin{equation}
\begin{split}
        Q^q_\theta(g,N) = & - \sum_{n\ge q+1}  \frac{1}{n!} \left( -\frac{N}{2}\right)^{n} \left( \frac{g}{3} \right)^{n-1}    \crcr
    & \qquad \times \sum_{ \cT \in T_n }
 \int_{0}^1 \Bigl( \prod_{ (i,j) \in \cT } du_{ij} \Bigr)  
\Bigg[  e^{ \frac{1}{2} e^{2\imath \theta}\sum_{i,j} w_{ij}^{\cT}  \frac{\d }{\d \sigma_i} \frac{\d }{\d \sigma_j} } 
\;\; \prod_{i} \frac{(d_i-1)!}{ \left(   1 -\imath \sqrt{\tfrac{g}{3}}   e^{-\imath \theta}\sigma_i \right)^{d_i} } \Bigg]_{\sigma_i =0 }\; .
\end{split}
\end{equation}
Due to the overall $g^{n-1}$ factor, such trees contribute only to orders higher than $q$ in the Taylor expansion of $W_\theta(g,N)$, hence are entirely contained in the Taylor rest term. The trees with less than $q$ edges contribute both to the explicit terms of order lower than $q$ in the Taylor expansion and to the rest term. In order to isolate their contribution to the rest term we perform for each of them a Taylor expansion with integral rest up to an appropriate order. We do this by Taylor expanding the Gaussian measure, which generates loop edges, each of which comes equipped with a $g$.
For a tree with $n-1$ edges, $n-q$ explicit loop edges need to be expanded. The Taylor rest term coming from such trees writes then:
\begin{equation}
\begin{split}
 P^q_\theta(g,N) = &  - \frac{N}{2} \int_0^1 \frac{ (1-t)^{q-1} }{(q-1)! }  \left( \frac{d}{dt}\right)^q \Bigg[  e^{ \frac{t}{2} e^{2\imath \theta} \frac{\d }{\d \sigma} \frac{\d }{\d \sigma} } 
   \ln\Big( 1 -\imath \sqrt{\tfrac{g}{3}}  e^{-\imath \theta} \sigma \Big) \Bigg]_{\sigma =0 } \crcr
   & - \sum_{n=2}^{q}   \frac{1}{n!} \left( -\frac{N}{2}\right)^{n} \left( \frac{g}{3} \right)^{n-1}    \sum_{ \cT \in T_n } 
   \int_{0}^1 \Bigl( \prod_{ (i,j) \in \cT } du_{ij} \Bigr)  \int_0^1 dt \; \frac{(1-t)^{q-n}}{(q-n)!} \crcr
 & \qquad \times \left( \frac{d}{dt}\right)^{q-n+1}
\Bigg[  e^{ \frac{t}{2}  e^{2\imath \theta}\sum_{i,j} w_{ij}^{\cT}  \frac{\d }{\d \sigma_i} \frac{\d }{\d \sigma_j} } 
\;\; \prod_{i} \frac{(d_i-1)!}{ \left(   1 -\imath \sqrt{\tfrac{g}{3}}   e^{-\imath \theta}\sigma_i \right)^{d_i} } \Bigg]_{\sigma_i =0 }\; .
\end{split}
\end{equation}
The total rest term of the Taylor expansion of $W_\theta(g,N)$ is the sum $R_{\theta}^q(g,N) = P_\theta^q(g,N) + Q_\theta^q(g,N)$. 

\paragraph{The contribution of large trees $Q_\theta^q(g,N)$.} Using Eq.~\eqref{eq:boundW}, the contribution of large trees is immediately bounded by:
\begin{equation}
|Q_\theta^q(g,N) | \le \frac{1}{\sqrt{\cos(2\theta)}}  \sum_{n\ge q+ 1} \left| \frac{N}{2} \right|^{n} \left| \frac{g}{3} \right|^{n-1} \frac{1}{ \big[  \sqrt {\cos(2\theta) }  (\cos\frac{\varphi-2\theta}{2} )^{2} \big]^{n-1} } 
\; \frac{1}{n(n-1) }  \binom{2n-3}{n-1} \; ,
\end{equation}
and using $ \frac{1}{n(n-1) }  \binom{2n-3}{n-1}  \le 2^{2n-3}$ we get:
\begin{equation}\label{eq:bound1}
\begin{split}
|Q_\theta^q(g,N) | \le & \frac{|N| }{4 \sqrt{\cos(2\theta)}} 
\sum_{n\ge q+1}  \left(  \frac{1}{ \frac{3}{2}  \sqrt {\cos(2\theta) }  (\cos\frac{\varphi-2\theta}{2} )^{2}}  \right)^{n-1} |Ng|^{n-1} \crcr
\le &  \frac{\frac{ |N| }{4 \sqrt{\cos(2\theta)}}}{ 1 - \frac{|N g|}{  \frac{3}{2}  \sqrt {\cos(2\theta) }  (\cos\frac{\varphi-2\theta}{2} )^{2}  } } \left( \frac{1}{  \frac{3}{2}  \sqrt {\cos(2\theta) }  (\cos\frac{\varphi-2\theta}{2} )^{2} } \right)^q |N g|^q \;.
\end{split}
\end{equation}
Observe that this bound  \emph{does not} have a $q!$ growth. This seems much better than expected. In fact this is not surprising: the factorial growth of the rest term comes from the divergent number of graphs in perturbation theory. As the large trees do not need to be further expanded in graphs, they do not generate a factorial growth. In fact it is the small trees that pose a problem.

\paragraph{The contribution of small trees $P_\theta^q(g,N)$.} The contribution of small trees requires more work. We have:
\begin{equation}
\begin{split}
 P^q_\theta(g,N) = & -\frac{N}{2}  \int_0^1 \frac{ (1-t)^{q-1} }{(q-1)! } \Bigg[  e^{ \frac{t}{2} e^{2\imath \theta} \frac{\d }{\d \sigma} \frac{\d }{\d \sigma} } 
  \left( \frac{1}{2} e^{2\imath \theta}  \right)^q   \left( \frac{\d }{\d \sigma}\right)^{2q} \ln\Big( 1 -\imath \sqrt{\tfrac{g}{3}}  e^{-\imath \theta} \sigma \Big) \Bigg]_{\sigma =0 } \crcr
   & - \sum_{n=2}^{q}  \frac{1}{n!} \left( -\frac{N}{2}\right)^n \left( \frac{g}{3} \right)^{n-1}   \sum_{ \cT \in T_n } 
   \int_{0}^1 \Bigl( \prod_{ (i,j) \in \cT } du_{ij} \Bigr)  \int_0^1 dt \; \frac{(1-t)^{q-n}}{(q-n)!} \crcr
 &  \times  
\Bigg[  e^{ \frac{t}{2}  e^{2\imath \theta}\sum_{i,j} w_{ij}^{\cT}  \frac{\d }{\d \sigma_i} \frac{\d }{\d \sigma_j} }  \left(  \frac{1}{2}  e^{2\imath \theta}\sum_{i,j} w_{ij}^{\cT}  \frac{\d }{\d \sigma_i} \frac{\d }{\d \sigma_j}   \right)^{q-n+1}
\;\; \prod_{i} \frac{(d_i-1)!}{ \left(   1 -\imath \sqrt{\tfrac{g}{3}}   e^{-\imath \theta}\sigma_i \right)^{d_i} } \Bigg]_{\sigma_i =0 }\; .
\end{split}
\end{equation}
The first term is bounded by:
\begin{equation}\label{eq:bound2}
\begin{split}
& \frac{|N|}{2^{q+1}}  \left| \int_0^1 \frac{ (1-t)^{q-1} }{(q-1)! } \Bigg[  e^{ \frac{t}{2} e^{2\imath \theta} \frac{\d }{\d \sigma} \frac{\d }{\d \sigma} }   \frac{ (2q-1)!\left( \imath \sqrt{\tfrac{g}{3}  }\right)^{2q}  } { \left(  1 -\imath \sqrt{\tfrac{g}{3}} e^{-\imath \theta }\sigma \right)^{2q} }
 \Bigg]_{\sigma =0 } \right|  
 \crcr 
 & \le \frac{|N|}{2\sqrt{\cos(2\theta)}} \left( \frac{ |g| }{ 6 } \right)^q \frac{(2q-1)!}{q!} \; 
 \frac{1}{ \left(  \cos \frac{\varphi -2\theta }{2} \right)^{2q} } 
 \le \frac{|N|}{ 4 \sqrt{\cos(2\theta)} }  \; (q-1)!\;
 \left(  \frac{1}{ \frac{3}{2} ( \cos \frac{\varphi -2\theta }{2} )^2 } \right)^q    \; |g|^q  \;.
\end{split}
\end{equation}

Observe that, due to the expansion of the $q$ loop edges, this first term displays a $q!$ growth. Thus already the first term has a worse large-order behaviour than the sum over large trees.

We now analyze the contribution to the rest term of the trees with $2\le n \le q$ vertices.
Note that a term has initially $\sum d_i = 2(n-1) $ corners (factors $1 / ( 1 -\imath \sqrt{\tfrac{g}{3}}   e^{-\imath \theta}\sigma_i ) $) and $2(q-n+1)$ derivatives will act on it. The additional derivatives (corresponding to the loop edges) create each a new corner on which subsequent derivatives can act. In total, for each tree, the loop edges generate exactly:
\begin{align}
     [ 2(n-1) ]      [ 2(n-1) + 1 ] \dots [ 2(n-1) + 2 (q-n+1) - 1 ] =\frac{ (2q - 1)! }{(2n - 3)! } \;,
\end{align}
possible terms, each with $ 2(n-1) + 2 (q-n+1) $ corners (factors $1 / ( 1 -\imath \sqrt{\tfrac{g}{3}}   e^{-\imath \theta}\sigma_i ) $). The $w$'s are bounded by $1$ and the sum over trees is done as in 
Eq.~\eqref{eq:boundW} leading to the global bound:
\begin{equation}
\begin{split}
 & \sum_{n=2}^q  \frac{ |N|^{n}|g|^{q}}{ 2^{q+1} 3^{q} } \; \frac{(2q-1)!}{(2n-3)!} \; \frac{1}{(q-n+1)!} \;
  \frac{(n-2)!}{n!  }  \binom{2n-3}{n-1}  
  \left( \frac{1}{  \cos\frac{\varphi - 2\theta}{2} } \right)^{2q} \frac{1}{ [\cos(2\theta)]^{n/2} } \crcr
 & \le \frac{ (2q-1)! \; |g|^q}{ 2^{q+1} 3^q (  \cos\frac{\varphi - 2\theta}{2} )^{2q} }
   \sum_{n=2}^q  \frac{1}{(q-n+1)!n!(n-1)!} \;  \frac{|N|^n}{ [\cos(2\theta)]^{n/2} } \;.
\end{split}
\end{equation}
Now, using $(2q-1)! \le 2^{2q-1}q!(q-1)!$ we get and upper bound:
\begin{equation}\label{eq:bound3}
\begin{split}
& \frac{1}{4} (q-1)! \left( \frac{1} { \frac{3}{2} (\cos\frac{\varphi-2\theta}{2})^2 }  \right)^q |g|^q    
  \sum_{n=2}^q  \frac{q!}{(q-n+1)!n!(n-1)!} \;  \frac{|N|^n}{ [\cos(2\theta)]^{n/2} } \crcr
&  \le \frac{1}{4}    (q-1)! \left( \frac{1} { \frac{3}{2} (\cos\frac{\varphi-2\theta}{2})^2 }  \right)^q |g|^q \; 2^{q} e^{\frac{|N|}{ \sqrt{ \cos(2\theta) } } } \;.
\end{split}
\end{equation}
Adding up Eq.~\eqref{eq:bound1},~\eqref{eq:bound2} and~\eqref{eq:bound3} we get:
\begin{equation}
\begin{split}
|R^q_\theta(g,N)| \le & \frac{\frac{|N|}{4 \sqrt{\cos(2\theta)}}}{ 1 - \frac{|Ng|}{  \frac{3}{2}  \sqrt {\cos(2\theta) }  (\cos\frac{\varphi-2\theta}{2} )^{2}  } } \left( \frac{1}{  \frac{3}{2}  \sqrt {\cos(2\theta) }  (\cos\frac{\varphi-2\theta}{2} )^{2} } \right)^q |Ng|^q \crcr 
& + \frac{1}{4}  (q-1)!\;
 \left(  \frac{1}{ \frac{3}{2} ( \cos \frac{\varphi -2\theta }{2} )^2 } \right)^q    \; |g|^q  
 \left(  \frac{|N|}{  \sqrt{\cos(2\theta)} }  + 2^q  e^{\frac{|N|}{ \sqrt{ \cos(2\theta) } } } \right) \;.
\end{split}
\end{equation}
We opportunistically chose $\theta = \frac{\varphi}{6}$ and using some trivial  bounds we get
\begin{equation}
     |R^q_\theta(g,N)| \le \frac{ \frac{|N|}{ 4 \sqrt{ \cos\frac{\varphi}{3} } } }{1 - \frac{|Ng|}{ \frac{3}{2} (\cos\frac{\varphi}{3})^{5/2} }} \left( \frac{1 }{ \frac{3}{2} (\cos\frac{\varphi}{3})^{5/2} } \right)^q |N g|^q + (q-1)!  \; |g|^q \left( \frac{1 }{ \frac{3}{2} (\cos\frac{\varphi}{3})^{2} } \right)^q
      2^{q} e^{\frac{|N|}{ \cos(\frac{\varphi}{3})^{1/2}  } } \;.
\end{equation}

We are now in the position to prove that $W(g,N)$ is Borel sumable along all the directions in the cut plane
$\mathbb{C}_\pi$ by verifying the conditions of Theorem~\ref{thm:Sokal}, Appendix~\ref{app:Sokal}. This comes about as follows:
\begin{itemize}
\item let us fix some $\alpha\in (-\pi,\pi)$. As we are interested in analyticity and rest bound in some Sokal disk extending op to $\varphi= \alpha \pm \pi/2 $, 
we denote $c = \min\{ \cos \frac{\alpha + \pi/2}{3}) ; \cos \frac{\alpha - \pi/2}{3}) \} > 0 $ (because as $|\alpha|<\pi$).
    
\item $W(g,N)$ can be analytically continued via $W_{\varphi/6}(g,N) $\footnote{Indeed, $W_{\varphi/6}(g,N) = Q^0_{\varphi/6}(g,N)$ and the series defining it converges absolutely in such a disk $|Q^0_{\varphi/6}(g,N)  | \le 1/ ( 4c^{1/2} \nu )$.} to any $g$ in a Sokal disk (with $0$ on its boundary) tilted by $\alpha$, that is $g \in {\rm Disk}^{\alpha}_R = \{ z \mid  \mathrm{Re}(e^{\imath \alpha} / z)  > 1/R  \} $ provided that $ R  = \frac{3}{2|N|} c^{5/2} (1 - \nu) $ for some fixed $\nu>0$.
In this disk the Taylor rest term obeys the uniform bound:
\begin{equation}
 |R^q_\theta(g)|  \le \frac{|N|}{4c^{1/2} \nu } \left(\frac{1}{\frac{3}{2} c^{5/2} }\right)^q |g|^q  
 + (q-1)! |g|^q \left(\frac{2}{\frac{3}{2} c^{2} }\right)^q e^{|N|/c^{1/2}} \;.
\end{equation}
Fixing $K = \max\{ |N|/ ( 4c^{1/2} \nu ) , e^{|N|/c^{1/2}} \}$  and 
$\rho =\min \{ 3/2 \, c^{5/2};  3/4 \, c^{2} ) \} $, we finally obtain uniformly in the Sokal disk:
\begin{equation}
|R^q_\theta(g,N)| \le K \; q! \; \rho^{-q} \; |g|^q \; ,
\end{equation}
hence the rest obeys the bound in Eq.~\eqref{eq:NS-bound} and from Theorem~\ref{thm:Sokal} we conclude that $W(g,N)$ is Borel summable along $\alpha$.

Note that Borel summability is lost in the $N\to \infty$ limit. 
    
\end{itemize}
This concludes the proof of Proposition~\ref{propW:BS}.
\end{proof}

\subsection{Transseries expansion of $W_n(g)$ and $W(g,N)$}
\label{app:Trans}

\begin{proof}[Proof of Proposition~\ref{prop:TSW}]
An explicit expression for $W_n(g)$ is obtained by combining the M\"obis inversion formula in Eq.\eqref{eq:Moeb} with the transseries expansion of $Z_n(g)$ in Eq.~\eqref{eq:transsereiesZn}. In order to prove Proposition~\ref{prop:TSW}, we have to use standard sum and product manipulation tricks and factor the powers of the transseries monomial $e^{\f{3}{2g}}$ in front. Although the manipulations are not complicated, the expressions are very lengthy and we will introduce some bookkeeping notation to keep the formulas readable.

\paragraph{Property~\ref{propTSW:1}.}
Combining the M\"obis inversion formula \eqref{eq:Moeb} with the transseries expansion of $Z_n(g)$ in Eq.~\eqref{eq:transsereiesZn} leads to:
\begin{equation}\label{eq:Prop5proof1}
 W_n (g) \simeq \sum_{k = 1}^n (-1)^{k-1} (k-1)! \sum_{ \substack{n_1, \ldots, n_{n-k+1} \ge 0 \\ \sum in_i = n ,\, \sum n_i = k }}  \frac{n!}{\prod_i n_i! (i!)^{n_i}}
\prod_{i=1}^{n-k+1} \Bigg(Z_i^{\rm pert.}(g) + \eta \, e^{\f{3}{2g}} Z_i^{(\eta)}(g)\Bigg)^{n_i}
\;,
\end{equation}
with:
\be
Z_i^{\rm pert.}(g)=\sum_{a=0}^\infty \left(-\frac{2g}{3}\right)^a \; G(a;i) \; ,
\qquad
G(a;i) = \f{(2a)!}{2^{2a} a!} \sum_{\substack{a_1,\ldots,a_{2a-i+1} \ge 0  \\ \sum k  a_k=2a , \; \sum a_k=i}} \f{(-1)^i i!}{\prod_k k^{a_k}a_k!} \; ,
\ee
and:
\be\begin{aligned}
Z_i^{(\eta)}(g)=   \f{\imath }{\sqrt{2\pi}}\sqrt{ \f{  g}{3}}
& \sum_{a=0}^{\infty} \sum_{b=0}^{i} \f{1}{a!} \left(\f{g}{6}\right)^a \binom{i}{b} \f{d^b \G(z)}{dz^b}\Big|_{z=2a+1} \crcr
& \tau \left[  \left(\ln\left(\tfrac{g}{3}\right)  \right)^{i-b} - \left(\ln\left(\tfrac{g}{3}\right) + \tau 2  \pi \imath \right)^{i-b} \right] \; ,
\end{aligned}\ee
where we used $\left[\left(\ln( e^{\imath \tau \pi } \f{g}{3}) -i\pi \right)^{i-b}
-\left(\ln( e^{\imath \tau \pi } \f{g}{3})  +i\pi \right)^{i-b}\right]
=\tau \left[  \left(\ln\left(\tfrac{g}{3}\right)  \right)^{i-b} - \left(\ln\left(\tfrac{g}{3}\right) + \tau 2  \pi \imath \right)^{i-b} \right]
$ as $\tau = \pm$. 
Using now $
\tau \left[  \left(\ln\left(\tfrac{g}{3}\right)  \right)^{n}
-\left(\ln\left(\tfrac{g}{3}\right) + \tau 2  \pi \imath \right)^{n} \right] = -\tau \sum_{c=0}^{n-1} \binom{n}{c} \left(\ln\left(\tfrac{g}{3}\right)\right)^{c}\left(\tau 2  \pi \imath \right)^{n-c}$, commuting the sums over $b$ and $c$, and combining the binomials, this can further be written as:
\begin{equation}\begin{aligned}
Z_i^{(\eta)}(g) &=
\f{\imath}{\sqrt{2\pi}}\sqrt{\f{g}{3}}\; \sum_{a=0}^\infty \sum_{b=0}^{i-1}
\f{1}{a!} \left(\f{g}{6}\right)^a \binom{i}{b} \f{d^b \G(z)}{dz^b}\Big|_{z=2a+1}
(-\tau) \sum_{c=0}^{i-1-b} \binom{i-b}{c} \left(\ln\left(\tfrac{g}{3}\right)\right)^{c}\left(\tau 2  \pi \imath \right)^{i-b-c} \\
&= \sqrt{2\pi} \sqrt{\f{g}{3}} \; \sum_{a=0}^\infty \sum_{c=0}^{i-1} \left(\f{g}{6}\right)^a \left(\ln\left(\tfrac{g}{3}\right)\right)^{c} \; G(a,c;i)
\; ,
\end{aligned}\end{equation}
with:
\begin{equation}
G(a,c;i) = \sum_{b=0}^{i-1} \left(\imath\tau 2  \pi \right)^{i-1-b-c} \f{i!}{a!\,b!\,c!\,(i-b-c)!} \f{d^b \G(z)}{dz^b}\Big|_{z=2a+1}
\; , \quad\text{and}\quad Z_0^{(\eta)}=0
\; .
\end{equation}
First, we pull the transseries monomial to the front in Eq.~\eqref{eq:Prop5proof1}:
\begin{equation}\label{eq:Prop5proof2}\begin{aligned}
& W_n(g) \simeq \sum_{k = 1}^n (-1)^{k-1} (k-1)! \sum_{ \substack{n_1,\ldots ,n_{n-k+1} \ge 0 \\ \sum in_i = n ,\, \sum n_i = k }}  \frac{n!}{\prod_i n_i! (i!)^{n_i}} \ \prod_{i=1}^{n-k+1} \big(Z_i^{\rm pert.}(g) + \eta e^{\frac{3}{2g}} Z_i^{(\eta)}(g)\big)^{n_i}
\\
&= \sum_{p=0}^n e^{\frac{3}{2g}p} \sum_{\substack{k = p \\ k+p\geq 1}}^{n} (-1)^{k-1} (k-1)! \kern-1em \sum_{ \substack{n_1,\ldots,n_{n-k+1} \ge 0 \\ \sum in_i = n ,\, \sum n_i = k }} \ \sum_{\substack{\{0\leq p_i\leq n_i\} \\ _{i=1,\dots,n-k+1} \\ \sum p_i=p}} 
\frac{n!\prod_{i=1}^{n-k+1}  \big(Z_i^{\rm pert.}(g)\big)^{n_i-p_i} \big(\eta Z_i^{(\eta)}(g)\big)^{p_i} }{\prod_i (n_i-p_i)!p_i! (i!)^{n_i}}
\; .
\end{aligned}\end{equation}
This is a transseries in $g$, with $Z_i^{(\eta)}(g)$ carrying also powers of $\sqrt{g}$ and $\ln(g)$. In a second step, we want to make this statement manifest and use the expressions for $Z_i^{\rm pert.}$ and $Z_i^{(\eta)}$ to calculate:
\begin{equation}\begin{aligned}
\prod_{i\ge 1}  &\big(Z_i^{\rm pert.}(g)\big)^{n_i-p_i} \big(\eta Z_i^{(\eta)}(g)\big)^{p_i} = \prod_{i\ge 1} \Bigg( \sum_{a^i_1,\dots a^i_{n_i-p_i}\geq 0} \prod_{j=1}^{n_i-p_i} \left(-\f{2g}{3}\right)^{a^i_j} G(a^i_j;i) \Bigg)
\\ & \qquad\qquad\quad \cdot
\Bigg( \Big(\eta\sqrt{2\pi} \sqrt{\f{g}{3}}\Big)^{p_i} \sum_{a^i_1,\dots a^i_{p_i}\geq 0} \sum_{c^i_1,\dots c^i_{p_i}= 0}^{i-1} \prod_{j=1}^{p_i} \left(\f{g}{6}\right)^{a^i_j} \left(\ln\left(\tfrac{g}{3}\right)\right)^{c^i_j} G(a^i_j,c^i_j;i) \Bigg) \;,
\\
=& \Big(\eta\sqrt{2\pi} \sqrt{\f{g}{3}}\Big)^{\sum_{i\geq 1}p_i} \sum_{\substack{l\geq 0\\ \sum_{i}(i-1)p_i\geq l\geq 0}} g^l\left(\ln\left(\tfrac{g}{3}\right)\right)^{l'} \sum_{\substack{\{a^i_j\geq 0\}^{i=1,\dots,n-k+1}_{j=1,\dots,n_i} \\ \sum_i\sum_j a^i_j=l}} \sum_{\substack{\{i-1\geq c^i_j\geq 0\}^{i=1,\dots,n-k+1}_{j=1,\dots,p_i} \\ \sum_i\sum_j c^i_j=l'}}
\\ & \qquad\qquad\quad \cdot
\left(\f{1}{6}\right)^{\sum_{i\geq 1}\sum_{j=1}^{p_i}a^i_j} \left(-\f{2}{3}\right)^{\sum_{i\geq 1}\sum_{j=p_i+1}^{n_i}a^i_j} 
\prod_{i\geq 1} \Bigg(\prod_{j=1}^{p_i} G(a^i_j,c^i_j;i)\Bigg)\Bigg(\prod_{j=p_i+1}^{n_i} G(a^i_j;i)\Bigg)
\; .
\end{aligned}\end{equation}
Finally, inserting this into Eq.~\eqref{eq:Prop5proof2} we obtain the transseries expansion of $W_n(g)$, organized into instanton sectors:
\be
W_n(g)\simeq \sum_{p=0}^n e^{\frac{3}{2g}p} \Big(\eta\sqrt{2\pi} \sqrt{\f{g}{3}}\Big)^{p} \sum_{l\geq 0,\, n-p\geq l'\geq 0} g^l\left(\ln\left(\tfrac{g}{3}\right)\right)^{l'}
    W^{(p)}_{n;l,l'} \,,
\ee
with 
\begin{equation}\begin{aligned}
& W^{(p)}_{n;l,l'}  =
\smashoperator[l]{ \sum_{\substack{k = p \\ k+p\geq 1}}^{n} } (-1)^{k-1} (k-1)! \sum_{ \substack{n_1,\ldots,n_{n-k+1} \ge 0 \\ \sum in_i = n ,\, \sum n_i = k }} \sum_{\substack{\{0\leq p_i\leq n_i\} \\ _{i=1,\dots,n-k+1} \\ \sum p_i=p}}
\frac{n!}{\prod_i (n_i-p_i)!p_i! (i!)^{n_i}}
\\ & \qquad\qquad
\smashoperator[l]{ \sum_{\substack{\{a^i_j\geq 0\}^{i=1,\dots,n-k+1}_{j=1,\dots,n_i} \\ \sum_i\sum_j a^i_j=l}} } \sum_{\substack{\{i-1\geq c^i_j\geq 0\}^{i=1,\dots,n-k+1}_{j=1,\dots,p_i} \\ \sum_i\sum_j c^i_j=l'}}
\left(\f{1}{6}\right)^{\sum_{i=1}^{n-k+1}\sum_{j=1}^{p_i}a^i_j} \left(-\f{2}{3}\right)^{\sum_{i=1}^{n-k+1}\sum_{j=p_i+1}^{n_i}a^i_j} 
\\ & \hspace{9.5em}
\prod_{i=1}^{n-k+1} \Bigg(\prod_{j=1}^{p_i} G(a^i_j,c^i_j;i)\Bigg)\Bigg(\prod_{j=p_i+1}^{n_i} G(a^i_j;i)\Bigg)
\; ,
\end{aligned}\end{equation}
as advertised. 

\paragraph{Property~\ref{propTSW:2}.}
It remains to transseries expansion for the full free energy $W(g,N)$. In this case, the expressions are simpler because the $Z_i$ have been summed.
Starting from the relation between $W(g,N)$ and the $W_n(g)$:
\begin{equation}\begin{aligned}
W(g,N) &= \sum_{n\ge 1} \frac{1}{n!} \left( -\frac{N}{2}\right)^n W_n(g) 
\\
&\simeq \sum_{n\ge 1} \frac{ \left( -\frac{N}{2}\right)^n }{n!} \sum_{k = 1}^n (-1)^{k-1} (k-1)! \kern-1em \sum_{ \substack{n_1,\ldots,n_{n-k+1} \ge 0 \\ \sum in_i = n ,\, \sum n_i = k }}  \tfrac{n!}{\prod_i n_i! (i!)^{n_i}} \! \prod_{i=1}^{n-k+1} \big(Z_i^{\rm pert.}(g) + \eta e^{\frac{3}{2g}} Z_i^{(\eta)}(g)\big)^{n_i}
\\
&=  \sum_{k\ge 1} (-1)^{k-1} (k-1)! \sum_{ \substack{n_1,n_2\ldots \ge 0 \\ \sum n_i = k }} \ \prod_{i\ge 1} \frac{1}{n_i!} \Big(\frac{1}{i!}\big(-\frac{N}{2}\big)^i\big(Z_i^{\rm pert.}(g) + \eta e^{\frac{3}{2g}} Z_i^{(\eta)}(g)\big)\Big)^{n_i} \;,
\end{aligned}\end{equation}
which becomes:
\begin{equation}\label{eq:circularW} \begin{aligned}
& \sum_{k\ge 1} (-1)^{k-1} (k-1)!\; \frac{1}{k!} \Big( \sum_{i\ge 1} \frac{1}{i!}\big(-\frac{N}{2}\big)^i\big(Z_i^{\rm pert.}(g) + \eta e^{\frac{3}{2g}} Z_i^{(\eta)}(g)\big) \Big)^k
\\
&= \sum_{k\ge 1} (-1)^{k-1} (k-1)!\; \frac{1}{k!}\big( (Z^{\rm pert.}(g,N)-1 + \eta e^{\frac{3}{2g}} Z^{(\eta)}(g,N) ) \big)^k
\\
&= \sum_{k\ge 1} \frac{(-1)^{k-1}}{k} (Z^{\rm pert.}(g,N)-1 + \eta e^{\frac{3}{2g}} Z^{(\eta)}(g,N) )^k
= \ln(Z^{\rm pert.}(g,N) + \eta e^{\frac{3}{2g}} Z^{(\eta)}(g,N) ) \; ,
\end{aligned}\end{equation}
we unsurprisingly recover, that formally $W(g,N)=\ln(Z(g,N))$. Here, we used the notation (cf. Eq.~\eqref{eq:fullZ}) \linebreak[4]\mbox{$Z(g,N)=Z^{\rm pert.}(g,N) + \eta e^{\frac{3}{2g}} Z^{(\eta)}(g,N)$}, with:
\be\begin{aligned}
Z^{\rm pert.}(g,N)&=\sum_{n=0}^{\infty} \, \frac{\Gamma(2n+N/2) }{2^{2n}n! \, \Gamma(N/2) }    \; \left(- \frac{2g}{3}\right)^n
\;, \\
Z^{(\eta)}(g,N)&= e^{\imath \tau\pi(1-\frac{N}{2})}  \;  \sqrt{2\pi}  \left(  \frac{ g} {3} \right)^{\frac{1-N}{2} } \sum_{q=0}^{\infty} \frac{1}{ 2^{2q} q! \; \Gamma(\frac{N}{2} -2q ) } \left( \frac{2 g}{3} \right)^q \;.
\end{aligned}\ee
As before, the expansion into instanton sectors can be made manifest, by pulling the transseries monomials to the front in $W(g,N)$. We start from the second to last line in Eq.~\eqref{eq:circularW}), that can also be written as:
\begin{equation}\label{eq:Prop5proof3}\begin{aligned}
W(g,N)&\simeq \sum_{k\ge 1} (-1)^{k-1} (k-1)! \sum_{\substack{p,q\geq 0\\ p+q=k}} \frac{1}{p!q!} \big(Z^{\rm pert.}(g)-1\big)^{q} \big(\eta e^{\frac{3}{2g}}Z^{(\eta)}(g)\big)^{p}
\\ &= \sum_{p\ge 0} e^{\frac{3}{2g}p} \sum_{\substack{q\geq 0\\ p+q\geq 1}} (-1)^{p+q-1} \frac{(p+q-1)!}{p!q!} \big(Z^{\rm pert.}(g)-1\big)^{q} \big(\eta Z^{(\eta)}(g)\big)^{p}
\; .
\end{aligned}\end{equation}
Next, we use the expressions for $Z^{\rm pert.}$ and $Z^{(\eta)}$ to calculate:
\begin{equation}\begin{aligned}
&\big(Z^{\rm pert.}(g)-1\big)^{q} \big(\eta Z^{(\eta)}(g)\big)^{p} 
\\ 
&= \Bigg( \sum_{n=1}^{\infty} \, \tfrac{\Gamma(2n+N/2) }{2^{2n}n! \, \Gamma(N/2) }  \left(- \tfrac{2g}{3}\right)^n \Bigg)^{q}
\Big(\eta\sqrt{2\pi}e^{\imath\tau\f{\pi}{2}} \left(\tfrac{e^{\imath\tau\pi}g}{3}\right)^{\frac{1-N}{2}}\Big)^{p}
\Bigg( \sum_{m\ge 0} \tfrac{1}{ 2^{2m} m! \; \Gamma(\frac{N}{2} -2m ) } \left( \tfrac{2g}{3} \right)^m \Bigg)^{p}
\\
&= \Big(\eta\sqrt{2\pi}e^{\imath\tau\f{\pi}{2}} \left(\tfrac{e^{\imath\tau\pi}g}{3}\right)^{\frac{1-N}{2}}\Big)^{p}
\kern-.5em
\sum_{\substack{n_1,\dots,n_q\geq 1 \\ m_1,\dots,m_p\geq 0 }} \kern-.5em
\Bigg( \prod_{i=1}^{q} \tfrac{\Gamma(2n_i+\frac{N}{2}) }{2^{2n_i}n_i! \, \Gamma(\frac{N}{2}) } \left(- \tfrac{2g}{3}\right)^{n_i}\! \Bigg)
\Bigg( \prod_{j=1}^{p} \tfrac{1}{ 2^{2m_j} m_j! \; \Gamma(\frac{N}{2} -2m_j ) } \left( \tfrac{2g}{3} \right)^{m_j} \!\Bigg)
\\
&= \Big(\eta\sqrt{2\pi}e^{\imath\tau\f{\pi}{2}} \left(\tfrac{e^{\imath\tau\pi}g}{3}\right)^{\frac{1-N}{2}}\Big)^{p}
\sum_{l\ge 0} \left(- \tfrac{2g}{3}\right)^{l}
\sum_{\substack{n_1,\dots,n_q\geq 1 \\ m_1,\dots,m_p\geq 0 \\ \sum n_i+\sum m_j=l }} \kern-1em
\Bigg( \prod_{i=1}^{q} \tfrac{\Gamma(2n_i+\frac{N}{2}) }{2^{2n_i}n_i! \, \Gamma(\frac{N}{2}) }\Bigg)
\Bigg( \prod_{j=1}^{p} \tfrac{(-1)^{m_j}}{ 2^{2m_j} m_j! \; \Gamma(\frac{N}{2} -2m_j)}\Bigg) \; .
\end{aligned}\end{equation}
Finally, inserting this into Eq.~\eqref{eq:Prop5proof3} we obtain the transseries expansion of the free energy, organized into instanton sectors:
\begin{equation}\begin{aligned}
W(g,N) &= \sum_{p\geq 0} e^{\frac{3}{2g}p} \; \Big(\eta\sqrt{2\pi}e^{\imath\tau\f{\pi}{2}} \left(\tfrac{e^{\imath\tau\pi} g}{3}\right)^{\frac{1-N}{2}}\Big)^{p} \; \sum_{l\ge 0} \left(- \tfrac{2g}{3}\right)^{l}
\\ & \cdot
\Bigg(
\sum_{\substack{q\geq 0\\ p+q\geq 1}} (-1)^{p+q-1} \tfrac{(p+q-1)!}{p!q!}
\kern-.5em
\sum_{\substack{n_1,\dots,n_q\geq 1 \\ m_1,\dots,m_p\geq 0 \\ \sum n_i+\sum m_j=l }} \kern-1em
\Bigg( \prod_{i=1}^{q} \tfrac{\Gamma(2n_i+N/2) }{2^{2n_i}n_i! \, \Gamma(N/2) }\Bigg)
\Bigg( \prod_{j=1}^{p} \tfrac{(-1)^{m_j}}{ 2^{2m_j} m_j! \; \Gamma(N/2 -2m_j)}\Bigg)
\Bigg)
\; ,
\end{aligned}\end{equation}
which has the desired form
\be
W(g)\simeq \sum_{p\ge 0} e^{\frac{3}{2g}p} \; \Big(\eta\sqrt{2\pi}e^{\imath\tau\f{\pi}{2}} \left(\frac{e^{\imath\tau\pi} g}{3}\right)^{\frac{1-N}{2}}\Big)^{p} \; \sum_{l\ge 0} \left(- \frac{2g}{3}\right)^{l} W^{(p)}_{l}(N) \,,
\ee
and we can read up the coefficients $W^{(p)}_{l}$.

This concludes the proof of Proposition~\ref{prop:TSW}.
\end{proof}

\section{The BKAR formula}
\label{app:BKAR}

The Brydges-Kennedy-Abdesselam-Rivasseau (BKAR) \cite{Brydges:1987,Abdesselam:1994ap} forest formula is a Taylor formula for functions of several variables. 
Due to its symmetry and positivity properties it is very well adapted for nonperturbative quantum field theory. 

Let us consider a set of $n$ points labeled $i=1\dots n$, which we identify with the set of vertices of the complete graph $\cK_n$. 
The set of unordered pairs of such points has $n(n-1)/2$ elements $e = (i,j)$ for $1\le i , j \le n , \;  i\neq j$ and can be identified with the set of edges of $\cK_n$. Let us consider a smooth (and arbitrarily derivable) 
function $f: [0,1]^{n(n-1)/2} \to \mathbb{R}$ depending on the edge variables $x_e \equiv x_{ij}, \; e=(i,j)$. 
\begin{theorem}\label{thm:forest}
[The Forest Formula, \cite{Abdesselam:1994ap,Brydges:1987}] We have (with the convention that empty products are $1$):
\begin{equation}
    f( 1,\dots 1)  =      \sum_{\cF}  \underbrace{\int_0^1  \dots \int_0^1}_{|\cF| \; {\rm times}} \left( \prod_{e \in \cF}du_e \right) 
 \;  \left[ \left (  \prod_{e\in \cF}  \frac{\partial  }{ \partial x_{e}  }  \right) f \right] \Big( w^\cF_{kl} (u_{\cF}) \Big) \; ,
\end{equation}  
where:
\begin{itemize}
\item the sum runs over the forests\footnote{Acyclic edge-subgraphs of the complete graph $\cK_n$.} 
$\cF$ drawn over the $n$ labeled vertices $ i $, including the empty forest (having no edge). 
To each edge $e\in \cF$ we attribute a variable $u_e$ that is integrated from $0$ to $1$ and we denote $u_{\cF}=\{ u_e \, | \, e\in \cF\}$. 

\item the derivative $\left (  \prod_{e\in \cF}  \frac{\partial  }{ \partial x_{e}  }  \right) f $
is evaluated at the point:
\begin{equation}
 w^\cF_{kl}( u_{\cF})  =    \inf_{ e' \in  P^{\cF}_{k -l} } \{  u_{e'} \} \; ,
\end{equation}
where $  P^{\cF}_{k - l} $ denotes the unique path in the forest $\cF$ joining the vertices $k$ and $l$, and 
the infimum is set to zero if such a path does not exist.
\end{itemize}

Setting by convention $w^\cF_{kk}( u_{\cF})  \equiv 1 $, for any assignment of tree edge variables $0 \le u_{\cF}\le 1$ the symmetric $n\times n$ matrix
$ W^{\cF}(u^{\cF}) = \bigl( w^\cF_{kl} ( u_{\cF} )\bigr)_{1\le k,l \le n}$ is \emph{positive}. 
\end{theorem}

The most subtle point in this formula is that $W^{\cF}(u^{\cF})$ is a positive matrix. To see this we proceeded as follows. A forest $\cF$ divides the complete graph $\cK_n$ into several connected components (or blocks) corresponding to the trees in the forest. 
For instance, if $\cF$ is the empty forest the blocks are all singletons consisting in a unique vertex per block. 
For any forest $\cF$, the matrix:
\begin{equation}
 B^{\cF}_{kl} = \begin{cases}
                    1 \; , \;\; & \text{ if }k,l \text{ belong to the same block of } \cF \\
                    0 \; , \;\; & \text{ otherwise }
                \end{cases} \; ,
\end{equation}
is positive. Indeed, denoting $b \subset \cF$ the blocks of $\cF$ and $k\in b$ the vertices in the block $b$:
\begin{equation}
 \sum_{k,l}  B^{\cF}_{kl} a_ka_l = \sum_{b \subset \cF } \Bigl( \sum_{k\in b} a_k\Bigr)^2 \; .
\end{equation}

Let us denote the number of edges in $\cF$ by $q \equiv |\cF|$. We order the edges of $\cF$ in decreasing order of  their parameters $u$:
 \begin{equation}
 1\ge  u_{e_1} \ge u_{e_2} \ge \dots u_{e_{q}}\ge 0 \; .
  \end{equation}
Adding edges one by one starting from the highest edge we obtain a family of subforests of $\cF$:
 \begin{equation}
 \cF^0 = \emptyset \; , \;\; \cF^1 = \{e_1\} \; , \;\; \cF^{2} = \{e_1,e_2\} \; , \dots \; , \;\; \cF^{q} =\{e_1,\dots e_q \}= \cF \; ,
 \end{equation}
and the matrix $W^{\cF}(u_\cF)$ writes as:
 \begin{equation}
 W^{\cF}( u_{\cF} ) = (1 - u_{e_1}) B^{\cF^0} + (u_{e_1}-u_{e_2}) B^{\cF^1} + \dots + u_{e_q}  B^{ \cF^q } \; .
 \end{equation}
Indeed, if $i$ and $j$ do not belong to the same block of $\cF^q=\cF$, then they do not belong to the same block in any of the $\cF^s$, $s\le q$ and none of the terms above contribute, hence $w_{ij}^{\cF}(u_\cF) = 0$. If, on the other hand, $i$ and $j$ belong to the same block of $\cF$, then:
 \begin{equation}
 \Big[ (1 - u_{e_1}) B^{\cF^0} + (u_{e_1}-u_{e_2}) B^{\cF^1} + \dots + u_{e_q} B^{\cF^q} \Big]_{ij} = u_{e_s} \; ,
 \end{equation}
where $s$ is such that $i$ and $j$ belong to the same block of $\cF^s$, but belong to two different blocks of $\cF^{s-1}$. As $u_{e_s}\le u_{e_{s-1}} \le u_{e_{s-2}} \le \dots $
it follows that $u_{e_s}$ is the infimum of the $u$s in the unique path in $\cF^s$ joining $i$ and $j$, hence it 
is also the infimum of the $u$s in the unique path in $\cF$ joining $i$ and $j$. The matrix $ W^{\cF}( u_{\cF} )$ is a convex combination of positive matrices, hence it is itself  positive.

\subsection{Feynam graphs and $W(g,N)$}
\label{app:FeynGr}

Each term in the convergent series in Eq.~\eqref{eq:W_BKAR} can be further expanded in a formal Taylor series in the coupling constant. The series thus obtained is asymptotic to $W(g,N)$ as long as Eq.~\eqref{eq:W_BKAR} converges hence in a cardioid domain of the cut complex plane $\mathbb{C}_\pi$ (see Proposition~\ref{prop:LVE+bounds}) which sweeps all the directions in $\mathbb{C}_\pi$. In this range of $g$ no singularity of the integrands crosses the real axis, which is the steepest-descent contour of the exponentials in Eq.~\eqref{eq:W_BKARterm}.  

The asymptotic series of $W(g,N)$ in the first Riemann sheet is well known to be the formal sum over connected Feynman graphs of amplitudes. However, due to the presence of the $w$ parameters and the integrals over the $u$s, it is not exactly transparent how this comes about starting from Eq.~\eqref{eq:W_BKAR}. 
The fact that Eq.~\eqref{eq:W_BKAR} does indeed reproduce the Feynman graph expansion has been proven in 
\cite{Rivasseau:2013ova}. We sketch below how this comes about.

\paragraph{Hepp sectors.}
To any graph $G$ with vertices
labeled $i,i=1,\dots n$ one can associate a \emph{characteristic} function depending on edge variables $f(x_{ij}) = \prod_{(i,j)\in G} x_{ij} $ where the product runs over the edges of $G$. The forest formula applied to this function 
yields:
\begin{equation}
1 =      \sum_{\cF\subset G}  \underbrace{\int_0^1  \dots \int_0^1}_{|\cF| \; {\rm times}} \left( \prod_{ (i,j) \in \cF}du_{ij} \right) \; 
 \left[  \prod_{ (k,l)\notin \cF}   w^\cF_{kl} (u_{\cF}) \right] \; .
\end{equation}
Remark that the derivative of the characteristic function is nonzero only if the forest $\cF$ is made of edges of $G$ (which we signify by $\cF \subset G$).
Furthermore, if $\cF$ has more than one block, then one of the $w$s in the product is set to zero. It follows that only 
trees contribute:
\begin{equation}
 1 = \sum_{\cT\subset G} \underbrace{\int_0^1  \dots \int_0^1}_{|\cT| \; {\rm times}} \left( \prod_{ (i,j) \in \cT}du_{ij} \right) \; 
 \left[  \prod_{ (k,l)\notin \cT}   w^\cT_{kl} (u_{\cT}) \right] \; .
\end{equation}
This formula defines normalized weights associated with the graph $G$ and the spanning trees $\cT\subset G$:
\begin{equation}
 w(G,\cT) = \underbrace{\int_0^1  \dots \int_0^1}_{|\cT| \; {\rm times}} \left( \prod_{ (i,j) \in \cT}du_{ij} \right) \; 
 \left[  \prod_{ (k,l)\notin \cT}   w^\cT_{kl} (u_{\cT}) \right] \; , \qquad 
 \sum_{\cT \subset G} w(G,\cT) =1 \; ,
\end{equation}
which admit a striking combinatorial interpretation.
We define a \emph{Hepp sector} as a total ordering $\pi$ of the edges of $G$, that is a bijection 
$ \pi :   E(G) \to \{1,\dots |E(G)|\} $,
where $E(G)$ is the set of edges of $G$. For any $\pi$ \emph{the leading spanning tree} in $\pi$, denoted 
$\cT(\pi)$, is the tree such that $ \sum_{e\in \cT(\pi)} \pi(e)$
is minimal. The tree $\cT(\pi)$ is obtained by Kruskal's greedy algorithm: 
at each step one adds the edge $e\in G$ with minimal $\pi(e)$ that does not form a loop. 

\begin{lemma}\label{lem:Heppsect} We have:
\begin{equation}
w(G,\cT) =\frac{N(G,\cT)}{|E(G)|!} \;, 
\end{equation}
where  $N( G,\cT)$ is the number of sectors $\pi$
such that $\cT(\pi)=\cT$, that is $w(G,\cT) $ is the percentage of Hepp sectors in which $\cT$ is the leading spanning tree of $G$.
\end{lemma}

\begin{proof} Let us define the function:
\begin{equation}
  \chi(u_{ E(G) \setminus \cT} \le u_{\cT} )
  \begin{cases}
1 \; , \text{ if } \forall (i,j)\in E(G) \setminus \cT  \; , \quad                                                                 u_{ij}\le \inf_{ (k,l) \in P^{\cT}_{i\to j} } u_{kl} \\
0 \; , \text{ otherwise }
\end{cases} \; .
\end{equation}
On the one hand we have:
\begin{equation}
 w(G,\cT) = \int_{0}^1 \left( \prod_{e\in E(G)} du_e \right) \;  \chi(u_{ G \setminus \cT} \le u_{\cT} )   \;,
\end{equation}
as, at any fixed $\{u_{e},e\in \cT\}$ the integral over the loop edge variable $u_{kl}$ yields $w^{\cT}_{kl}(u_{\cT})$. 

We now split the integration interval according to Hepp sectors. In the sector $\pi$ corresponding to $u_{\pi^{-1}(1)} > u_{\pi^{-1}(2)} >\dots$
the characteristic function $\chi$ tests whether every loop edge $(i,j)$ has smaller $u_{ij}$ than $\inf\{u_{kl}\}$ in the tree path
$P^{\cT}_{i\to j}$ connecting $i$ and $j$. This is true if and only if $\cT$ is the leading spanning tree in $\pi$:
\begin{equation}
\begin{split}
& \int_{0}^1 \left( \prod_{e\in \cE(G)} du_e \right) \;  \chi(u_{ G \setminus \cT} \le u_{\cT} )  \crcr
& =  
\sum_{\pi} \int_0^1 du_{\pi^{-1}(1)} \int_{0}^{u_{\pi^{-1}(1)}} du_{\pi^{-1}(2)} \dots 
\int_{0}^{u_{\pi^{-1}( |E(G)|-1) } } 
du_{\pi^{-1} ( |E(G)| ) } \;\; \chi(u_{ G \setminus \cT} \le u_{\cT} )  \crcr
& = \sum_{\pi, \cT(\pi) = \cT} \int_0^1 du_{\pi^{-1}(1)} \int_{0}^{u_{\pi^{-1}(1)}} du_{\pi^{-1}(2)} \dots 
\int_{0}^{u_{\pi^{-1}( |E(G)|-1) } } 
du_{\pi^{-1} ( |E(G)| ) } 
 = \frac{1}{ |E(G)| !  }\sum_{\pi, \cT(\pi) = \cT} 1
\;.
\end{split}
\end{equation}
\end{proof}

\begin{lemma}[Asymptotic series \cite{Rivasseau:2013ova}.] \label{prop:AsymptoticSeries}
The asymptotic expansion at zero of the free energy $W(g,N)$ in the cut plane $g\in\mathbb{C}_\pi$ 
is the formal sum over connected Feynman graphs of Feynman amplitudes.
\end{lemma}
\begin{proof} Taylor expanding the Gaussian integrals in Eq.~\eqref{eq:W_BKAR} to infinity yields:
\be\label{eq:Wintermed}
\begin{split}
W (g,N) & \simeq 
- \tfrac{N}{2}  \Bigg[ \sum_{l\ge 0} \tfrac{1}{l!} \left( \tfrac{1}{2}  \tfrac{\d }{\d \sigma} \tfrac{\d }{\d \sigma} \right)^{l}   \ln\Big( 1 -\imath \sqrt{\tfrac{g}{3}}   \sigma \Big) \Bigg]_{\sigma =0 } - \sum_{n\ge 2} \frac{\left( -\frac{N}{2}\right)^n}{n!} \left( \tfrac{g}{3} \right)^{n-1}  \sum_{ \cT \in T_n } \int_{0}^1 \Bigl( \prod_{ (i,j) \in \cT } du_{ij} \Bigr)   \\
&  \qquad \times 
\Bigg[   \sum_{ \{ l_{i} \ge 0 \} } \tfrac{1}{ l_{i} !}\left( \tfrac{1}{2}\tfrac{\d }{\d \sigma_i} \tfrac{\d }{\d \sigma_i} \right)^{l_{i}} 
\sum_{ \{ l_{ij} \ge 0 \}_{i \le j} } \tfrac{1}{ l_{ij} !}\left( w_{ij}^{\cT}  \tfrac{\d }{\d \sigma_i} \tfrac{\d }{\d \sigma_j} \right)^{l_{ij}} 
\;\; \prod_{i} \tfrac{(d_i-1)!}{ \left(   1 -\imath \sqrt{\tfrac{g}{3}}   \sigma_i \right)^{d_i} } \Bigg]_{\sigma_i =0 } \crcr
&  =\sum_{n\ge 1} \frac{1}{n!}\left( -\frac{N}{2}\right)^n  \sum_{ \cT \in T_n } \sum_{\{ l_{i} \ge 0\}, \{l_{ij}\ge 0\}_{i<j}}
 \frac{\prod_i(d_i + l_i + \sum_{j}l_{ij}-1)!} { (\prod_i 2^{l_i}l_i!) (\prod_{i<j} l_{ij}!) } \crcr
& \qquad \qquad \left[ \int_0^1  \left( \prod_{ (i,j) \in \cT} du_{ij} \right) \; 
  \left( \prod_{ (i<j) }   w^\cT_{ij} \right)^{l_{ij}} \right]  \left(-\frac{g}{3} \right)^{n-1+\sum_l l_i + \sum_{i<j}l_{ij}} \;.
\end{split}
\ee
The additional derivatives with respect to $\sigma$ generate loop edges decorating the tree $\cT$: $l_{ij}$ is the multiplicity of the loop edge between $i$ and $j$ and $l_i$ the number of tadpole edges (or self loops) at the vertex $i$. We denote the graph consisting in $\cT$ decorated by such extra loop edges by $G$. Of course, $\cT$ is a spanning tree of $G$ which we denote $G\supset \cT$.
From Lemma~\ref{lem:Heppsect}, the integral over $u$ in Eq.~\eqref{eq:Wintermed} is: 
\be
w(G,\cT) = \underbrace{\int_0^1  \dots \int_0^1}_{|\cT| \; {\rm times}} \left( \prod_{ (i,j) \in \cT}du_{ij} \right) \; 
 \left[  \prod_{ (k,l)\notin \cT}   w^\cT_{kl} (u_{\cT}) \right] =\frac{N(G,\cT)}{|E(G)|!} \;, 
\ee 
yields the percentage of Hepp sectors in which $\cT$ is the leading spanning tree of $G$. Collecting the coupling constants and the symmetry factor in the amplitude $A(G)$ of the graph $G$, we write:
\be
W (g) =
\sum_{n\ge 1} \frac{1}{n!} \sum_{ \cT \in \cT_n } \sum_{G\supset \cT} w(G,\cT) A(G) \; ,
\ee
that is for each tree $\cT$ we resum the amplitudes of the graphs in which $\cT$ is a spanning tree with a weight given by the percentage of Hepp sectors of $G$ in which $\cT$ is the leading tree. 
Denoting $G_n$ the set of connected graphs over $n$ vertices, we commute the sums over trees $\cT$ and graphs $G$ and using $\sum_{T\subset G}w(G,\cT) = 1$ we get:
\be
 W (g) =\sum_{n\ge 1} \frac{1}{n!} \sum_{G \in G_n} A(G) \;,
\ee 
which is the familiar perturbative expansion of the free energy in connected graphs.
\end{proof}

	\bibliographystyle{JHEP}
	\bibliography{Refs-RESURGENCE} 
	
	\addcontentsline{toc}{section}{References}
	
\end{document}